\documentclass[12pt]{article}
 
\usepackage[margin=1in]{geometry} 

\usepackage{ulem,amsmath,amsthm,amssymb}
\usepackage{graphicx}
\usepackage{enumitem}
\usepackage{setspace}
\usepackage[authoryear]{natbib}
\usepackage{titling}
\usepackage[colorlinks=true,urlcolor= magenta,citecolor=blue,hypertexnames=false]{hyperref}
\usepackage{indentfirst}
\usepackage{mathtools}
\usepackage{nth}
\usepackage[english]{babel}
\usepackage{supertabular, longtable, caption, afterpage, tabularx, threeparttable}
\usepackage{booktabs} 
\usepackage{multirow}
\usepackage{array} 
\newcolumntype{P}[1]{>{\centering\arraybackslash}p{#1}}
\usepackage{ragged2e}
\usepackage{dcolumn}
\newcolumntype{d}[1]{D..{#1}} 
\usepackage{siunitx}
\usepackage{lscape}
\counterwithin{figure}{section}
\counterwithin{table}{section}
\usepackage{xcolor}
\usepackage{pgfplots}

\newcommand{\E}{\mathbb{E}_g}

\DeclareMathOperator*{\plim}{plim}

\newtheorem{theorem}{Theorem}
\newtheorem{corollary}{Corollary}
\newtheorem{proposition}{Proposition}

\newtheorem{assumption}{Assumption}

\makeatletter
\renewcommand{\p@subsection}{}
\renewcommand{\p@subsubsection}{}
\makeatother

\usepackage{xcolor}
\pgfplotsset{compat=1.18}

\begin{document}

\title{Identification of Heterogeneous Peer Effects\thanks{We are grateful to Yann Bramoullé, Vincent Boucher, St\'{e}phane Bonhomme, and Mathieu Lambotte for their helpful comments and suggestions. We also thank participants at the Econometrics Seminars at Washington University in St. Louis, New York University, University of Exeter, and Concordia University, as well as participants at the Africa Meeting of the Econometric Society (AFES, 2026) and World Congress of the Econometric Society (2025). Our package is available at: \url{https://github.com/ahoundetoungan/QuantilePeer}.}
}
\author{Eyo I. Herstad\thanks{Erasmus University Rotterdam and Tinbergen Institute. email: \href{mailto:herstad@ese.eur.nl}{herstad@ese.eur.nl}} \and Aristide Houndetoungan\thanks{D\'{e}partement d'\'{e}conomique, Universit\'{e} Laval. email: \href{mailto:ahoundetoungan@ecn.ulaval.ca}{ahoundetoungan@ecn.ulaval.ca}} \and Myungkou Shin\thanks{School of Social Sciences, University of Surrey. email: \href{mailto:m.shin@surrey.ac.uk}{m.shin@surrey.ac.uk}}}
\date{November 16, 2026 \\[0.8em]
}




\maketitle
\doublespacing
\begin{abstract}
    We develop a peer effects model in which each peer’s influence depends on their outcome rank within the peer group. Unlike existing models, it allows non-monotonic effects across ranks. The model admits a unique equilibrium, and its parameters can be identified using peers’ exogenous characteristics. We introduce a parsimonious quantile-dependent specification and develop a specification test to guide model selection. Applying the model to student outcomes in Add Health reveals heterogeneous, often non-monotonic spillovers. Counterfactual analysis shows that allowing for non-monotonic peer effects changes the rankings of influential students relative to existing models. Our R package, QuantilePeer, implements these methods.
\end{abstract}

\clearpage
\doublespacing


\section{Introduction}\label{sec:intro}
Economic agents are often connected. These connections raise important questions about how linked units influence one another through spillovers and conformity, known as peer effects. For example, we may expect a student to perform better academically if they have a friend with high grades, or to get in trouble at school if the majority of their friends do. These types of spillovers are inherently heterogeneous. The effect of each peer's outcome on the student's outcome depends on the outcomes of all the peers of the student. However, the existing models of peer effects either assume homogeneous effects from peers through the linear-in-means (LiM) model \citep{manski1993identification} or impose strong parametric restrictions on the underlying peer effects \citep{boucher2024toward}. These models, when misspecified, may lead to inaccurate counterfactual analyses. This greatly limits the scope for peer effect estimation in empirical settings, where it is well known that peer effects are heterogeneous \citep{sacerdote2011peer,masten2018random,bramoulle2020peer}. 

Our paper fills this gap between theory and empirical reality by introducing a flexible framework for the identification and estimation of heterogeneous peer effects. Our framework, in its most general form, allows for each peer's outcome to have a differential effect on the individual, determined by their rank in the peer group of the said individual, in terms of the outcome. The framework, therefore, allows for a type of endogenous heterogeneity where the effect of one individual on another depends on the full distribution of outcomes among each individual's peers. The framework covers several special cases of endogenous peer effects that are of empirical interest, such as quantile peer effects and extreme order statistics. 

In its most general form, our model can be motivated as a regression of an individual's outcome on the distribution of their peers' outcomes. This is because ordered peer outcomes have a one-to-one relationship with the (empirical) distribution of peer outcomes. As such, our model allows for a wide range of peer effect patterns that are functions of the empirical distribution of peer outcomes. 

The generality of our model captures many economically relevant features of peer effects. Firstly, our model allows for substitutability and complementarity among peers. Perfectly substitutable peers, i.e., linear-in-means, and perfectly complementary peers, i.e., linear-in-mins, are nested as special cases of our model.\footnote{By a linear-in-mins model we mean a model where the outcome is a linear function of the minimum outcome of an individual's peers. This model has previously been discussed in \citet{tao2014social}.} Secondly, our model allows for returns to scale in terms of the number of peers, as the total peer effect can be increasing in the number of peers.

Importantly, our model allows for a non-monotonic marginal effect of peer outcomes. This means that we do not impose that peers with higher outcomes necessarily have a greater impact on an individual's outcome than peers with lower outcomes, or vice versa. This is a particularly distinct feature of our model, which is not shared by any existing models for heterogeneous peer effects. For example, a recent and influential paper, \citet{boucher2024toward}, models peer effects as linear in either convex or concave functions of peer outcomes. This paper recommends using the Constant Elasticity of Substitution (CES) aggregator to measure the effect of peer outcomes. Since the CES function is either convex or concave, the marginal effect of peer outcomes in this model is necessarily monotonic.\footnote{This model also requires all outcomes to be strictly positive, which our model does not require.} This is not true in our model, as can be seen from a special case of it where an individual is only affected by the peer whose outcome is the median among said individual's peers. In this case, the marginal effect of a peer's outcome is zero for peers who are not the median, meaning the extreme outcomes have no impact on the peer effect, while the median has a nonzero effect. These types of non-monotonic peer effects have been heavily discussed in theories on social interactions, such as the theory of social comparison where peer influence is driven by the outcome of a representative peer \citep{festinger1954theory}. In this framework, the outcome of this representative peer serves as a focal point for the group, and high-outcome peers may cause negative spillovers due to discouragement effects \citep{rogers2016discouraged}.

This flexibility in our model does not only allow us to detect more varied patterns of peer effects, but also has direct implications for policy. Often, researchers want to estimate spillovers to identify key players in the network \citep{ballester2006s}, to determine optimal treatment assignment \citep{galeotti2020targeting}, or to design optimal network allocations to improve outcomes in a network \citep{carrell2013natural}. In all of these cases, having a misspecified model for spillovers can lead to suboptimal or even damaging policies. For example, if a researcher uses a linear-in-means model when the true model is linear-in-median, they may mistakenly target treatment to the wrong individuals, when in reality these individuals have little influence on the median and therefore the true spillover. 

In the general version of our model, the dimension of the model parameter grows quadratically in the maximum number of peers, making interpretation and estimation challenging. We therefore show how to construct lower-dimensional restricted versions of the model, with a focus on interpolated quantiles. The quantile model allows us to capture peer effect heterogeneity by allowing peers to exert distinct influences depending on their location in the peer outcome distribution. As the researcher needs to select the number of quantiles to include, we develop a specification test that allows practitioners to decide between two potentially misspecified models that do not nest each other. We then provide guidance for researchers in how to apply this test to select model specifications in practice. Our simulations indicate that this model selection method performs well and has good power properties in realistic settings. In our empirical application, we apply the quantile version of our model and use the test to select the number of quantiles to use. 

The primary contribution of this paper is to introduce and show the identification of a general model of heterogeneous peer effects. This model allows for any type of spillovers that can be represented through the empirical peer outcome distribution. Our identification result makes use of the conditional distribution of the outcomes given the network connections and the individual characteristics, with the network assumed to be exogenous. The steps of the identification argument are as follows. First, we show that the general model admits a unique equilibrium when the total incoming spillover from all peers is bounded, in absolute value, by one for each individual. As in most nonlinear peer effects models, our identification conditions rely on a high-level rank restriction \citep{brock2007identification, yang2017social, boucher2024toward}. We motivate this restriction with a network structure where individuals have friends-of-friends who are not their direct friends. This condition is the same as that used to obtain identification in the LiM model \citep{BDF}. Under such a network structure, we show that the distribution of the exogenous characteristics of these friends-of-friends can be used as instruments. We then construct moment conditions to identify all the parameters of the model. This step relies on an instrument relevance condition closely related to the equivalent condition of the LiM model. We also provide a specification test of the model that allows practitioners to choose a model specification in a data-driven way.


Our second contribution is to use our model to show that non-monotonic peer effects are highly prevalent in real-world network data. We do this by estimating a three-quantile version of our model on eleven different outcomes in the National Longitudinal Study of Adolescent to Adult Health (Add Health) dataset. We then apply our estimator and find a wide range of different spillovers. We find non-monotonic marginal effects of peers for most of our outcomes, with the middle of the peer outcome distribution being the most important, consistent with psychological models of peer effects \citep{festinger1954theory,rogers2016discouraged}. For example, we find that for smoking behavior, it is the middle of the distribution that matters. The minimum peer outcome has no statistically significant effect, while the maximum has a statistically significant effect but it is only 20\% of the magnitude of the effect of the median peer. Similarly, for academic achievement (GPA), what matters is the median, while the minimum and the maximum have statistically insignificant effects. Existing methods such as LiM and the CES approach of \citet{boucher2024toward} would not be able to capture these spillover patterns as they are non-monotonic.

Beyond the parameter estimates, we show that the choice of peer effect model has strong effects on estimates of policy objects of interest, such as student influence. Using the empirical results, we identify key players as the individuals whose removal of ties in the network yields the largest changes in the outcome distribution. Identifying such nodes is important, for example, for a social planner who is implementing targeted policies \citep{ballester2006s, lee2021key}. The results indicate that node rankings based on influence differ between the proposed model and the LiM and CES specifications, especially for outcomes that are characterized by non-monotonic peer effects, which these specifications fail to capture. Analyzing such outcomes using the LiM or CES models can distort the identification of key players, thereby reducing the effectiveness of targeted policy interventions.

Peer effects have been studied extensively in economics \citep{bramoulle2020peer,de2020econometric}, with seminal contributions by \citet{manski1993identification}, \citet{sacerdote2001peer} and \citet{BDF}. A strand of this literature has focused on estimating peer effects as nonlinear functions of peer outcomes, such as extreme values \citep{tao2014social,tatsi2015endogenous} and the proportion of peers above a cutoff \citep{Kang,GKN}.\footnote{Another strand of this literature has estimated heterogeneous peer effects based on latent traits or observable characteristics, such as race or gender \citep{boucher2022peer,comola2025heterogeneous}. These models, while conceptually related, do not directly map into our framework. } However, these contributions remain sporadic and disconnected, lacking both a unifying framework and a formal theoretical treatment of the identification and estimation challenges involved. This paper provides a rigorous treatment of peer effect heterogeneity that prior work has postulated informally, unifying the intuitions underlying the empirical practice in a tractable framework.

The rest of this paper is structured as follows. Section \ref{sec:model} discusses the model and its implications, and Section \ref{sec:ident} presents our formal identification results, specification tests, and comparisons to other models. Section \ref{sec:sim} presents and discusses our simulation results, and Section \ref{sec:emp} applies our methodology to the Add Health dataset. 

\section{Model}\label{sec:model}
We consider a many-network setup with $m$ networks in total, indexed by $g=1,\dots,m$. In the $g$-th network, there are $n_g$ units, indexed by $i=1,\dots,n_g$. We assume independence across networks, but not across units within a network. The network structure of the $g$-th network is represented by an adjacency matrix $\mathbf{A}_g = [a_{i,j,g}]$. We define $a_{i,j,g} = 1$ if the $j$-th unit in the $g$-th network is a friend (or peer) of the $i$-th unit in the $g$-th network, and $a_{i,j,g} = 0$ otherwise.\footnote{In weighted networks, $a_{i,j,g}$ can take nonnegative values (not necessarily binary) to reflect the intensity of the outgoing link from $i$ to $j$. The results derived in this paper can also be extended to such networks.} The network may be directed; i.e., $a_{i,j,g} = 1$ does not necessarily imply $a_{j,i,g} = 1$. We also assume there are no self-links, meaning $a_{i,i,g} = 0$ for every $i$ and $g$. The network adjacency matrices $\{\mathbf{A}_g\}_{g=1}^m$ are observed. In addition to the network adjacency matrices, we observe unit-level outcome $y_{i,g} \in \mathbb{R}$ and unit-level control covariates $\boldsymbol{x}_{i,g} \in \mathbb{R}^l$. In sum, the network dataset given in the econometric framework of this paper is $$
\left\lbrace \mathbf{A}_g, \{ y_{i,g}, \boldsymbol{x}_{i,g} \}_{i=1}^{n_g} \right\rbrace_{g=1}^m.
$$
The theoretical results of Section \ref{sec:model} are discussed with regard to a single network. Therefore, for notational brevity, we will suppress the network subscript $g$ for the rest of the section.

In our model, the outcome $y_i$ of unit $i$ depends on the outcomes of all their peers, i.e., all $j$ such that $a_{i,j} = 1$. We characterize this dependence in a highly flexible way, with every peer having a separate effect on unit $i$ depending on the ranking of their outcome among unit $i$'s friends. Specifically, let $d_i = \sum_{j=1}^n a_{i,j}$ denote the number of peers of unit $i$, i.e., $d_i$ is the row-wise sum of $\mathbf{A}$. Let also $\bar{d} = \max_i d_i$ be the maximum number of peers in the data. Then, for each individual $i$ and $k=1,\dots,d_i$, let $\tilde{y}_{i,k}$ denote the outcome of the $k$-th lowest-performing peer of unit $i$. That is, $\tilde{y}_{i,k}$ is the $k$-th lowest value in the set $\{y_j : a_{i,j} = 1\}$.\footnote{In case of ties, any peer may be selected as the lowest without loss of generality.} Our model for $y_i$ can then be written as:
\begin{align}
    y_i = \sum_{k=1}^{d_i} \beta_{k, d_i} \tilde{y}_{i,k} + \boldsymbol{x}_i^\intercal {\gamma} + \boldsymbol{\bar{x}}_i^\intercal {\delta}+  \varepsilon_i,  \label{eq:peer_effect_model}
\end{align}
for $i = 1, \ldots, n$ where $\boldsymbol{\bar{x}}_i$ is the average of peers' control covariates\footnote{We include this average to control for \textit{contextual effects} \citep{calvo2009peer}, as it is the most common specification in the literature. However, the framework can easily be extended to have a similar structure for peers' covariates as we have for peer outcomes. Furthermore, our identification results do not impose restrictions on $\delta$, allowing for $\delta = 0$.} and $\varepsilon_i$ is an error term. The peer effect coefficient $\beta_{k, d_i}$ captures the effect of the peer with the $k$-th lowest outcome among the $d_i$ peers.\footnote{These effects depend on $d_i$ because, as the number of peers increases, the influence of each individual peer may decrease.} For example, if unit $i$ has three peers, $\beta_{2,3}$ corresponds to the effect of the peer at the median of the peer outcome distribution, whereas $\beta_{1,3}$ and $\beta_{3,3}$ correspond to the effects of the peers with the lowest and highest outcomes, respectively.

The flexibility of our peer effect model in \eqref{eq:peer_effect_model} enables us to address key questions in the peer effect literature concerning the complementarity and substitutability of peers’ contributions to $y_i$ \citep[see, for example,][]{Kang}. When only the lowest-ranked peer matters (i.e., $\beta_{1,d} \neq 0$ and $\beta_{k,d} = 0$ for all $k \geq 2$), the peer effect reflects perfect complementarity, in the sense that an individual’s outcome depends entirely on the weakest performer in the group. Conversely, when all coefficients are equal across ranks (i.e., $\beta_{k,d} = \beta_{k',d}$ for all $k,k'$), meaning that every peer outcome contributes equally, peers are perfect substitutes. 

The model of this paper nests many widely used specifications in the literature. For example, the LiM model is a special case of our model, by letting $\beta_{k, d} = \dfrac{\beta}{d}$ for all $k$, where $\beta$ is the single peer effect parameter in the LiM specification, capturing the effect of the peers’ average outcome. The linear-in-sums (LiS) model is a special case of our model as well, by letting $\beta_{k, d} = \beta$ for all $k$, where $\beta$ is the single peer effect parameter in the LiS specification. In addition, our model nests models where the peer effect is decided by the minimal and maximal peer outcome, by letting $\beta_{k,d} = 0$ for every $k > 1$ or $k < d$ respectively \citep{tao2014social, tatsi2015endogenous}.

Our model allows for both monotonic and non-monotonic marginal peer effects, as can be seen through a simple example, the linear-in-median model.\footnote{To map the linear-in-median model to our framework, we set $\beta_{k,d} = \beta^{\text{med}}$ whenever $d$ is odd and $k =\frac{d+1}{2}$, and $\beta_{k,d} = \frac{\beta^{\text{med}}}{2}$ when $d$ is even and $k \in \{\frac{d}{2}, \frac{d}{2}+1\}$, and set $\beta_{k,d} = 0$ otherwise.} In this model, a peer's outcome has a marginal peer effect equal to zero when it is so high or low that the peer is not the median, but is nonzero when it is the median outcome among $i$'s peers. These types of non-monotonic marginal effects of peer outcomes are natural in many settings with spillovers as they align with theories of social comparison, which suggest that aligning with moderately performing peers, rather than with outliers, can be socially or psychologically optimal \citep{festinger1954theory,rogers2016discouraged,diel2019inspired}. 

Imposing monotonicity when peer effects are non-monotonic can have important policy implications. Since peer effects generate social multipliers, they are often leveraged to increase the effectiveness of interventions. For example, optimally assigning treatment within a network can substantially increase policy impact \citep{galeotti2020targeting}. Likewise, identifying key players can facilitate the efficient diffusion of information or behaviors throughout the network \citep{ballester2006s}. The optimal targeting strategy, however, depends critically on the underlying structure of peer effects. In the example above, where peers at the tails of the outcome distribution exert less influence, the most effective individuals to target are likely to be those whose outcomes lie closer to the center of the distribution. In contrast, existing specifications that assume monotonic peer effects instead predict that peers at the tails of the outcome distribution are the most influential and therefore the optimal targets for intervention. 

Similarly, researchers have used peer effects to motivate interventions in the network directly. A cautionary tale of such interventions is \citet{carrell2013natural}, who estimated a linear-in-means model using data on air force cadets to create optimal peer groups. However, they then found negative effects for individuals assigned to these optimal groups. The authors attribute this to endogenous network formation. Our model suggests it may also be due to the underlying peer effect model being misspecified as this would make the treatment groups non-optimal, even if the network formation behavior was not affected by treatment. 

In our empirical application, we find non-monotonic peer effect patterns in over half of the student outcomes we study. We then conduct a counterfactual analysis to identify key players within the network and find that the CES and LiM specifications may systematically misidentify them. This shows that our heterogeneous peer effect specification is not simply a more granular representation of aggregate peer effects. It instead recovers features of social interactions that are ruled out by existing models, and therefore leads to qualitatively different policy prescriptions. Accounting for this heterogeneity is essential both for accurately characterizing peer influence and, more importantly, for supporting the design of effective network-based interventions.

\subsection{Microfoundations}\label{sec:game}
In this subsection, we present microfoundations for our peer effect model introduced above using a complete-information game. In this game, each individual $i$ selects a strategy $y_i$ (e.g., academic effort) and derives utility from the distribution of their peers' strategies. Preferences are represented by a standard linear-quadratic utility function, which is commonly used in the context of the LiM model \citep[see][]{ballester2006s, calvo2009peer}, in which we allow for each peer to exert a specific influence. The utility function is given by:
\begin{equation}\label{eq:utility}
    U_i(y_i, ~ \mathbf y_{-i}) = \underbrace{\alpha_i \cdot y_i - \frac{1}{2} {y_i}^2}_{\text{private benefit}} + \underbrace{\sum_{k = 1}^{d_i}\beta_{k,d_i} \tilde y_{i,k} \cdot y_i }_{\text{social benefit}} ,
\end{equation}
\noindent where $\alpha_i$ is individual $i$'s type observable to all players and $\mathbf y_{-i} = (y_1, \dots, y_{i-1}, y_{i+1}, \dots, ~y_n)^\intercal$. The type $\alpha_i$ is generally specified as a linear function of observable characteristics $\boldsymbol{x}_i$ and an error term  \cite[see, e.g., ][]{blume2015linear}.

The utility function is additively separable into a private benefit, which depends only on individual $i$'s own effort $y_i$, and a social benefit, which is a function of both $y_i$ and the effort vector of other individuals, $\mathbf{y}_{-i}$. The private benefit is linear in the individual’s type $\alpha_i$, and the private cost (or disutility) of exerting effort $y_i$, $\dfrac{1}{2} {y_i}^2$, is quadratic in $y_i$. The social benefit captures spillover effects on an individual's effort through the peer effect $\sum_{k=1}^{d_i} \beta_{k,d_i} \tilde{y}_{i,k}$, a weighted sum of ranked peer efforts with the peer effect coefficients $\{\beta_{k,d_i}\}_{k=1}^{d_i}$ as weights. The parameters $\{\beta_{k,d_i}\}_{k=1}^{d_i}$ govern the degree of substitutability among peer efforts in the peer effect construction as discussed above, and also the direction of spillovers onto individual $i$'s own effort. Specifically, when individual $j$'s effort is ranked $k$-th among $i$'s peers, the sign of $\beta_{k,d_i}$ determines the local spillover: $\beta_{k,d_i}>0$ indicates that individual $j$'s effort raises individual $i$'s incentive to exert effort, while a negative weight indicates that it discourages effort. Crucially, since the weight applied to individual $j$'s effort depends on its rank, the spillover from any given peer is not fixed but varies with the efforts of all other peers.

The utility function shown in Equation \eqref{eq:utility} assumes that social benefits take the form of \textit{spillovers}. We can further extend this by adding \textit{conformity} to the social benefit as well. While the extension yields a similar reduced-form equation, the implication for counterfactual analysis is not the same \citep[see][]{boucher2016some}. In Supplemental Appendix~\ref{app:conformity}, we discuss the extension where both spillover and conformity coexist. As in \cite{boucher2024toward}, we require the presence of isolated individuals (individuals without friends) to estimate the shares of spillover and conformity effects.

Individual $i$ chooses their effort $y_i$ by maximizing the utility function $U_i(y_i, \mathbf y_{-i})$. By solving the first-order condition of this maximization problem, we obtain the best response function:
\begin{equation} \label{BRF}
    BR_i(\mathbf y_{-i}) = \sum_{k = 1}^{d_i} \beta_{k,d_i} \tilde y_{i,k} + \alpha_i.
\end{equation}
By defining $\alpha_i = \boldsymbol{x}_i^{\intercal}\gamma + \boldsymbol{\bar{x}_i}^\intercal \delta + \varepsilon_i$, this best response function corresponds to our peer effect specification in Equation \eqref{eq:peer_effect_model}. A strategy profile $\mathbf y = (y_1, ~\dots, ~y_n)^{\intercal}$ is a Nash equilibrium if $y_i = BR_i(\mathbf y_{-i})$ for all $i$. To establish the existence and the uniqueness of the Nash equilibrium, we impose the following condition.

\begin{assumption}\label{ass:equilibrium}
    The peer effect parameters satisfy the condition $\bar{\beta} := \displaystyle\max_{d \leq \bar d}\sum_{k = 1}^{d} \lvert \beta_{k,d} \rvert < 1$.
\end{assumption}

\noindent Assumption~\ref{ass:equilibrium} ensures that a one-unit increase in peer outcomes does not lead to an increase in individual $i$'s effort greater than one, regardless of their number of friends $d_i$. This requirement is standard in peer effect models to establish the existence and uniqueness of the Nash equilibrium. In the LiM model, an equivalent condition imposes that the absolute value of the (single) peer effect parameter is strictly less than one \citep[see][]{BDF}.

\begin{proposition}\label{prop:uniqueNE} Assume that the preferences of each individual are characterized by the utility function \eqref{eq:utility}. Under Assumption \ref{ass:equilibrium}, there exists a unique Nash equilibrium $\mathbf y^{\ast} = (y_1^{\ast}, ~\dots,   y_n^{\ast})^{\intercal}$ such that $y_i^{\ast} = BR_i(\mathbf{y}^{\ast}_{-i})$ for every $i$.
\end{proposition}
\begin{proof}
    See Appendix \ref{app:uniqueNE}.
\end{proof}
\noindent The key part of the proof of Proposition \ref{prop:uniqueNE}  is to show that the mapping \[
BR(\mathbf{y})= (BR_1(\mathbf y_{-1}),\dots,~BR_n(\mathbf y_{-n}))^{\intercal}
\]
is a contraction. That is, for any two strategy profiles $\mathbf y, \mathbf y^{\prime}\in\mathbb R^n$, $\| BR(\mathbf y) - BR(\mathbf y^{\prime})\|_\infty \leq \bar\beta \lVert \mathbf y - \mathbf y^{\prime}\rVert_{\infty}$, where $\bar\beta < 1$ is defined in Assumption \ref{ass:equilibrium}.\footnote{For any $\mathbf u = (u_1, u_2, \dots, u_n)\in\mathbb R^n$, the infinity norm of $\mathbf u$ is given by $\lVert \mathbf u \rVert_{\infty} = \displaystyle\max_{1 \leq i\leq n} \lvert u_i\rvert$.} Under this condition, $BR$ has a unique fixed point, which corresponds to the Nash equilibrium. 

\subsection{Empirical Specification}\label{sec:implementation:restriction}

The number of peer effect parameters in the fully saturated model is $\frac{\bar{d}(\bar{d}+1)}{2}$, which is challenging for estimation in most real-world applications where the maximum number of peers, $\bar{d}$, is large. This is the case in our empirical application where $\bar{d} = 10$, which means we would have 55 endogenous variables if we wanted to estimate the fully saturated model. As this would make both interpretation and estimation challenging, we impose restrictions on the parameter space to reduce the number of parameters. While there are many ways to do this,\footnote{For example, there are nine different ways of calculating quantiles in the literature \citep{hyndman1996sample}, all of which can be implemented using our results. The model also allows for non-quantile restrictions, such as imposing different linear slopes for different sets of peers. Since our model also nests the linear-in-means model, one could use such specifications to construct tests of the linear-in-means assumption. } we will focus on equally spaced, linearly interpolated quantiles as they are highly interpretable, easy to implement, and used extensively in economics. 

The quantile version of our model can be written as:
\begin{equation}
    y_i = \sum_{k=1}^{M} \beta^k \, \tilde{y}_{i}^k + {\boldsymbol{x}_i}^{\intercal} \gamma + {\boldsymbol{\bar{x}}_i}^\intercal \delta +\varepsilon_i, 
    \label{eq:model:q}
\end{equation}
where $\tilde{y}_{i}^k$ is the $(k-1)/(M-1)$ quantile among the peer outcomes for individual $i$ for $k=1,\dots,M$.\footnote{When $M=1$, $\tilde{y}_i^1 = \tilde{y}_{i,1}$.} For example, when there are three peers, i.e. $d_i=3$, and three quantile variables, i.e. $M=3$, we have $\tilde{y}_i^1 = \tilde{y}_{i,1}, \tilde{y}_i^2 = \tilde{y}_{i,2}$ and $\tilde{y}_i^3 = \tilde{y}_{i,3}$, corresponding to the minimum, median, and maximum peer outcomes. When $d_i \neq M$, we linearly interpolate.\footnote{For example, to compute the 0.3 quantile among peer outcomes for individual $i$ such that $d_i=5$, first we find that $\tilde{y}_{i,2}$ is the 0.25 quantile and $\tilde{y}_{i,3}$ is the 0.5 quantile. Then, the 0.3 quantile can be computed as $ \left( 0.8 \cdot \tilde{y}_{i,2} + 0.2 \cdot \tilde{y}_{i,3} \right)$, since 0.3 lies 20\% of the way from 0.25 to 0.5.} Supplemental Appendix \ref{app:samp_quant} discusses the specifics of the linear interpolation in constructing $\{\tilde{y}_i^k\}_{k=1}^M$ in a generalized manner. 

There are several merits to the quantile model. Firstly, the quantile-based model is nested in the fully saturated version of our peer effect model \eqref{eq:peer_effect_model}, since $\{\tilde{y}_i^k\}_{k=1}^M$ are defined to be linear combinations of $\{\tilde{y}_{i,k}\}_{k=1}^{d_i}$. Therefore, the microfoundation from Section \ref{sec:game} and the identification results that will be discussed in Section \ref{sec:ident:ident} naturally extend to the quantile-based model as well. Secondly, quantiles capture important features of the ordered peer outcomes $\{\tilde{y}_{i,k}\}_{k=1}^{d_i}$: the median captures centrality, while the minimum and maximum capture tail behavior, etc. Thirdly, with equally spaced quantile levels, the peer outcomes that are not included in the restricted model lie between the chosen quantiles, making the information loss minimal. In our empirical application, we show that the quantile specification can reveal many interesting patterns of peer effects.

The choice of whether to use a quantile model or a different specification is inherently context-dependent. Even within the quantile specification, it is still necessary to pick the number and the locations of the included quantiles. In Section \ref{sec:econometrics:test} we therefore develop a specification test to help practitioners make these decisions. 

\section{Identification and Inference}\label{sec:ident}
In this section, we discuss the identification of the structural parameter $\big( \{\{\beta_{k,d}\}_{k=1}^d\}_{d=1}^{\bar{d}}, \break\gamma, \delta \big)$ in terms of network-level moments, under some normalization due to uneven network sizes $n_g$. As is customary in the literature, we assume independence across networks and assume that the network adjacency matrix and the network-level collection of individual-level control covariates are exogenous. To provide empirical guidance on the tractable specification in Section \ref{sec:implementation:restriction}, we develop an encompassing test. 

\subsection{Identification}\label{sec:ident:ident}
To simplify notation, we use $\E$ to denote the conditional expectation of a random variable given exogenous variables from the same network throughout the paper. For an arbitrary random variable $w_g$ in network $g$, we write
$$
\E \left[ w_g \right] := \mathbb{E} \left[ w_g | \mathbf{X}_g, \mathbf{A}_g \right]
$$
where $\mathbf{X}_g = \begin{pmatrix} \boldsymbol{x}_{1,g} & \cdots &  \boldsymbol{x}_{n_g,g} \end{pmatrix}^\intercal$. Note that $\E [w_g] = 0$ also implies $$
\mathbb{E}[w_g |  \{\boldsymbol{x}_{i,g}, \boldsymbol{\bar{x}}_{i,g}, a_{i,j,g} \}_{1 \leq i,j \leq n_g} ]=0
$$
as $\boldsymbol{\bar{x}}_{i,g}$ is a deterministic function of $\boldsymbol{X}_{g}$ and $\mathbf{A}_g$. Also, we vectorize the peer effect coefficients and the ranked peer outcomes
\begin{align*}
    \beta &= \begin{pmatrix} {\beta_{1,1}} & {\beta_{1,2}} & \beta_{2,2} & \cdots & \beta_{1, \bar{d}} & \cdots & {\beta_{\bar{d}, \bar{d}}} \end{pmatrix}^\intercal \\
    \tilde{\boldsymbol{y}}_{i,g} &= \begin{pmatrix} 0 & \cdots & 0 &
    \tilde{y}_{i,1,g} & \cdots & \tilde{y}_{i,d_i,g}  & 0 & \cdots & 0 \end{pmatrix}^\intercal
\end{align*}
and let $\boldsymbol{w}_{i,g} = \big( {\tilde{\boldsymbol{y}}_{i,g}}^\intercal, {\boldsymbol{x}_{i,g}}^\intercal ,{\boldsymbol{\bar{x}}_{i,g}}^\intercal\big)^\intercal \in \mathbb{R}^{\frac{\bar{d}(\bar{d}+1)}{2}+2l}$. We also denote by $\theta = (\beta^{\intercal}, \gamma^{\intercal}, \delta^{\intercal})^{\intercal}$ the vector of all parameters. The rank-dependent peer effect model \eqref{eq:peer_effect_model} can be rewritten as follows: \begin{align}
y_{i,g} &= {\tilde{\boldsymbol{y}}_{i,g}}^\intercal \beta + {\boldsymbol{x}_{i,g}}^\intercal \gamma + {\bar{\boldsymbol{x}}_{i,g}}^\intercal \delta + \varepsilon_{i,g} = {\boldsymbol{w}_{i,g}}^\intercal \theta + \varepsilon_{i,g} \label{eq:id_peer_effect_model} 
\end{align}
for $i=1,\hdots,n_g$ and in matrix form: $$
\mathbf{y}_g = \tilde{\mathbf{Y}}_g \beta + \mathbf{X}_g \gamma + \bar{\mathbf{X}}_g \delta + \boldsymbol{\varepsilon}_g = \mathbf{W}_g \begin{pmatrix} \beta^\intercal & \gamma^\intercal & \delta^\intercal \end{pmatrix}^\intercal + \boldsymbol{\varepsilon}_g
$$ 
with $\tilde{\mathbf{Y}}_g = \begin{pmatrix} \tilde{\boldsymbol{y}}_{1,g} & \cdots & \tilde{\boldsymbol{y}}_{n_g,g} \end{pmatrix}^\intercal$, $\bar{\mathbf{X}}_g = \begin{pmatrix} {\bar{\boldsymbol{x}}_{1,g}} & \cdots &  {\bar{\boldsymbol{x}}_{n_g,g}} \end{pmatrix}^\intercal$, $\mathbf{W}_g = \begin{pmatrix} \boldsymbol{w}_{1,g} & \cdots & \boldsymbol{w}_{n_g,g} \end{pmatrix}^\intercal$ and $\boldsymbol{\varepsilon}_g = \begin{pmatrix} \varepsilon_{1,g} & \cdots & \varepsilon_{n_g,g} \end{pmatrix}^\intercal$.

\begin{assumption} \label{ass:exogeneity} \textup{\textsc{(Exogenous Network and Covariates)}} For each $g=1,\dots,m$ and $i=1,\ldots,n_g$
$$ 
\E[\varepsilon_{i,g}] =\mathbb{E}\left[\varepsilon_{i,g} |\mathbf{X}_g, \mathbf{A}_g  \right] =  0.
$$  
\end{assumption}
\noindent Assumption \ref{ass:exogeneity} rules out endogenous network formation and treats the network structure as fixed, meaning it does not depend on the randomness in $\{\varepsilon_{i,g}\}_{i=1}^{n_g}$. There have been several advances in analyzing peer effect models with endogenous network formation. A common approach has been to estimate the probabilities $\Pr \{a_{i,j,g} = 1 | \mathbf{X}_g \}$, and to use these to construct instruments \citep{konig2019r,lee2021key}. This approach can easily be incorporated into our model by constructing instruments using the $\Pr \{a_{i,j,g} = 1 | \mathbf{X}_g \}$-weighted distribution of $\boldsymbol{x}_{i,g}$ in the network. Alternatively, some papers have used a control function approach \citep{johnsson2021estimation}, which can similarly be implemented for our specification as it is linear in its parameters, just as the linear-in-means model. Of course, this is dependent on the networks being sufficiently large to estimate the required network formation models.\footnote{Other approaches, such as Bayesian approaches dependent on likelihoods \citep{hsieh2016social} or building moments based on the leave-one-out principle \citep{jochmans2023peer}, have also been suggested in the literature, but do not apply to our setting. } 

Compared to these approaches, our model instead focuses on endogeneity in peer effects conditional on the network. To see this, note that we can view our model as a second step of a two-stage process where, firstly, the network is formed, and secondly, the peer effect is determined based on the connections established in the first step. Models of endogenous network formation address endogeneity in the first step, whereas our model addresses endogeneity in the second step by allowing peer effect coefficients to depend on $\{\varepsilon_{i,g}\}_{i=1}^{n_g}$. 

This endogeneity means that, although $\E[\mathbf{y}_g]$ admits a linear representation, the quantities entering the representation are not known functions of $(\mathbf{X}_g, \mathbf{A}_g)$. This is in contrast to the LiM model, where $\E[\mathbf{y}_g]$ is linear in $\mathbf{X}_g, \mathbf{G}_g \mathbf{X}_g, {\mathbf{G}_g}^2 \mathbf{X}_g, \dots$ with $\mathbf{G}_g = \big[ a_{i,j,g} / d_{i,g} \big]$.\footnote{When $d_{i,g} = 0$, we follow convention and set $g_{i,j,g} = 0$ for every $j$.} This is due to the rank-dependence of the peer effect coefficients. Consider a simple network with three individuals where individual 1 is friends with both individuals 2 and 3, but individuals 2 and 3 are only friends with individual 1.\footnote{The network adjacency matrix is $$
\begin{pmatrix} 0 & 1 & 1 \\ 
1 & 0 & 0 \\
1 & 0 & 0 
\end{pmatrix}.
$$} From Assumption \ref{ass:exogeneity}, we obtain the following:
\begin{align*}
    \E \left[ y_{1,g} \right] &= \beta_{1,2} \E \left[ y_{2,g} \mathbf{1}\{y_{2,g} \leq y_{3,g}\} + y_{3,g} \mathbf{1}\{y_{2,g} > y_{3,g}\}  \right] \\
    & \hspace{5mm} + \beta_{2,2} \E \left[ y_{3,g} \mathbf{1}\{y_{2,g} \leq y_{3,g}\} + y_{2,g} \mathbf{1}\{y_{2,g} > y_{3,g}\}  \right] + {\boldsymbol{x}_{1,g}}^\intercal \gamma + {\bar{\boldsymbol{x}}_{1,g}}^\intercal \delta.
    \intertext{Substituting $\beta_{1,1} y_{1,g} + {\boldsymbol{x}_{2,g}}^\intercal \gamma + {\bar{\boldsymbol{x}}_{2,g}}^\intercal \delta + \varepsilon_{2,g}$ for $y_{2,g}$ and similarly for $y_{3,g}$, we get:}
    \tilde{\beta} \E \left[ y_{1,g} \right] &= \beta_{1,2} \E \left[ \varepsilon_{2,g} \mathbf{1}\{y_{2,g} \leq y_{3,g}\} + \varepsilon_{3,g} \mathbf{1}\{y_{2,g} > y_{3,g}\}  \right] \\
    & \hspace{5mm} + \beta_{2,2} \E \left[ \varepsilon_{3,g} \mathbf{1}\{y_{2,g} \leq y_{3,g}\} + \varepsilon_{2,g} \mathbf{1}\{y_{2,g} > y_{3,g}\}  \right] + \sum_{i=1}^3 \left(  {\boldsymbol{x}_{i,g}}^\intercal \tilde{\gamma}_{i,g} + {\bar{\boldsymbol{x}}_{i,g}}^\intercal \tilde{\delta}_{i,g} \right)
\end{align*}
for some $\tilde{\beta}$, $\{\tilde{\gamma}_{i,g}\}_{i=1}^3$ and $\{\tilde{\delta}_{i,g}\}_{i=1}^3$. Without further assumptions on the distribution of $\{ \varepsilon_{i,g} \}_{i=1}^{n_g}$, quantities such as $\E \left[ \varepsilon_{2,g} \mathbf{1}\{ y_{2,g} \leq y_{3,g}\}\right]$ cannot be written as a function of observables $(\mathbf{X}_g, \mathbf{A}_g)$. An exception occurs when $\beta_{1,2} = \beta_{2,2}$, in which case the conditional expectations cancel out, reverting to the LiM or LiS model. By the same argument, the quantities entering the reduced-form representation of $\E[\mathbf{W}_g]$ in our model depend on the distribution of $\varepsilon_{i,g}$ in a way that resists closed-form characterization without further distributional assumptions.\footnote{We describe and discuss the closed-form characterization of our model in Supplemental Appendix \ref{app:reduced_form}.} For this reason, Assumption~\ref{ass:id} is stated directly in terms of $\E[\mathbf{W}_g]$ rather than in terms of low-level conditions on the primitives.

The structural parameter $\theta$ of our model $\big( \beta, \gamma, \delta \big)$, or $\big( \{\beta^k\}_{k=1}^{M}, \gamma, \delta \big)$ in a quantile model, is identified under the following conditions. 
\begin{assumption}\label{ass:id} \text{ } 
\begin{enumerate}[label=\textbf{\alph*.}]
\item \label{ass:id:iid} $\boldsymbol{\varepsilon}_g \, | \, \big( n_g, \mathbf{X}_g, \mathbf{A}_g) \sim \mathcal{F}(n_g, \mathbf{X}_g, \mathbf{A}_g)$.
\item \label{ass:id:nonsingular} $\displaystyle \plim_{m \to \infty} \frac{1}{m} \sum_{g=1}^m \E[\mathbf{W}_g]^\intercal \E[\mathbf{W}_g]$ is nonsingular.
\end{enumerate}
\end{assumption}

\noindent Theorem~\ref{theorem:id} formally establishes the identification result. 

\begin{theorem}
\label{theorem:id} Suppose that Assumptions \ref{ass:equilibrium}\textendash\ref{ass:id} hold. Then, $\beta, \gamma$, and $\delta$ are identified from the moment condition below:
\begin{align}
\plim_{m \to \infty} \frac{1}{m} \sum_{g=1}^m \E[\mathbf{W}_g]^\intercal \E[\mathbf{W}_g]\theta &= \plim_{m \to \infty} \frac{1}{m} \sum_{g=1}^m \E[\mathbf{W}_g]^\intercal \E[ \mathbf{y}_g ] 
\end{align}
\end{theorem}
\begin{proof}
    See Appendix \ref{app:theoremidproof}.
\end{proof}

Part \ref{ass:id:nonsingular} of Assumption \ref{ass:id} provides a high-level condition for the identification of the rank-dependent peer effect model. Deriving a low-level sufficient condition for Part~\ref{ass:id:nonsingular} of Assumption~\ref{ass:id} is nontrivial due to the rank-dependence of the peer effect coefficients. However, we believe it will still hold in many empirically relevant cases. Note that this assumption only fails if $\mathbb E_g[\tilde{\boldsymbol{y}}_{i,g}]$, $\boldsymbol{x}_{i,g}$, and ${\bar{\boldsymbol{x}}_{i,g}}$ are linearly dependent. If $\boldsymbol{x}_{i,g}$ and ${\bar{\boldsymbol{x}}_{i,g}}$ are not linearly dependent, which can be easily verified, then this must mean that
\begin{equation}\label{eq:notfullrank}
    \mathbb E_g[\tilde{\boldsymbol{y}}_{i,g}]^\intercal \check \beta + \boldsymbol{x}_{i,g}^\intercal \check \gamma  + {\bar{\boldsymbol{x}}_{i,g}}^\intercal \check \delta = 0,
\end{equation}
for some parameters $\check \beta$, $\check \gamma$, and $\check \delta$ such that $\check \beta \ne 0$. To see when Equation \eqref{eq:notfullrank} does not hold for all $i$ and $g$, assume that $i$ has friends of friends, denoted by $j$, who are not direct friends of $i$. Changes in the characteristics $\boldsymbol{x}_{j,g}$ have no influence on $\boldsymbol{x}_{i,g}$ and ${\bar{\boldsymbol{x}}_{i,g}}$ because $j$ is not a direct friend of $i$. Equation \eqref{eq:notfullrank} implies that the total variation in $\mathbb E_g[\tilde{\boldsymbol{y}}_{i,g}]^\intercal \check \beta$ following changes in $\boldsymbol{x}_{j,g}$ is zero. This would not hold in many contexts because individuals whose outcomes determine $\tilde{\boldsymbol{y}}_{i,g}$ are friends of $i$, and some of these friends have $j$ as a direct friend. Therefore, if contextual effects matter, changes in $\boldsymbol{x}_{j,g}$ are expected to influence the outcomes of $i$'s friends, except for special cases where the influence on each element of the vector $\tilde{\boldsymbol{y}}_{i,g}$ cancels out in $\mathbb E_g[\tilde{\boldsymbol{y}}_{i,g}]^\intercal \check \beta$. This identification argument is also employed by \cite{houndetoungan2026count} and extends the intuition of \cite{BDF} to nonlinear models.

The presence of friends of friends who are not direct friends provides exclusion restrictions to construct instruments for $\tilde{\boldsymbol{y}}_{i,g}$. The natural instrument, $\mathbb E_g[\tilde{\boldsymbol{y}}_{i,g}]$, cannot be used directly because it is unobserved and cannot be calculated without making distributional assumptions about $\varepsilon_{i,g}$.\footnote{When errors are iid within a network and homoskedastic, $\E[\tilde{y}_{i,g}]$ is the optimal instrument: \citet{chamberlain1987asymptotic}. As shown in Supplemental Appendix \ref{app:reduced_form}, it does not have an easily computable closed-form even when the model parameters are known. Furthermore, assuming iid and homoskedastic errors is a very strong assumption in a network setting. }  However, as $\mathbb E_g[\tilde{\boldsymbol{y}}_{i,g}]$ is determined by friends' outcomes, we can use the distribution of $\boldsymbol{x}_{j,g}$ and ${\bar{\boldsymbol{x}}_{j,g}}$ for all $j$ who are direct friends of $i$ to construct instruments. The presence of friends of friends who are not direct friends ensures that ${\bar{\boldsymbol{x}}_{j,g}}$  captures relevant information excluded from $\boldsymbol{x}_{i,g}$ and ${\bar{\boldsymbol{x}}_{i,g}}$. Given the nonlinear nature of our model, we use the quantiles of these distributions, for example the deciles of $\boldsymbol{x}_{j,g}$ and ${\bar{\boldsymbol{x}}_{j,g}}$ among friends, as instruments. The strength of these instruments can be easily tested in practice using weak instrument tests. 

Given the constructed instrument, $\mathbf{Z}_g$, $\theta$ can be easily estimated using a two-stage least squares (TSLS) estimator. Consistency and asymptotic normality of the IV estimator in this setting follow from standard arguments for IV estimators. In particular, under Assumption \ref{ass:asymptotics}, the sample moment matrices formed from $\{\mathbf{Z}_g,\mathbf{W}_g\}_{g=1}^m$ converge to their population counterparts with full rank, and those formed from $\{\mathbf{Z}_g\}_{g=1}^m$ and residual vectors converge to zero. As a result, the TSLS estimator is consistent for $\theta$ and asymptotically normal. For completeness, Assumption \ref{ass:asymptotics} and formal asymptotic results with their proofs are provided in Supplemental Appendix \ref{append:inference}.

Finally, a popular empirical practice given a network dataset with many networks is to include network fixed-effects $\alpha_g$, giving us the model: \begin{align*}
y_{i,g} &= \alpha_{g} + {\tilde{\boldsymbol{y}}_{i,g}}^\intercal \beta + {\boldsymbol{x}_{i,g}}^\intercal \gamma + {\boldsymbol{\bar{x}}_{i,g}}^\intercal \delta  + \varepsilon_{i,g}  
\end{align*}
for $i=1,\dots,n_g$ and $g=1,\dots,m$. By demeaning at the network level, we obtain 
$$
y_{i,g} - \bar{y}_{g}= \left( \tilde{\boldsymbol{y}}_{i,g} - \bar{\tilde{\boldsymbol{y}}}_{g} \right)^\intercal \beta + \left( \boldsymbol{x}_{i,g} - \bar{\boldsymbol{x}}_{g} \right)^\intercal \gamma + \left( \boldsymbol{\bar{x}}_{i,g}- \boldsymbol{\bar{\bar{x}}}_{g} \right)^\intercal \delta + \varepsilon_{i,g} - \bar{\varepsilon}_g
$$
for $i=1,\dots,n_g$ and $g=1,\dots,m$. Theorem~\ref{theorem:id} can be directly extended to a setup with network fixed-effects by replacing $\mathbf{W}_g$ with its demeaned counterpart in Part \ref{ass:id:nonsingular} of Assumption~\ref{ass:id}. 

\subsection{Encompassing Test on Empirical Specifications}\label{sec:econometrics:test}

In Subsection \ref{sec:implementation:restriction}, we propose a quantile restriction on the fully saturated rank-dependent peer effect model, for parsimony and tractability. While certain empirical contexts may naturally call for a particular empirical specification, such as the max or min specification, there are many empirical contexts that offer little basis for imposing such \textit{a priori} quantile selections. Therefore, in this subsection, we introduce an encompassing test which provides guidance on the choice of location and number of quantiles, based on \citet{smith1992non}, and develop a data-driven specification selection procedure.\footnote{This test can also be used to test non-quantile variations of our saturated model \eqref{eq:peer_effect_model}. As we use quantiles in our empirical application, we will focus our discussion on this case. }

The encompassing test we propose compares a pair of candidate specifications, treating one as the null hypothesis and the other as the alternative. Importantly, the test does not require the alternative specification to nest the null specification. Moreover, consistency of the test does not depend on the alternative specification being correctly specified. Instead, we use the alternative specification as a lens to evaluate the null specification. Therefore, the procedure can be used to compare any two non-nesting specifications, such as a three-quantile model with quantiles at $\{0,1/2,1\}$ and a four-quantile model with quantiles at $\{0,1/3,2/3,1\}$. In this example, the test remains consistent even if the true model is not the four-quantile specification as long as the true model differs from the three-quantile specification from the perspective of the four-quantile specification. Similarly, the test can be used to test the LiM model, since it is a special case of our general model. 

Let $h_1$ and $h_2$ denote the two empirical specifications that we compare with the encompassing test, with $h_1$ being the null specification and $h_2$ being the alternative. Given $h_1$ and $h_2$, we will have two sets of endogenous variables. Let $\{\mathbf{W}_g^1\}_{g=1}^m$ and $\{\mathbf{W}_g^2\}_{g=1}^m$ denote the explanatory variables and $\{\mathbf Z_g^1\}_{g=1}^m$ and $\{\mathbf{Z}_g^2 \}_{g=1}^m$ denote the instruments, respectively under $h_1$ and $h_2$. When $h_2$ refers to the fully saturated specification \eqref{eq:peer_effect_model}, $\mathbf{W}_g^2$ collects the explanatory variable vectors from \eqref{eq:id_peer_effect_model}. When $h_1$ and/or $h_2$ refer to a quantile-based model, the endogenous variables in $\mathbf{W}_g^1$ and/or $\mathbf{W}_g^2$ will follow the quantile construction as discussed in Subsection \ref{sec:implementation:restriction}.

The encompassing test \citep{smith1992non} is implemented as follows: \begin{enumerate}
\item Regress $\mathbf{y}_g$ on $\mathbf{W}_g^1$, using $\mathbf{Z}_g^1$ as instruments: $$
\hat{\theta}^1 = \big(\mathbf H^1 \big)^{-1} \cdot \frac{1}{m} \sum_{g=1}^m {\mathbf{W}_g^1}^{\intercal} \mathbf{Z}_g^1 \left(\frac{1}{m} \sum_{g=1}^m{\mathbf Z_g^1}^{\intercal} \mathbf{Z}_g^1 \right)^{-1} \frac{1}{m} \sum_{g=1}^m{\mathbf{Z}_g^1}^{\intercal} \mathbf y_g
$$
where $\mathbf H^1 = \frac{1}{m} \sum_{g=1}^m {\mathbf{W}_g^1}^{\intercal} \mathbf{Z}_g^1 \big( \frac{1}{m} \sum_{g=1}^m {\mathbf Z_g^1}^{\intercal} \mathbf{Z}_g^1 \big)^{-1} \frac{1}{m} \sum_{g=1}^m {\mathbf{Z}_g^1}^{\intercal} \mathbf W_g^1$. $\hat{\theta}^1$ is the TSLS estimator for the model parameter in the null specification.
\item Construct the residuals from the null specification: $\hat{\boldsymbol\varepsilon}_g^1 = \mathbf{y}_g - \mathbf{W}_g^1 \hat{\theta}^1, \quad \forall g=1,\dots,m$.
\item Regress $\hat{\boldsymbol{\varepsilon}}_g^1$ on $\mathbf{W}_g^2$, using $\mathbf{Z}_g^2$ as instruments:  \[
\hat{\psi} = \big( {\mathbf H^2} \big)^{-1} \cdot \frac{1}{m} \sum_{g=1}^m {\mathbf{W}_g^2}^{\intercal} \mathbf{Z}_g^2 \left( \frac{1}{m} \sum_{g=1}^m {\mathbf Z_g^2}^{\intercal}\mathbf{Z}_g^2 \right)^{-1} \frac{1}{m} \sum_{g=1}^m {\mathbf{Z}_g^2}^{\intercal} \hat{\boldsymbol{\varepsilon}}_g^1
\]
where $\mathbf H^2 = \frac{1}{m} \sum_{g=1}^m {\mathbf{W}_g^2}^{\intercal} \mathbf{Z}_g^2 \big( \frac{1}{m} \sum_{g=1}^m {\mathbf Z_g^2}^{\intercal} \mathbf{Z}_g^2 \big)^{-1} \frac{1}{m} \sum_{g=1}^m {\mathbf{Z}_g^2}^{\intercal} \mathbf W_g^2$. 
\item Construct a Wald statistic based on $\hat{\psi}$. Reject the null specification when the resulting Wald statistic exceeds the critical value from an appropriate chi-squared distribution. 
\end{enumerate}

$\hat{\boldsymbol{\varepsilon}}_g^1$ is the remaining variation in $\mathbf{y}_g$ that is not explained by the patterns of peer effect allowed in the first model specification $h_1$. If there are features of the true peer effect model that the first model specification $h_1$ fails to capture but the second model specification $h_2$ does capture, $\hat{\boldsymbol{\varepsilon}}_g^1$ will be correlated with $\mathbf{Z}_g^2$, leading to $\hat{\psi}$ being centered around a nonzero vector. 

In Supplemental Appendix \ref{app:test}, under some regularity conditions, we formally show that $\hat{\psi}$ is asymptotically normal with a consistently estimable asymptotic variance matrix. This result does not require either of the models $h_1$ or $h_2$ to be correctly specified. The asymptotic theory accounts for the estimation error in $\hat{\boldsymbol{\varepsilon}}_g^1$ and so does our asymptotic variance estimator. If the true peer effect model cannot be well approximated by the null specification $h_1$ when assessed through the alternative specification $h_2$, $\hat{\psi}$ is centered around a nonzero vector, giving us consistency of the test.

\subsection{Comparison to Other Peer Effect Estimands}\label{sec:ident:weights}

A natural question, given that our model generalizes the standard LiM model, is to what extent existing estimands recover the key parameters of interest from our model. A commonly accepted minimal standard for such estimands is that they represent weighted averages of the underlying heterogeneity, such as in the analysis of instrumental variables \citep{mogstad2024instrumental} or the analysis of Difference-in-Differences estimands \citep{de2020two,goodman2021difference}. To explore this, define $\tilde{\mathbf{X}}^{\text{LiM}} = \begin{pmatrix} \mathbf{X} & \mathbf G\mathbf{X} \end{pmatrix}$ and $\tilde{\mathbf{X}}^{\text{LiS}} = \begin{pmatrix} \mathbf{X} & \mathbf A\mathbf{X} \end{pmatrix}$, with network-level counterparts $\tilde{\mathbf{X}}_g^{\text{LiM}} = \begin{pmatrix} \mathbf{X}_g & \mathbf{G}_g\mathbf{X}_g \end{pmatrix}$ and $\tilde{\mathbf{X}}_g^{\text{LiS}} = \begin{pmatrix} \mathbf{X}_g & \mathbf{A}_g\mathbf{X}_g \end{pmatrix}$, and the population regression coefficients 
\begin{align*} 
\mathbf{P} = \left( \plim_{m \to \infty} \frac{1}{m} \sum_{g=1}^m {\tilde{\mathbf{X}}_g}^\intercal \tilde{\mathbf{X}}_g \right)^{-1} \plim_{m \to \infty} \frac{1}{m} \sum_{g=1}^m {\tilde{\mathbf{X}}_g}^\intercal \mathbf{Z}_g
\end{align*} 
for the LiM and LiS models. We can then define the standard LiM estimand in a just-identified case as follows: \begin{align*} 
\beta^{\mathrm{LiM}} := \left( \frac{1}{m} \sum_{g=1}^m  \E\left[\left(  \tilde{\mathbf{Z}}_g^{\mathrm{LiM}} \right)^\intercal \mathbf G_g\mathbf{y}_g \right] \right)^{-1} \frac{1}{m} \sum_{g=1}^m  \E \left[ \left( \tilde{\mathbf{Z}}_g^{\mathrm{LiM}} \right)^\intercal \mathbf{y}_g \right]
\end{align*} 
where \[
\tilde{\mathbf{Z}}_g^{\mathrm{LiM}} = \mathbf{Z}_g - \tilde{\mathbf{X}}_g^{\mathrm{LiM}} \mathbf{P}^{\mathrm{LiM}}
\]
given some set of instruments $\{\mathbf{Z}_g\}_{g=1}^m$ and similarly for the LiS model by replacing $\mathbf G_g$ with $\mathbf A_g$.

\begin{assumption} \label{ass:miss} \quad 
\begin{enumerate}[label=\roman*] 
\item \label{ass:lim} There exist instruments $\{\mathbf{Z}_{g}\}_{g=1}^m$, known functions of $(\mathbf{X}_g, \mathbf{A}_g)$, satisfying \begin{align*} \frac{1}{m} \sum_{g=1}^m  \E\left[\left( \tilde{\mathbf{Z}}_g^{\mathrm{LiM}} \right)^\intercal \mathbf G_g\mathbf{y}_g \right] \end{align*} 
is nonsingular.
\item \label{ass:lis} There exist instruments $\{\mathbf{Z}_{g}\}_{g=1}^m$, known functions of $(\mathbf{X}_g, \mathbf{A}_g)$, satisfying \begin{align*} \frac{1}{m} \sum_{g=1}^m  \E \left[ \left( \tilde{\mathbf{Z}}_g^{\mathrm{LiS}} \right)^\intercal \mathbf A_g \mathbf{y}_g \right] \end{align*} 
is nonsingular.
\end{enumerate} 
\end{assumption} 

\noindent These conditions are the instrument relevance conditions used for the LiM model and the LiS model, as shown in \citet{BDF}. 

\begin{proposition} \label{prop:misspecification} 
Under Assumptions \ref{ass:equilibrium}, \ref{ass:exogeneity} and \ref{ass:miss}-\ref{ass:lim}, \begin{align*} 
{\beta}^{\mathrm{LiM}} = \sum_{d=1}^{\bar{d}}\sum_{k=1}^d w_{k,d}^{\textup{LiM}} d \cdot \beta_{k,d} + o_p(1)
\end{align*} 
where $\sum_{d=1}^{\bar{d}}\sum_{k=1}^d w_{k,d}^{\textup{LiM}} = 1$ and $w_{k,d}^{\textup{LiM}} \lessgtr 0$. Similarly, under Assumptions \ref{ass:equilibrium}, \ref{ass:exogeneity} and \ref{ass:miss}-\ref{ass:lis}, \begin{align*} 
{\beta}^{\textup{LiS}} = \sum_{d=1}^{\bar{d}}\sum_{k=1}^d w_{k,d}^{\textup{LiS}} \beta_{k,d} + o_p(1)
\end{align*} where $\sum_{d=1}^{\bar{d}}\sum_{k=1}^d w_{k,d}^{\textup{LiS}} = 1$ and $w_{k,d}^{\textup{LiS}} \lessgtr 0$. 
\end{proposition} 

\begin{proof} 
See Appendix \ref{app:missspecproof}. 
\end{proof} 

\noindent Note that, in the case of LiM misspecification, the weights are applied not directly to $\beta_{k,d}$ but to $d \cdot \beta_{k,d}$. The rank-dependent peer effect coefficient $\beta_{k,d}$ represents the impact of a single peer outcome when there are $d$ peers in total, so it should be rescaled by $d$ to interpret it as a coefficient on a `representative' or `average' peer. For instance, when $\beta_{1,d} = \dots = \beta_{d,d} = \beta$, we have: $
\sum_{k=1}^d \beta_{k,d} \tilde{y}_{i,k} = d \beta \bar{y}_i, 
$ where the coefficient on the average peer outcome is $d \beta$.

Proposition \ref{prop:misspecification} demonstrates that both the LiM estimand $\beta^{\mathrm{LiM}}$ and the LiS estimand $\beta^{\mathrm{LiS}}$ are weighted sums of the rank-dependent peer effect coefficients $\{\beta_{k,d}\}_{k,d}$. While the weights sum to one, which is reasonable, there is no guarantee that the weights have the same sign. As a result, it is possible for all $\beta_{k,d}$ to be positive, while the LiM estimand $\beta^{\mathrm{LiM}}$ is negative, and vice versa. The expressions for the weights are provided in the Appendix.

\section{Simulation}\label{sec:sim}

In this section, we investigate the finite sample properties of our estimators and compare their performance to existing estimators, specifically the CES approach of \citet{boucher2024toward}, under various data generating processes. To ensure our simulations are as informative as possible, we will design the setup of the simulations to be similar to the empirical application from Section \ref{sec:emp}.  

As most empirical applications contain many small networks, we generate networks of size $n_g=50$ and vary the number of networks, $m$. 
To generate networks, we generate the links based on a logit model:
\begin{align*}
    a_{i,j,g} = \mathbf{1}\{ v_{i,j,g} \leq -3 + r_{i,j,g} \}
\end{align*}
where $v_{i,j,g} \overset{\text{iid}}{\sim} \text{logistic}$ is an unobserved shock and $r_{i,j,g} \sim N(0,1)$ is a covariate. When we estimate the fully saturated version of our model, we will limit the maximum number of links in the network, $\bar{d}$, to 5. This is achieved by dropping links from individuals with more links than $\bar{d}$ until all individuals have at most $\bar{d}$ links. This procedure generates a network with a roughly uniform degree distribution, with a small share of individuals in each network having zero links. For the results in this section, we fix the networks across simulations.

Our model specification uses the quantile specification from Section \ref{sec:implementation:restriction}, with equally spaced quantiles. We primarily use four quantiles at $\{0,1/3,2/3,1\}$. This gives the outcome equation 
\begin{align*}
    y_{i,g} = \sum_{k=1}^4 \beta^k \tilde{y}_{i,g}^k + \gamma_1 x_{1,i,g} + \gamma_2 x_{2,i,g}+ \delta_1 \bar{x}_{1,i,g} + \delta_2 \bar{x}_{2,i,g} + \alpha_g + \varepsilon_{i,g}
\end{align*}
where $\tilde{y}_{i,g}^k$ is the $(k-1)/3$-th quantile of peer outcomes for unit $i$, $x_{1,i,g}$ and $x_{2,i,g}$ are the two covariates for unit $i$, and $\alpha_g$ is the network fixed effect for network $g$. The fixed effects are drawn from a uniform distribution on the interval [8, 10]. This ensures that the outcome is always positive, which is necessary for the CES estimator to be well defined. The continuous covariate $x_{1,i,g}$ is normally distributed with a mean equal to the fixed effect and a variance equal to one, while the discrete covariate $x_{2,i,g}$ is distributed Poisson with a mean equal to $2$. We also include contextual effects $\bar{x}_{1,i,g}$ and $\bar{x}_{2,i,g}$, the average covariate values for all $i$'s peers, for each of our two covariates. The error term $\varepsilon_{i,g}$ is generated from a normal distribution with variance equal to $1$.  The coefficients on the two covariates are $\gamma_1 = 1.5$ and $\gamma_2 = -0.8$, respectively, and the coefficients on the two contextual effects are $\delta_1 = -0.5$ and $\delta_2 = 1.2$, where the values are set to be similar to the patterns we see in our empirical application.\footnote{Supplemental Appendix \ref{app:sim_without_contextual} shows all results from this section excluding these contextual effects, i.e., setting $\delta_1 = \delta_2 = 0$. } To avoid unnecessary repetition in our tables, we will only show results for a subset of $\beta$ for the saturated model. 

We apply both the correctly specified quantile model and the fully saturated rank-dependent model to estimate the peer effect in simulated samples. Since the fully saturated model nests any quantile model, the fully saturated model is also correctly specified. Given a quantile specification, we calculate the corresponding parameter values of $\beta_{k,d}$ as the true values used to calculate bias and MSE for the saturated model.\footnote{For $d_i = 1$, $\beta_{1,1} = \beta^1 + \cdots + \beta^4$; for $d_i = 2$ we get $\beta_{1,2} = \beta^1 + \frac{2}{3} \beta^2 + \frac{1}{3} \beta^3$ and so on.} In the saturated rank model, we use the peer covariates and contextual effects as instruments, ordered by the values of each variable. Since we have two covariates, this yields $4\sum_{d=1}^{\bar{d}}d$ instruments in total. For the quantile models we instead use the 4 quantiles of the covariates and contextual effects as our instruments, giving a total of 16 instruments.

\begin{table}[!t]
\vspace{3mm}
    \centering
    \small
    \renewcommand{\arraystretch}{1.8}
    \caption{Finite Sample Performance of Rank and Quantile Models} \label{tab:sim_change_ng}
    \begin{tabular}{l ccccc c ccccc c}
        \hline \hline
         & \multicolumn{5}{c}{Bias } & & \multicolumn{5}{c}{MSE} \\
         \cline{2-6} \cline{8-12}       
         & (1) & (2) & (3) & (4) &(5) &  & (6) & (7) & (8)& (9) & (10) \\
                 
\multicolumn{5}{l}{\textbf{Panel A: Saturated model}} \\ 
         \hline 
$\beta_{1,2}$ & 0.007 & 0.001 & 0.002 & -0.001 & 0.000 &  & 0.010 & 0.010 & 0.005 & 0.002 & 0.001 \\ 
$\beta_{2,2}$ & 0.003 & 0.005 & 0.001 & 0.002 & 0.000 &  & 0.009 & 0.009 & 0.004 & 0.002 & 0.001 \\ 
$\beta_{1,5}$ & -0.010 & 0.002 & 0.010 & 0.001 & -0.001 &  & 0.035 & 0.025 & 0.040 & 0.013 & 0.011 \\ 
$\beta_{2,5}$ & 0.023 & -0.005 & -0.016 & -0.000 & 0.000 &  & 0.087 & 0.083 & 0.109 & 0.064 & 0.072 \\ 
$\beta_{3,5}$ & -0.019 & 0.008 & 0.010 & -0.006 & 0.003 &  & 0.101 & 0.130 & 0.122 & 0.111 & 0.110 \\ 
$\beta_{4,5}$ & 0.005 & -0.003 & 0.004 & 0.009 & -0.001 &  & 0.101 & 0.103 & 0.089 & 0.086 & 0.061 \\ 
$\beta_{5,5}$ & 0.011 & 0.004 & -0.006 & -0.003 & -0.001 &  & 0.047 & 0.030 & 0.027 & 0.018 & 0.010 \\ 

\multicolumn{5}{l}{\textbf{Panel B: Restricted model } } \\
         \hline
      $\beta^1$ & 0.005 & -0.002 & 0.003 & 0.000 & -0.000 &  & 0.019 & 0.016 & 0.014 & 0.006 & 0.003 \\ 
$\beta^2$ & -0.006 & 0.004 & -0.008 & -0.001 & 0.001 &  & 0.089 & 0.109 & 0.090 & 0.041 & 0.021 \\ 
$\beta^3$ & 0.004 & 0.004 & 0.007 & 0.002 & -0.000 &  & 0.104 & 0.113 & 0.078 & 0.038 & 0.023 \\ 
$\beta^4$ & -0.000 & -0.003 & -0.001 & -0.001 & -0.000 &  & 0.024 & 0.018 & 0.011 & 0.006 & 0.004 \\ 
 \hline
         $m$ &  5 & 10 & 20 & 50& 100 & &  5 & 10 & 20 & 50& 100 \\
         $\bar{d}$ &  5 & 5 & 5 & 5& 5 & & 5 & 5 & 5 & 5& 5 \\
         \hline
    \end{tabular}
     \begin{minipage}{0.97\textwidth}   \footnotesize \vspace{3mm} \textit{Notes:}  The data generating process is the quantile model with $\beta = \left(-0.05,0.35,0.15,0.1\right)$. $m$ gives the number of networks used in the simulation and $\bar{d}$ is the maximum number of peers for each individual. All the simulations include network fixed effects. The instruments used are the distribution of the peers' covariate values, as measured either by the ranked peer covariate values or the quantiles of peer covariate values. The results are based on 10{,}000 draws.
 \end{minipage}
\end{table}

Table \ref{tab:sim_change_ng} gives the results for TSLS for the saturated rank model in Panel A and the quantile model in Panel B.\footnote{In Supplemental Appendix \ref{app:sim}, we show that the OLS estimator is, as expected, biased in both the saturated and quantile versions of our model.} Each row gives the values for a given parameter $\beta_{k,d}$ or $\beta^k$. Across the columns, we increase the sample by gradually increasing the number of networks from $5$ to $100$. These sample sizes therefore go from relatively small to approaching larger network datasets, though the networks are still small. For comparison, in our empirical application we will have 141 networks (schools) with an average size of 531. 
 
In Panel A we see the results for the saturated model for the parameters $\beta_{k,2}$ and $\beta_{k,5}$. Interestingly, the parameters $\beta_{k,2}$ are generally better estimated than $\beta_{k,5}$ and also the precision of the parameters depends on the ranking, with the MSEs for $\beta_{2,5},\beta_{3,5},\beta_{4,5}$ being bigger than those of $\beta_{1,5},\beta_{5,5}$ up to a factor of 10. This is because with larger $d_i$, there is less variation in the middle of the ordered peer outcomes $\big(\tilde{y}_{i,1}, \dots, \tilde{y}_{i,d_i} \big)$. Conversely, when the identity of an individual's 2nd, 3rd, or 4th ranked friend is less dependent on the shocks and mostly driven by the covariates, we expect the TSLS estimator to perform better for $\beta_{2,5}, \dots, \beta_{4,5}$ as well. 

Comparing the quantile to the rank model, we see that the quantile model generally performs better in terms of MSE, and the TSLS estimators are consistent for every $\beta^k$; (correct) knowledge of the peer effect pattern helps estimation. We see a similar pattern as before when comparing the middle to the extreme quantiles, with $\beta^2,\beta^3$ having a relatively lower precision compared to $\beta^1,\beta^4$, especially at smaller sample sizes.

\subsection{Model Selection}

\begin{table}[b!]
\centering
\caption{Monte Carlo Simulations --- Encompassing Tests}
\label{tab:sim_tests}
\begin{threeparttable}
\begin{tabular}{P{1.6cm}P{1.6cm}lP{1.6cm}P{1.6cm}lP{1.6cm}P{1.6cm}}
\toprule
\multicolumn{2}{c}{2 vs. 3} & &\multicolumn{2}{c}{3 vs. 4} & &\multicolumn{2}{c}{4 vs. 5}\\[0.5ex]
\cline{1-2} \cline{4-5} \cline{7-8}
\addlinespace[0.5ex]
5\% & 10\% && 5\% & 10\% && 5\% & 10\%\\
\midrule
\multicolumn{8}{c}{DGP A:   $\beta = (0, 0.05, 0.2, 0.3)$}      \\[0.5ex]
1.000    & 1.000   &    & 0.105   & 0.180  &   & 0.011  & 0.020 \\[1.5ex]
\multicolumn{8}{c}{DGP B: $\beta =   (0.3, 0.2, 0.05, 0)$}      \\[0.5ex]
1.000    & 1.000   &    & 0.914   & 0.953  &   & 0.013  & 0.026 \\[1.5ex]
\multicolumn{8}{c}{DGP C: $\beta = (0,   0.275, 0.275, 0)$}     \\[0.5ex]
1.000    & 1.000   &    & 0.994   & 0.998  &   & 0.009  & 0.021 \\[1.5ex]
\multicolumn{8}{c}{DGP D: $\beta =   (0.275, 0, 0, 0.275)$}     \\[0.5ex]
0.011    & 0.023   &    & 0.008   & 0.017  &   & 0.013  & 0.025 \\[1.5ex]
\multicolumn{8}{c}{DGP E: $\beta =   (-0.05, 0.35, 0.15, 0.1)$} \\[0.5ex]
1.000    & 1.000   &    & 1.000   & 1.000  &   & 0.009  & 0.021 \\[1.5ex]
\multicolumn{8}{c}{DGP F (LiM model):   $\beta = 0.55$}         \\[0.5ex]
1.000    & 1.000   &    & 0.961   & 0.984  &   & 0.442  & 0.574 \\\bottomrule
\end{tabular}
\begin{tablenotes}[para,flushleft]
The columns labeled ``$a$ vs. $b$,'' for integers $a$ and $b$, report the share of rejections of the null hypothesis that the model with $a$ quantile levels does not perform worse than the model with $b$ quantile levels, at the significance levels indicated in the second row.
\end{tablenotes}
\end{threeparttable}
\end{table}

Implementing our estimator requires the researcher to select how many quantiles they want to include in the model. Table \ref{tab:sim_tests} shows the rejection rates for the encompassing test from Section \ref{sec:econometrics:test}. We run these tests for six different DGPs, four DGPs with four quantiles as the true model, one with only the min and max peer outcome mattering, and the LiM model.  

For the DGPs A, B, C, E, and F, we see that the test correctly rejects the null hypothesis of the two-quantile specification with probability one. On the other hand, for the null of the four-quantile specification, the rejection rates are fairly steady at around $0.2\alpha$, except for the LiM DGP F where the four-quantile specification is incorrect. This seeming contradiction of conservative size but high power likely stems from the large difference between our two hypotheses. This result is not driven by the structure of the peer effects, as can be seen from the remarkably stable results across DGPs B, C, and E. It is also not driven by having 4 as the true number of quantiles in the DGP, as DGP D shows that when the true number of quantiles is 2, the test does not reject keeping 3 quantiles over 4, since the three-quantile model is also correctly specified.

In total, these results indicate that as long as the true DGP is some number of equally spaced quantiles, this test is capable of distinguishing the correct number with a large degree of power. However, this is not the case when the model is misspecified, as is the case in DGP F. To capture the mean peer outcome we would need an infinite number of quantiles, which is clearly infeasible. As one would expect, the test therefore consistently prefers more quantiles.  

\subsection{Comparison to the CES and LiM Estimators}
Our next simulation compares the ability of our estimator to capture different patterns of peer effects to the other two main estimators in the literature, the LiM model and the CES estimator of \citet{boucher2024toward}. We will do this by simulating data from our DGP and then estimating the quantile version of our model, the CES model, and the LiM model on the simulated data. Since we do not estimate the saturated model in this exercise, we no longer limit the maximum number of links, $\bar{d}$, to 5. Other than that, the network formation model is the same as in the rest of this section. In most networks, most individuals have between 0 and 3 links, with the maximum number of links being 12. This is similar to the degree distributions we see in our empirical application. 

The CES aggregator from \citet{boucher2024toward} is defined as: $$
y_{i,g} = \beta^\textup{CES} \left(\sum_{j=1}^{n_g} \frac{a_{i,j,g}}{d_i} y_{j,g}^{\rho}\right)^{\frac{1}{\rho}} + {\boldsymbol{x}_{i,g}}^\intercal \gamma + {\bar{\boldsymbol{x}}_{i,g}}^\intercal \delta + \varepsilon_{i,g},
$$
where $\rho$ is a substitution parameter. We estimate the model following the method of \citet{boucher2024toward}. For details on the estimation of the model, see our package documentation. 

\begin{table}[t!]
\centering
\small
\caption{Monte Carlo Simulations --- Estimation of the Quantile Specification}
\label{tab:sim_compare_estimand}
\begin{threeparttable}
\begin{tabular}{cccclclcc}
\toprule
\multicolumn{4}{c}{Quantiles}                                           &            &    \multicolumn{1}{c}{LiM}           &           & \multicolumn{2}{c}{CES}           \\
\cline{1-4} \cline{6-6} \cline{8-9}
\addlinespace[0.5ex]
$\beta^1$ & $\beta^2$ & $\beta^3$ & $\beta^4$ & & 
$\beta^\textup{LiM}$   & & $\rho$  & $\beta^\textup{CES}$ \\
\midrule
\multicolumn{9}{c}{DGP A:   $\beta = (0, 0.05, 0.2, 0.3)$}            \\[0.5ex]
-0.000    & 0.050     & 0.200     & 0.300     &  & 0.386   &  & 30.057  & 0.562   \\
(0.005)   & (0.023)   & (0.044)   & (0.028)   &  & (0.020) &  & (8.914) & (0.011) \\[2ex]
\multicolumn{9}{c}{DGP B: $\beta =   (0.3, 0.2, 0.05, 0)$}                        \\[0.5ex]
0.300     & 0.200     & 0.050     & -0.000    &  & 0.715   &  & -6.383  & 0.502   \\
(0.008)   & (0.033)   & (0.058)   & (0.033)   &  & (0.022) &  & (9.338) & (0.029) \\[2ex]
\multicolumn{9}{c}{DGP C: $\beta = (0,   0.275, 0.275, 0)$}                       \\[0.5ex]
0.000     & 0.275     & 0.275     & 0.000     &  & 0.476   &  & 3.057   & 0.554   \\
(0.006)   & (0.025)   & (0.048)   & (0.030)   &  & (0.013) &  & (0.275) & (0.010) \\[2ex]
\multicolumn{9}{c}{DGP D: $\beta =   (0.275, 0, 0, 0.275)$}                       \\[0.5ex]
0.275     & -0.000    & -0.000    & 0.275     &  & 0.635   &  & -1.714  & 0.511   \\
(0.006)   & (0.028)   & (0.052)   & (0.031)   &  & (0.014) &  & (0.253) & (0.014) \\[2ex]
\multicolumn{9}{c}{DGP E: $\beta =   (-0.05, 0.35, 0.15, 0.1)$}                   \\[0.5ex]
-0.050    & 0.350     & 0.150     & 0.100     &  & 0.434   &  & 5.004   & 0.561   \\
(0.006)   & (0.024)   & (0.047)   & (0.029)   &  & (0.016) &  & (0.575) & (0.010) \\[2ex]
\multicolumn{9}{c}{DGP F (LiM model):   $\beta = 0.55$}                           \\[0.5ex]
0.112     & 0.196     & 0.098     & 0.144     &  & 0.550   &  & 1.002   & 0.550   \\
(0.006)   & (0.028)   & (0.052)   & (0.031)   &  & (0.008) &  & (0.130) & (0.010)\\\bottomrule
\end{tabular}
\begin{tablenotes}[para,flushleft] 
\footnotesize
The models are simulated and estimated 1{,}000 times. Values without parentheses represent average peer effect estimates, while those in parentheses correspond to standard deviations. The instrument matrix includes the quantiles of $\boldsymbol{x}$ and $\bar{\boldsymbol{x}}$ among friends, computed at ten levels uniformly spaced between 0 and 1. DGPs A--E are generated from the proposed quantile-based model, with four quantiles at $\{0, ~1/3, ~ 2/3,~ 1\}$, and $\beta = (\beta^1, ~\beta^2, ~\beta^3, ~\beta^4)$ is the vector of peer effects at each quantile. DGP F follows the standard LiM model with only spillover effects, where $\beta^\textup{LiM} = 0.55$. All estimations account for unobserved network-level heterogeneity using fixed effects.
\end{tablenotes}
\end{threeparttable}
\end{table}

Table \ref{tab:sim_compare_estimand} shows the results from estimating our quantile estimator and the CES estimator on the same data. DGP A and B both have monotonically increasing or decreasing peer effects. Our estimator fully captures these effects. These effects are also well captured by the CES aggregator, which is able to correctly assign $\rho > 1$ for the increasing case and $\rho < 1$ for the decreasing case. In case B, however, it is still not able to fully capture the total effect coming from the peer effects, which equals 0.55 across all our DGPs. The LiM model similarly struggles for both these DGPs, either underestimating or overestimating the total spillover in the data. 

DGPs C and D instead put all the weight on either the middle or extremes of the distribution. Importantly, our estimator is fully able to estimate the zeros in the parameter space, correctly identifying the parts of the peer outcome distributions that have no effect on an individual's own outcome. Interestingly, the LiM model performs much better for these symmetric cases, implying the LiM model is more able to approximate the true peer effects when they are symmetric, or in cases where the peers around the median have the largest weight. The CES estimator, however, struggles with these cases. Since neither of these cases has concave or convex peer effects, the assumptions of the model do not hold. This leads to the CES estimator wrongly estimating that higher peer outcomes are more important in DGP~C ($\rho > 1$), and that lower peer outcomes are more important in DGP~D ($\rho < 1$), despite this not being the case for either DGP. 

DGP E introduces peer effects of different signs, by having the lowest peer outcome have a negative effect while the other quantiles have a positive effect. Our estimator is fully able to capture this pattern, and since the majority of the weight is on the middle of the distribution, the LiM model also performs decently. The CES estimator wrongly fits $\rho =5$, implying peers with higher outcomes are more important for the spillover. This is likely due to the sharp increase in the peer effect from the lowest outcome friend ($\beta^1$) to the $33\%$ quantile ($\beta^2$). However, it completely contradicts the true DGP where the most influential peers are those at the second quantile. 

Finally, we simulate the data from a LiM model. As expected, both the LiM model and the CES model are able to fully capture the spillover, and the CES correctly identifies that $\rho = 1$, i.e., a LiM model. Though the quantile version of our model does not nest the LiM model, it still estimates a roughly uniform effect of peers. This is an important reminder that restrictions can reduce the performance of our model if the added restrictions are not correct, though the model can still recover meaningful results under misspecification. 

\section{Empirical Application} \label{sec:emp}

In this section, we present an empirical application using the Wave I dataset from the National Longitudinal Study of Adolescent to Adult Health (Add Health). We examine several outcomes and show that the quantile model captures diverse patterns of peer effects that challenge existing specifications. To illustrate the effectiveness of our approach, we also conduct a counterfactual analysis, demonstrating that ignoring these patterns can lead to the misidentification of key players and reduce the impact of targeted policies.

\subsection{Add Health Data}

Wave I of the Add Health survey provides nationally representative and detailed information on \nth{7}--\nth{12} graders from 144 schools during the 1994--1995 school year in the United States (US). Approximately 90{,}000 students completed a questionnaire covering demographics, family background, academic performance, health-related behaviors, and friendship links. Each respondent could nominate up to five male and five female best friends within the same school.\footnote{The dataset has some limitations that merit discussion. First, some referred friend identifiers are missing and are therefore removed from the network. Recent literature proposes methods to address this issue; however, most of these approaches focus on linear-in-means models and cannot be easily extended to nonlinear settings \cite[e.g.,][]{lewbel2023social, boucher2025estimating}. Second, the observed degree is censored, as students can nominate up to ten friends. However, only 1.12\% of students reach this maximum. To maintain focus on the main objective of the paper, we do not address this censoring issue. Despite these limitations, it is worth noting that the Add Health dataset remains the most comprehensive network dataset that is currently available for studying peer effects.}

We study eleven outcomes, including grade point average (GPA), academic effort, participation in extracurricular activities, future perceptions, trouble at school, smoking, drinking, risky behaviors, self-esteem, physical exercise, and fighting. All of these outcomes, except for future perceptions, have been studied by \cite{boucher2024toward} using the CES model. The future perceptions outcome is constructed as the sum of binary indicators for whether students believe they will live to age 35, avoid HIV/AIDS, graduate from college, and have a middle-class family income by age 30. We control for several exogenous variables, including student age, grade, sex, race, Hispanic ethnicity (Spanish-speaking), and mother's education and employment. We also control for contextual variables, defined as the average of the exogenous variables among a student’s friends. 

After removing observations with missing values, the final sample comprises approximately 75{,}000 students from 141 schools.\footnote{See Table \ref{tab:size} in Supplemental Appendix \ref{append:addhealth} for the sample sizes for different outcomes.} The average number of friends per student is 3.47; 22\% of students have no friends, and approximately 64\% have four friends or fewer.

\subsection{Empirical Results}

The specification test results and other auxiliary diagnostics are presented in Table~\ref{tab:appresult:test} and the peer effect estimation results are provided in Table~\ref{tab:appresult}. The left panel of Table~\ref{tab:appresult:test} contains the encompassing test results on the quantile specification and the LiM specification; within each quantile specification, the quantile levels are set to be evenly spaced from 0 to 1. The $p$-values for the encompassing test in Table~\ref{tab:appresult:test} indicate that the null specification with two quantiles is rejected in favor of the one with three quantiles at the 5\% significance level, for every outcome other than \textit{Future perception} and \textit{Fighting}. On the other hand, the three-quantile null is often not rejected against the four-quantile alternative and the LiM.\footnote{The choice of three quantiles is likely driven by the relative sparsity of the friendship network, where 53\% of students have 3 or less friends, and 64\% have four or fewer.} This suggests that for most of our outcomes the three-quantile specification sufficiently captures the peer effect heterogeneity while omitted ordered peer outcomes, either at additional quantile levels or as an average, exert little influence. We therefore adopt the specification with three quantiles throughout the main empirical analysis and present results based on the four-quantile specification in Table~\ref{tab:append:addhealth} in Supplemental Appendix~\ref{append:addhealth}.\footnote{As a robustness check, we also estimate the three-quantile model on the subsample of nonisolated students; that is, students who have at least one friend (see Panel B of Table~\ref{tab:append:addhealth} in Supplemental Appendix~\ref{append:addhealth}). Our results are robust to the exclusion of isolated students.} 

\begin{table}[htbp]
\centering
\footnotesize
\caption{Empirical Results --- Specification Tests}
\label{tab:appresult:test}
\begin{threeparttable}
\begin{tabular}{ccccccccccccc}
\toprule
\multicolumn{3}{c}{Encompassing test }&& \multicolumn{3}{c}{} && \multirow{3}{4em}{\centering KP LM test \\ $p$-value} && \multicolumn{3}{c}{Monotonicity test} \\
\multicolumn{3}{c}{$p$-value} && \multicolumn{3}{c}{First Stage F Stat} && && \multicolumn{3}{c}{$p$-value} \\ 
\cmidrule{1-3} \cmidrule{5-7} \cmidrule{11-13}
2 vs. 3 & 3 vs. 4 & 3 vs. LiM & &$\tilde y_1$ & $\tilde y_2$ & $\tilde y_3$ &&&& Unif & Incr. & Decr. \\
\midrule
\multicolumn{12}{c}{Academic achievements (GPA)} \\[0.5ex]
0.000 & 0.000 & 0.011 & & 117.559 & 44.905 & 114.064 & & 0.000 & & 0.000 & 0.000 & 0.000 \\[1.5ex]
\multicolumn{12}{c}{Academic effort} \\[0.5ex]
0.041 & 0.274 & 0.770 & & 101.214 & 13.114 & 102.663 & & 0.000 & & 0.089 & 0.001 & 0.020 \\[1.5ex]
\multicolumn{12}{c}{Extracurricular activities} \\[0.5ex]
0.000 & 0.749 & 0.241 & & 71.217 & 45.216 & 41.322 & & 0.000 & & 0.000 & 0.000 & 0.000 \\[1.5ex]
\multicolumn{12}{c}{Future perception} \\[0.5ex]
0.056 & 0.996 & 0.157 & & 79.396 & 36.390 & 113.498 & & 0.000 & & 0.092 & 0.001 & 0.125 \\[1.5ex]
\multicolumn{12}{c}{Trouble at school} \\[0.5ex]
0.000 & 0.004 & 0.905 & & 258.743 & 29.292 & 62.370 & & 0.000 & & 0.010 & 0.000 & 0.000 \\[1.5ex]
\multicolumn{12}{c}{Smoking} \\[0.5ex]
0.000 & 0.489 & 0.012 & & 79.666 & 53.835 & 89.261 &  & 0.000 & & 0.000 & 0.000 & 0.000 \\[1.5ex]
\multicolumn{12}{c}{Drinking} \\[0.5ex]
0.000 & 0.699 & 0.619 & & 72.494 & 54.534 & 34.545 &  & 0.000 & & 0.006 & 0.000 & 0.000 \\[1.5ex]
\multicolumn{12}{c}{Risky behaviors} \\[0.5ex]
0.000 & 0.650 & 0.014 & & 151.985 & 33.638 & 70.135 &  & 0.000 & & 0.000 & 0.000 & 0.000 \\[1.5ex]
\multicolumn{12}{c}{Self-esteem} \\[0.5ex]
0.011 & 0.255 & 0.285 & & 94.684 & 29.424 & 137.644 &  & 0.000 & & 0.004 & 0.000 & 0.007 \\[1.5ex]
\multicolumn{12}{c}{Physical exercise} \\[0.5ex]
0.000 & 0.873 & 0.682 & & 116.840 & 20.395 & 213.806 &  & 0.000 & & 0.000 & 0.000 & 0.002 \\[1.5ex]
\multicolumn{12}{c}{Fighting} \\[0.5ex]
0.418 & 0.608 & 0.024 & & 67.563 & 16.853 & 79.309 &  & 0.000 &  & 0.270 & 0.575 & 0.017 \\ 
\bottomrule
\end{tabular}
\begin{tablenotes}[para,flushleft]
For the encompassing test, the columns labeled `$a$ vs. $b$' test the null hypothesis that the specification $a$ does not perform worse than the specification $b$; an integer $a$ means a quantile specification with $a$ equally spaced quantiles. The KP test refers to the $p$-value of the KP LM rank Wald test for weak identification \citep{kleibergen2006generalized}. The columns labeled `monotonicity test $p$‑value' report the $p$‑values for tests where the null hypothesis is that peer effects are uniform (Unif), increasing (Incr), or decreasing (Decr) across quantile levels \citep[see][for Wald‑type criteria for monotonicity tests]{kodde1986wald}.
\end{tablenotes}
\end{threeparttable}
\end{table}

Table~\ref{tab:appresult:test} contains additional test diagnostics. In the quantile specification, we use the quantiles of $\boldsymbol{x}$ and $\bar{\boldsymbol{x}}$ among friends and friends of friends computed at the levels $(0, 1/3, 2/3, 1)$, as instruments.\footnote{This procedure generates 102 instruments. This represents approximately $0.1\%$ of the sample size in every specification considered, suggesting that the many-instrument bias discussed in the literature \citep{bekker1994alternative,newey2004higher} is unlikely to be a concern. As a robustness check, we also employ an alternative approach based on cross-fitting that reduces the number of instruments to 10. These instruments correspond to the predicted quantiles of the outcome among friends at 10 quantile levels (see Supplemental Appendix~\ref{append:addhealth}). The estimation results reported in Supplemental Appendix Table~\ref{tab:append:addhealth} are broadly consistent with those presented in the main text.} To rule out weak instruments, Table~\ref{tab:appresult:test} reports the first-stage F-statistics and the LM rank Wald test \cite[see][]{kleibergen2006generalized}. These diagnostics indicate that the instruments are strong and that our estimates do not suffer from underidentification.

Secondly, to detect the potential non-monotonicity in peer effects, we perform monotonicity tests. Specifically, we test whether peer effects are uniform, increasing, or decreasing across quantile levels. For the increasing (respectively decreasing) test, the null hypothesis is that $\beta^1 \leq \beta^2 \leq \beta^3$ (respectively $\beta^1 \geq \beta^2 \geq \beta^3$). We test these inequality restrictions using a Wald-type criterion \citep[see][]{kodde1986wald}. By taking the maximum of the two $p$-values, the monotonicity hypothesis is rejected at the 5\% significance level, for all of the outcomes except for \textit{Future perception} and \textit{Fighting}. This suggests that the LiM and CES specifications may not be well suited for analyzing peer effects in the Add Health context.

\begin{table}[htbp]
\setlength{\tabcolsep}{5pt}
\footnotesize
\centering
\caption{Empirical Results}
\label{tab:appresult}
\begin{threeparttable}
\begin{tabular}{cccd{1}cd{1}cc}
\toprule
\multicolumn{3}{c}{Quantile} && \multicolumn{1}{c}{LiM} && \multicolumn{2}{c}{CES}\\
$\beta^1$ & $\beta^2$ & $\beta^3$ & & $\beta^\textup{LiM}$ &  & $\rho$ & $\beta^\textup{CES}$ \\
\midrule
\multicolumn{8}{c}{Academic achievements (GPA)}                 \\[0.5ex]
0.057   & 0.722   & 0.045   &  & 0.804   &  & 0.552   & 0.801   \\
(0.068) & (0.102) & (0.060) &  & (0.048) &  & (0.683) & (0.042) \\[1.5ex]
\multicolumn{8}{c}{Academic effort}                             \\[0.5ex]
0.173   & 0.379   & 0.086   &  & 0.640   &  & -2.126  & 0.479   \\
(0.052) & (0.110) & (0.078) &  & (0.115) &  & (2.917) & (0.126) \\[1.5ex]
\multicolumn{8}{c}{Extracurricular activities}                  \\[0.5ex]
0.059   & 0.687   & 0.053   &  & 0.805   &  & -0.171  & 0.699   \\
(0.061) & (0.070) & (0.033) &  & (0.053) &  & (0.397) & (0.035) \\[1.5ex]
\multicolumn{8}{c}{Future perception}                           \\[0.5ex]
0.177   & 0.327   & 0.110   &  & 0.590   &  & 0.680   & 0.576   \\
(0.054) & (0.124) & (0.070) &  & (0.061) &  & (0.639) & (0.059) \\[1.5ex]
\multicolumn{8}{c}{Trouble at school}                           \\[0.5ex]
0.065   & 0.549   & 0.096   &  & 0.780   &  & 0.335   & 0.835   \\
(0.080) & (0.097) & (0.064) &  & (0.109) &  & (0.361) & (0.099) \\[1.5ex]
\multicolumn{8}{c}{Smoking}                                     \\[0.5ex]
-0.121  & 0.721   & 0.153   &  & 0.783   &  & 1.588   & 0.684   \\
(0.070) & (0.089) & (0.028) &  & (0.051) &  & (0.667) & (0.097) \\[1.5ex]
\multicolumn{8}{c}{Drinking}                                    \\[0.5ex]
-0.009  & 0.443   & 0.124   &  & 0.573   &  & 0.460   & 0.638   \\
(0.074) & (0.090) & (0.019) &  & (0.087) &  & (0.387) & (0.109) \\[1.5ex]
\multicolumn{8}{c}{Risky behaviors}                             \\[0.5ex]
0.001   & 0.500   & 0.183   &  & 0.672   &  & 0.743   & 0.708   \\
(0.059) & (0.105) & (0.030) &  & (0.060) &  & (0.364) & (0.071) \\[1.5ex]
\multicolumn{8}{c}{Self-esteem}                                 \\[0.5ex]
0.147   & 0.483   & 0.032   &  & 0.602   &  & -1.626  & 0.428   \\
(0.069) & (0.154) & (0.085) &  & (0.125) &  & (2.481) & (0.094) \\[1.5ex]
\multicolumn{8}{c}{Physical exercise}                           \\[0.5ex]
0.164   & 0.368   & 0.080   &  & 0.611   &  & 0.873   & 0.596   \\
(0.058) & (0.079) & (0.053) &  & (0.121) &  & (0.356) & (0.076) \\[1.5ex]
\multicolumn{8}{c}{Fighting}                                    \\[0.5ex]
0.227   & 0.169   & 0.269   &  & 0.614   &  & 2.025   & 0.548   \\
(0.075) & (0.116) & (0.039) &  & (0.072) &  & (0.757) & (0.081)\\\bottomrule
\end{tabular}
\begin{tablenotes}[para,flushleft]
Estimates are reported without parentheses, with standard errors (clustered at the school level) shown in parentheses. The first row indicates the model used: quantile, LiM, or CES. The full table, including coefficients for the control variables, is available upon request.
\end{tablenotes}
\end{threeparttable}
\end{table}

Table~\ref{tab:appresult} contains the estimation results for the three-quantile specification, along with those for the LiM and the CES specifications, as comparison. The decomposition of peer effects across peer outcome quantiles using the three-quantile specification reveals several patterns. A prevalent pattern is one in which mid-level outcomes tend to be the most influential, while students appear less sensitive to peers at the lowest or highest ends of the outcome distribution. This pattern is observed for ten of the eleven outcomes considered. For most outcomes exhibiting this pattern, only peers around the median exert influence, whereas those at the top or bottom of the distribution exert little to no influence. This is the case, for example, for \textit{Academic achievement}, \textit{Extracurricular activities}, and \textit{Trouble at school}. For \textit{Smoking}, \textit{Drinking}, and \textit{Risky behaviors}, although peers around the median remain the most influential, peers at the top of the distribution also exert a significant influence. In contrast, for \textit{Self-esteem} and \textit{Physical exercise}, peers at the bottom of the distribution, in addition to those around the median, appear to be strongly influential.

The alternative specifications of the LiM and CES fail to identify which part of the peer‑outcome distribution is relevant for shaping individual behavior in this context, since they impose monotonicity.\footnote{Note that the estimates for the CES model differ from those reported by \cite{boucher2024toward}, as they report results excluding contextual effects.} Moreover, in all cases but \textit{Extracurricular activities}, the CES specification does not reject the null of uniform peer effects: $H_0: \rho = 1$. In the case of \textit{Extracurricular activities}, the substitution parameter $\rho$ is negative, with a standard error indicating that it is statistically different from one. This result suggests that peers with lower outcomes are more influential. Because influence is monotonic under this specification, one might conclude that peers at the bottom of the outcome distribution are the most influential. However, the quantile model provides a more nuanced picture: the most influential peers are those around the median, whereas peers with the lowest outcomes exert little to no influence.

A comparable identification issue arises for \textit{Self-esteem}, where the CES specification yields a negative substitution parameter ($\rho = -1.626$) that is not statistically different from one because of the large variance of the estimate. This imprecision indicates that the specification cannot clearly identify the relevant part of the peer-outcome distribution. Our quantile specification sheds light on this imprecision, showing that peers at neither the top nor the bottom of the distribution appear to be more influential than the median. 

By contrast, our estimation results also identify situations in which the assumption of monotonic peer effects is not rejected: \textit{Future perception} and \textit{Fighting}. For example, in the case of \textit{Fighting}, peers at the top and bottom of the distribution appear to be the most influential, although their effects are not statistically different from those of peers around the median. Consequently, for these outcomes, the conclusions drawn from the quantile specification are broadly consistent with those obtained under the CES specification.

Overall, our results are consistent with theories of social comparison, which suggest that aligning with moderately performing peers, rather than with outliers, may be socially or psychologically optimal \citep{festinger1954theory}. However, the mechanisms underlying these results may vary across outcomes. One explanation is that, because homophily in exogenous characteristics drives friendship formation, peers at the top or bottom of the distribution are more likely to differ substantially from the individual, whereas peers with intermediate values tend to be more similar and more closely connected. Moreover, friendships with high-achieving students may not improve learning if such peers are unwilling to collaborate with friends they perceive as weaker; in this case, these friendships may affect other outcomes rather than academic achievement. Our results are consistent with this interpretation, as we found no statistically significant effect of peers at the top of the distribution for the outcome academic achievement. A related pattern emerges for \textit{Extracurricular activities}, where highly involved peers are not necessarily more influential, as very high involvement may be negatively perceived in an academic context.

The differences in peer effects across quantiles suggest that not all peers shape individual behavior. These results carry significant policy implications, particularly for identifying key players within a network. Such identification may vary substantially across the three models, especially in cases where peers with mid‑level outcomes are the most influential.

\subsection{Measuring Student Influence}

This section studies the influence of students within their school. Influence is measured by the change in the school's average outcome when the student's outgoing and incoming links are removed. This corresponds to the change in the school's average outcome in a scenario in which the student becomes \textit{fully isolated}---that is, they have no friends and are not nominated by others. Students who are already fully isolated in the observed network have no influence, as their influence measure is zero. For students who are not fully isolated, removing their links can affect the outcome distribution at the game equilibrium because the peer set for those who nominated them will change. This measure of influence is also considered by \citet{ballester2006s} and \citet{lee2021key}, who define the key player as the student with the greatest influence.

We set the model parameters to their estimated values and compute the influence for each student. Within each school, we rank students by assigning the highest rank to the student with the largest influence. We then compare the rankings obtained from the quantile model to those from the LiM and CES models. Since the effect of removing a single student's links can be negligible in large networks, we focus on schools with fewer than 50 students.\footnote{This is about 6 schools, constituting 4\% of the available schools. In larger schools, a similar simulation exercise can be conducted by removing the links of a group of students rather than just one.}

Let $\mathbf{A}_g^{(i)}$ denote the network adjacency matrix obtained by setting the $i$-th row and column of $\mathbf{A}_g$ to zero. Let $y_{j,g}^{(i)}$ be the outcome of student $j$ when the school network is $\mathbf{A}_g^{(i)}$. The influence of student $i$ is measured by
$$
P_{i,g} = \dfrac{1}{n_g} \sum_{j = 1}^{n_g} \left(y_{j,g} - y_{j,g}^{(i)}\right).
$$

\begin{figure}[t!]
    \centering
    \includegraphics[scale = 1]{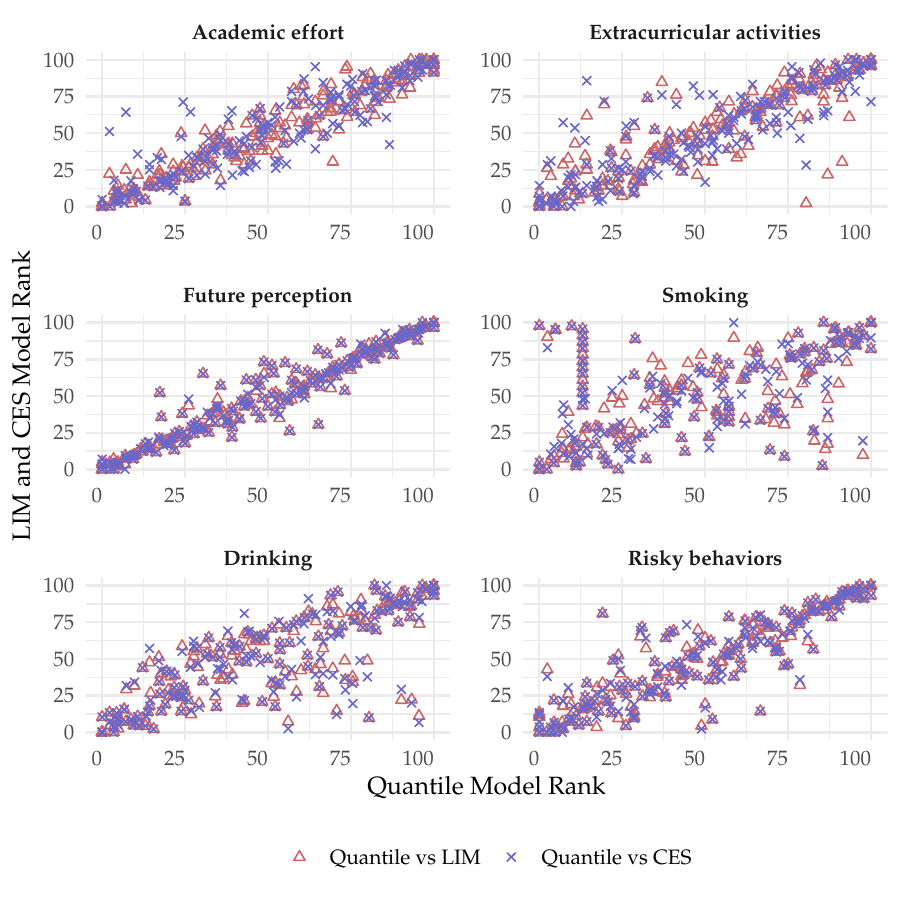}
    \caption{Influence Measure}
    \label{fig:counterfact}
    \vspace{-0.3cm}
    \footnotesize
    \justify
    The $x$-axis reports student ranks based on the influence measure in the quantile model, while the $y$-axis reports ranks based on the LiM and CES models. Each red triangle represents a student’s LiM model rank (on the $y$-axis) against their quantile model rank (on the $x$-axis). Each blue “$\times$” marker represents a student’s CES model rank (on the $y$-axis) versus their quantile model rank (on the $x$-axis).
\end{figure}

Student ranks within their school, normalized between 0 and 100, are presented in Figure~\ref{fig:counterfact} for selected outcomes.\footnote{In case of ties, we give all students in the tie the same rank.} Rank gaps are substantial, particularly for \textit{Extracurricular activities}, \textit{Smoking}, \textit{Drinking}, and \textit{Risky behaviors}. For these outcomes, some students who receive the highest influence scores under the CES and LiM models are assigned relatively low scores under the quantile model, and vice versa. These discrepancies reflect the non-monotonic pattern of peer effects associated with these outcomes. Under the quantile specification, peers at the top or bottom of the outcome distribution tend to receive lower ranks than under the other models, because the most influential peers are those with intermediate outcome levels. In contrast, the LiM and CES specifications cannot isolate the influence of mid-level outcome peers without simultaneously attributing substantial weight to peers at the extremes. As a result, students connected to peers with extreme outcome values are not assigned low ranks under these models; in the CES specification in particular, the most influential peers are necessarily located at the top or bottom of the peer-outcome distribution.

For the other outcomes, rank gaps are less pronounced because the pattern of peer effects is closer to a uniform structure. In particular, for \textit{Academic effort} and \textit{Future perception}, the uniformity assumption is not rejected at a 5\% significance level. Consistent with these results, the rank predictions are similar across the quantile, LiM, and CES specifications.

This counterfactual analysis reveals that key player status varies across the three models. In the standard LiM model, a high outcome for a student who is not fully isolated increases their influence \citep[see][]{ballester2006s}, because the student is nominated as a friend and therefore contributes substantially to the average friend outcome. Since this average is the only aggregated measure of the peer-outcome distribution that affects individuals, these peers are systematically influential. In contrast, in the quantile or CES models, a high outcome does not necessarily imply greater influence. Nevertheless, the CES model still imposes restrictions on whether these peers are influential or not. The analysis above shows that when peer effects are non-monotonic, the CES model does not fully overcome the limitations of the LiM model and may misidentify the most influential students and key players. Identifying such agents is crucial in social interactions, for example, to implement targeted policies or to assign individuals to groups, such as students in dormitories. The quantile model offers a more flexible approach for conducting such analyses.

\section{Conclusion}
This paper introduced a model that allows for flexible patterns of peer effects. In this model, how you affect your peer depends not only on your own outcome level but also on the outcomes of the other peers that they have. The construction of the peer effect in the model is flexible enough to allow us to discuss many interesting questions regarding how the composition of a peer group affects the peer effect: complementarity vs. substitutability, non-monotonic marginal peer effects, etc. Existing peer effect models, such as LiM and the CES model, tend to oversimplify the dynamics of peer influence by focusing on a single scalar parameter or assuming monotonically increasing or decreasing marginal peer effects. This oversimplification neglects the heterogeneity within peer effects, where individuals may experience varying degrees of influence depending on the structure and nature of their interactions, across different empirical contexts. This paper allows researchers to investigate richer patterns of peer interactions, with specific peers, such as the median performers, having disproportionately larger effects on outcomes in certain contexts.

As an operational framework, we propose a quantile-based peer effect model to solve the dimensionality problem. There may exist alternative approaches with better finite sample properties, which warrants future research. For example, the model can benefit from having an estimator that has data-driven smoothness on the saturated peer effect coefficients. 

In addition, this paper has primarily considered settings with exogenous networks. As work on incorporating endogenous networks into peer effect models continues to develop, it would be interesting to see if similar approaches can be applied to our model. 

\section*{Acknowledgments}

This research uses data from the National Longitudinal Study of Adolescent to Adult Health (Add Health), a program that is directed by Kathleen Mullan Harris and designed by J. Richard Udry, Peter S. Bearman, and Kathleen Mullan Harris at the University of North Carolina at Chapel Hill, and funded by Grant P01-HD31921 from the Eunice Kennedy Shriver National Institute of Child Health and Human Development, with cooperative funding from 23 other US federal agencies and foundations. Special acknowledgment is given to Ronald R. Rindfuss and Barbara Entwisle for assistance in the original design. Information on how to obtain Add Health data files is available on the Add Health website. No direct support was received from Grant P01-HD31921 for this research.

Generative AI tools, Claude (Sonnet 4.6, Opus 5.5) and ChatGPT (5.6, 6.1), have been used in preparing this manuscript. The authors have examined every theoretical and quantitative result presented in this manuscript and take sole responsibility for any error. 

\bibliographystyle{ecta}
\bibliography{mylit}

\appendix
\vspace{3em}

\begin{center}
    {\LARGE APPENDIX}
\end{center}

\section{Proofs}

\subsection{Proof for Proposition \ref{prop:uniqueNE}}\label{app:uniqueNE}
We want to show that the game described by the utility function \eqref{eq:utility} has a unique Nash equilibrium $\mathbf y^{\ast} = (y_1^{\ast}, ~\dots,   y_n^{\ast})^\intercal$ such that $y_i^{\ast} = BR_i(\mathbf{y}^{\ast}_{-i})$. To achieve this result, it is sufficient to show that the mapping $BR(\mathbf y) = (BR_1(\mathbf y_{-1}),~\dots,~BR_n(\mathbf y_{-n}))^\intercal$ is a contraction. In other words, there exists some constant $\bar\beta < 1$ such that for any two strategy profiles $\mathbf y, \mathbf y^{\prime}\in\mathbb R^n$, $\lvert BR_i(\mathbf y_{-i}) - BR_i(\mathbf y^{\prime}_{-i})\rvert \leq \bar\beta \lVert \mathbf y - \mathbf y^{\prime}\rVert_{\infty}$ for every $i=1, \cdots, n$. Let $\bar{\beta} = \max_{d \leq \bar{d}} \sum_{k=1}^d | \beta_{k,d} |$. From Assumption \ref{ass:equilibrium}, $0 \leq \bar{\beta} < 1$.

Let $\pi$ be the ordering of $\mathbf{y}$, that is, a function that maps $i$ to the rank of $y_i$ in $\mathbf{y}$ under some tie-breaking rule. For example if $y_i = \max_{k =1,...n} y_k$ then $\pi(i) = n$. Note that $\pi$, together with the network $\mathbf A$, therefore fully determines the mapping from $\mathbf{y}$ to each $\tilde{y}_{i,k}$. We define $\pi'$ equivalently as the ordering of $\mathbf{y}'$. 

Consider first the case where $\pi = \pi'$. The inequality holds trivially: for each $i$, 
\begin{align*}
    | BR_i(\mathbf y_{-i}) - BR_i({\mathbf{y}_{-i}}') | &= \left| \sum_{k=1}^{d_i} \left( \tilde{y}_{i,k} - {\tilde{y}_{i,k}}' \right) \beta_{k,d_i}\right| \\
    &\leq \sum_{k=1}^{d_i} | \beta_{k, d_i} | \cdot \left\|\mathbf y - \mathbf y' \right\|_\infty \leq \bar{\beta} \left\|\mathbf y - \mathbf y' \right\|_\infty
\end{align*}
with $\bar{\beta} = \max_d \sum_{k=1}^d | \beta_{k,d} | < 1$. The first inequality holds since $\mathbf y$ and $\mathbf y'$ have the same order and therefore $\tilde{y}_{i,k}$ and ${\tilde{y}_{i,k}}'$ are the outcomes of the same individual: $\big| \tilde{y}_{i,k} - {\tilde{y}_{i,k}}' \big|$ is bounded by $\|\mathbf y - \mathbf y' \|_\infty$. The second inequality is from Assumption \ref{ass:equilibrium}. 

The non-trivial case is when $\pi \neq \pi'$. To proceed, fix $i$ and $k$ and let $j$ be the identity of person $i$'s $k$-th lowest friend under $\pi$, and $j'$ be the identity under $\pi'$. Consider the object 
\begin{align*}
    \left| \left( \tilde{y}_{i,k} - {\tilde{y}_{i,k}}' \right) \beta_{k,d_i}\right|.
\end{align*}
If $j = j'$, this quantity is bounded by $|\beta_{k,d_i}| ||\mathbf{y} - \mathbf{y}'||_\infty$ by a similar argument as above. If not, we have 
\begin{align*}
    \left| \left( \tilde{y}_{i,k} - {\tilde{y}_{i,k}}' \right) \beta_{k,d_i}\right| \leq |\beta_{k,d_i} | | y_j - y_{j'}' | .
\end{align*}
Let $k'$ be the ranking of $j$ among $i$'s friends under $\pi'$. Then there are two possible cases, either $k' > k$ or $k' < k$. 

\textbf{Case 1:} $k' > k$. Then we have that 
\begin{align*}
    y'_j = \tilde{y}_{i,k'}' \geq \tilde{y}_{i,k}' = y'_{j'}
\end{align*}
meaning $y'_j \geq y'_{j'}$. Therefore 
\begin{align*}
    y_j - y'_{j'} &\geq y'_{j} - ||\mathbf{y} - \mathbf{y'}||_\infty    - y'_{j'}
   \geq - ||\mathbf{y} - \mathbf{y'}||_\infty 
\end{align*}
where the first inequality uses $y'_{j} - y_j \leq || \mathbf{y} - \mathbf{y'}||_\infty $ and the second that $y'_j - y'_{j'} \geq 0 $.

To get the inequality for the other sign, note that, by the pigeonhole principle, there must be another friend of $i$, $p$, whose rank among $i$'s friends was above $k$ under $\pi$ and is lower than or equal to $k$ under $\pi'$. That is 
\begin{align*}
    y'_p \leq \tilde{y}'_{i,k} = y'_{j'} \quad \text{and} \quad y_p \geq \tilde{y}_{i,k} = y_j,
\end{align*}
and therefore: 
\begin{align*}
    y_j - y_{j'}' &\leq y_p - y'_{j'} \\ 
    &\leq y'_{p} + ||\mathbf{y} - \mathbf{y'} ||_\infty - y_{j'}' \\
    &\leq ||\mathbf{y} - \mathbf{y'} ||_\infty
\end{align*}
Combining this with the result above gives the desired inequality: $|y_j - y'_{j'} | \leq ||\mathbf{y} - \mathbf{y'}||_\infty $

\textbf{Case 2:} $k' < k$. As before, this implies 
\begin{align*}
    y_j' = \tilde{y}'_{i,k'} \leq \tilde{y}'_{i,k} = y'_{j'}
\end{align*}
and therefore 
\begin{align*}
    y_{j} - y'_{j'} \leq y'_j + ||\mathbf{y} - \mathbf{y'} ||_\infty - y'_{j'} \leq || \mathbf{y} - \mathbf{y'}||_\infty. 
\end{align*}
By the pigeonhole principle, we know that as $k' <k$, there must be some other individual $p$ whose rank is higher than or equal to $k$ under $\pi'$ but lower than $k$ under $\pi$. That is 
\begin{align*}
    y'_{p} \geq \tilde{y}'_{i,k} = y'_{j'} \quad \text{ and } \quad y_p \leq \tilde{y}_{i,k} = y_j 
\end{align*}
and therefore 
\begin{align*}
    y_j - y'_{j'} \geq y_p - y'_{j'} \geq y_{p}' - ||\mathbf{y} - \mathbf{y}' ||_\infty - y'_{j'} \geq - ||\mathbf{y} - \mathbf{y'}||_\infty.
\end{align*}
Combining the two inequalities we have our desired result: $| y_j - y'_{j'} | \leq ||\mathbf{y} - \mathbf{y'}||_\infty$.

Finally, by aggregating across $k$ for the fixed $i$ and then aggregating across $i$, we get 
\begin{align*}
    \left| BR_i(\mathbf{y}_{-i}) - BR_i(\mathbf{y}_{-i}') \right| &= \left| \sum_{k=1}^{d_i} \left( \tilde{y}_{i,k} - {\tilde{y}_{i,k}}' \right) \beta_{k,d_i}\right| \leq \sum_{k=1}^{d_i} | \beta_{k, d_i} | \cdot \left\| \mathbf{y} - \mathbf{y}' \right\|_\infty \leq \bar{\beta} \left\| \mathbf{y} - \mathbf{y}' \right\|_\infty  \\
    \left\| BR(\mathbf{y}) - BR(\mathbf{y}') \right\|_\infty &\leq \bar{\beta} \left\| \mathbf{y} - \mathbf{y}' \right\|_\infty. 
\end{align*}
Therefore, $BR(\mathbf{y})$ is a contraction mapping for any $\mathbf{y},\mathbf{y'}$.  

\subsection{Proof for Theorem \ref{theorem:id}}\label{app:theoremidproof} 
From Assumption \ref{ass:exogeneity}, we have
$$
\E[\mathbf{y}_g] = \E[\mathbf{W}_g]\theta.
$$
Premultiplying by $\E[\mathbf{W}_g]^\intercal$ and averaging over $g$ implies that
$$
\plim_{m \to \infty} \frac{1}{m} \sum_{g=1}^m \E[\mathbf{W}_g]^\intercal \E[\mathbf{W}_g]\theta = \plim_{m \to \infty} \frac{1}{m} \sum_{g=1}^m \E[\mathbf{W}_g]^\intercal \E[ \mathbf{y}_g ].
$$
Under Assumption \ref{ass:id}, the matrix $\plim_{m\to\infty}\frac{1}{m}\sum_{g=1}^m \E[\mathbf{W}_g]^\intercal \E[\mathbf{W}_g]$
is nonsingular. Thus, $\theta $ is uniquely determined, and identification follows.

\subsection{Proof for Proposition \ref{prop:misspecification}}
\label{app:missspecproof}
The following proof is written for a scalar $\boldsymbol{z}_{i,g}$. When $\boldsymbol{z}_{i,g}$ is not a scalar, we can modify the argument below and get similar expressions by replacing $\tilde{\boldsymbol{z}}_{i,g}$ with $c^\intercal \tilde{\boldsymbol{z}}_{i,g}$ with some weighting $c$. The weights $w_{k,d}^{\textup{LiS}}$ and $w_{k,d}^{\textup{LiM}}$ then depend on how the over-identifying moments are weighted.

Using Assumption \ref{ass:miss}-\ref{ass:lim} and plugging in
\begin{align*}
\E \left[ y_{i,g} \right] &= \sum_{k=1}^{d_{i,g}} \beta_{k,d_{i,g}} \E \left[ \tilde{y}_{i,k,g} \right] + {\boldsymbol{x}_{i,g}}^\intercal \gamma + {\bar{\boldsymbol{x}}_{i,g}}^\intercal \delta
\end{align*}
into the LiM estimand and re-arranging, we get 
$$
\beta^{\text{LiM}} = \sum_{d=1}^{\bar{d}} \sum_{k=1}^d w_{k,d}^\text{LiM} d \beta_{k,d} + o_p(1)
$$
with 
\begin{align*}
w_{k,d}^{\text{LiM}} &= \frac{ \frac{1}{m}\sum_{g=1}^m \sum_{i=1}^{n_g} \frac{1}{d} \tilde{\boldsymbol{z}}_{i,g}^{\mathrm{LiM}} \E \left[ \tilde{y}_{i,k,g} \mathbf{1}\{d_{i,g}=d\} \right] }{\sum_{d=1}^{\bar{d}} \sum_{k=1}^{d} \frac{1}{m}\sum_{g=1}^m \sum_{i=1}^{n_g} \frac{1}{d} \tilde{\boldsymbol{z}}_{i,g}^{\mathrm{LiM}} \E \left[ \tilde{y}_{i,k,g} \mathbf{1}\{d_{i,g}=d\} \right]}.
\end{align*}
The weights satisfy the sum-to-one constraint: $\sum_{d=1}^{\bar{d}} \sum_{k=1}^d w_{k,d}^\text{LiM} = 1$. 

Likewise, for the LiS model, we get
\begin{align*}
\beta^{\text{LiS}} 
&= \sum_{d=1}^{\bar{d}} \sum_{k=1}^d w_{k,d}^{\text{LiS}} \beta_{k,d} + o_p(1),
\end{align*}
where
$$
w_{k,d}^{\text{LiS}} = \frac{\frac{1}{m}\sum_{g=1}^m \sum_{i=1}^{n_g} \tilde{\boldsymbol{z}}_{i,g}^{\text{LiS}} \E \left[ \tilde{y}_{i,k,g} \mathbf{1}\{d_{i,g}=d\}\right]}{\sum_{d=1}^{\bar{d}} \sum_{k=1}^{d}\frac{1}{m}\sum_{g=1}^m \sum_{i=1}^{n_g} \tilde{\boldsymbol{z}}_{i,g}^{\text{LiS}} \E \left[ \tilde{y}_{i,k,g} \mathbf{1}\{d_{i,g}=d\}\right]}
$$
and $\tilde{\boldsymbol{z}}_{i,g}^{\text{LiS}}$ is the residual from the linear-in-sums projection $\mathbf{P}^{\mathrm{LiS}}$. Again, the weights sum to one: $\sum_{d=1}^{\bar{d}} \sum_{k=1}^d w_{k,d}^{\text{LiS}} = 1$.

The specifics of the weights rearranging are provided in Supplemental Appendix \ref{app:missspecproofdetails}.

\newpage 
\clearpage 

\begin{center}
    {\Large \bf Supplemental Appendix  to ``Identification of Heterogeneous Peer Effects''}\\[1em]
    {\large Eyo I. Herstad, Aristide Houndetoungan, and Myungkou Shin}
\end{center}

\setcounter{page}{1}
\setcounter{section}{0}
\renewcommand{\thesection}{S\Alph{section}}
\renewcommand{\theHsection}{supp.\Alph{section}}
\renewcommand{\theHtable}{supp.\thetable}
\renewcommand{\theHfigure}{supp.\thefigure}

\section{Additional Theoretical Results}\label{app:supp_theory}

\subsection{Spillover and Conformity in Preferences}\label{app:conformity}
We adopt the approach that was proposed by \cite{boucher2024toward}, which incorporates both spillover and conformity into the model, allowing their relative importance to be empirically identified from the data. The utility function is given by:
\begin{equation}\label{eq:utility:conf}
    U_i(y_i, ~ \mathbf y_{-i}) = \underbrace{\alpha_i y_i - \frac{1}{2}y_i^2}_{\text{private benefit}} + \underbrace{\sum_{k = 1}^{d_i}\beta_{k,d_i}^s y_i\tilde y_{i,k} - \frac{1}{2}\sum_{k = 1}^{d_i}\beta_{k,d_i}^c(y_i - \tilde y_{i,k})^2}_{\text{social benefit}}.
\end{equation}

The social benefit reflects both spillover effects and conformity between an agent’s effort and the efforts of their peers.\footnote{When $i$ has no friends (i.e., no outgoing links or $d_i = 0$), the utility function excludes the social benefit term and is given by $U_i(y_i, \mathbf{y}_{-i}) = \alpha_i y_i - \frac{1}{2} y_i^2$. The corresponding optimal strategy is therefore $y_i = \alpha_i$. An individual without friends is referred to as \textit{isolated}.} The term $\beta_{k,d_i}^s y_i\tilde y_{i,k}$ implies strategic complementarity between agent $i$’s effort and the effort of the peer with the $k$-th smallest effort when $\beta_{k,d_i}^s > 0$. Conversely, if $\beta_{k,d_i}^s < 0$, then $\beta_{k,d_i}^s y_i\tilde y_{i,k}$ reflects strategic substitution (i.e., negative spillover) between efforts. 
The term $\beta_{k,d_i}^c(y_i - \tilde y_{i,k})^2$ captures conformist preferences. If $\beta_{k,d_i}^c$ is nonnegative, agent $i$ incurs a social cost when their effort deviates from the $k$-th smallest effort among their peers. The case where $\beta_{k,d_i}^c$ is negative (indicating anticonformity) is rare and is not considered in this paper, as it introduces complications in the analysis.\footnote{Specifically, when $\beta_{k,d_i}^c < 0$, the utility function may become convex and lack a maximizer unless the strategy space is compact.}

For ease of exposition, we introduce the following notation:
\begin{align*}
    &\lambda_{k,d_i}^s = \dfrac{\beta_{k,d_i}^s}{1 + \sum_{k = 1}^{d_i} \beta_{k,d_i}^c}, \quad \lambda_{k,d_i}^c = \dfrac{\beta_{k,d_i}^c}{1 + \sum_{k = 1}^{d_i} \beta_{k,d_i}^c}, \quad \text{and} \quad
    \lambda_{k,d_i} = \lambda_{k,d_i}^s + \lambda_{k, d_i}^c \text{ for all } k.\\
    &\lambda_{d_i}^s = \sum_{k = 1}^{d_i}\lambda_{k,d_i}^s\quad \text{and} \quad \lambda_{d_i}^c = \sum_{k = 1}^{d_i}\lambda_{k,d_i}^c.
\end{align*}
Agent $i$ chooses their effort $y_i$ by maximizing the utility function $U_i(y_i, \mathbf y_{-i})$. By solving the first-order condition of this maximization problem, we obtain the best response function:
\begin{equation} \label{BRF:conf}
    BR_i(\mathbf y_{-i}) = (1 - \lambda_{d_i}^c)\alpha_i + \sum_{k = 1}^{d_i} \lambda_{k,d_i} \tilde y_{i,k}.
\end{equation}

The parameter $\lambda_{k,d_i}$ measures the total peer effect at the $k$-th smallest peer effort. This total effect is decomposed as $\lambda_{k,d_i} = \lambda_{k,d_i}^s + \lambda_{k,d_i}^c$, where $\lambda_{k,d_i}^s$ captures the spillover effect and $\lambda_{k,d_i}^c$ captures the conformity component. The total spillover effect across all peer efforts is given by $\lambda_{d_i}^s = \sum_{k = 1}^{d_i} \lambda_{k, d_i}^s$, while the total conformity parameter is $\lambda_{d_i}^c = \sum_{k = 1}^{d_i} \lambda_{k, d_i}^c$. 

By defining $\tilde{\alpha}_i = (1 - \lambda_{d_i}^c)\alpha_i$, the best response function given by Equation~\eqref{BRF:conf} is similar to that in Equation~\eqref{BRF}, where $\tilde{\alpha}_i$ represents the new productivity level. Therefore, Proposition~\ref{prop:uniqueNE}, which establishes the existence and uniqueness of the Nash equilibrium, readily extends to the game featuring both spillover and conformity under a condition analogous to Assumption~\ref{ass:equilibrium}, namely that $\sum_{k = 1}^{d_i} \lvert \lambda_{k,d_i} \rvert < 1$.

\subsection{Identification Under Spillover and Conformity}

In this section, we discuss the identification of the model in the case where it includes both spillover and conformity effects. As $\alpha_{i,g} = \boldsymbol{x}_{i,g}^{\intercal} \gamma + \bar{\boldsymbol{x}}_{i,g}^\intercal \delta + \varepsilon_{i,g}$, the reduced-form specification of the model for a non-isolated individual (i.e., one who has friends) is:
\begin{equation*}
    y_{i,g} = \sum_{k = 1}^{d_{i,g}} \lambda_{k,d_{i,g}} \tilde y_{i,k,g} + (1 - \lambda^c_{d_{i,g}}) {\boldsymbol{x}_{i,g}}^{\intercal}\gamma + (1 - \lambda^c_{d_{i,g}}) {\bar{\boldsymbol{x}}_{i,g}}^{\intercal} \delta + (1 - \lambda^c_{d_{i,g}}) \varepsilon_{i,g}.
\end{equation*}
For isolated agents, the outcome is simply given by:
\begin{equation*}
    y_{i,g} = {\boldsymbol{x}_{i,g}}^\intercal \gamma + \varepsilon_{i,g}.
\end{equation*}

As in \cite{boucher2024toward}, identifying the model's parameters requires the presence of isolated nodes. Given that our inference assumes many independent networks, this implies that the share of networks with isolated nodes is strictly positive asymptotically: \begin{align}
\plim_{m \to \infty} \frac{1}{m} \sum_{g=1}^m \frac{1}{n_g} \sum_{i=1}^{n_g} \mathbf{1}\{d_{i,g}=0\} \boldsymbol{x}_{i,g} {\boldsymbol{x}_{i,g}}^\intercal \label{eq:isolated_id}
\end{align}
is nonsingular. Under the exogeneity condition in Assumption \ref{ass:exogeneity} and the nonsingularity condition on the probability limit \eqref{eq:isolated_id}, $\gamma$ is identified:
$$
\gamma = \left(\plim_{m \to \infty} \frac{1}{m} \sum_{g=1}^m \frac{1}{n_g} \sum_{i=1}^{n_g} \mathbf{1}\{d_{i,g}=0\} \boldsymbol{x}_{i,g} {\boldsymbol{x}_{i,g}}^\intercal \right)^{-1}
\left(\plim_{m \to \infty} \frac{1}{m} \sum_{g=1}^m \frac{1}{n_g} \sum_{i=1}^{n_g} \mathbf{1}\{d_{i,g}=0\} \boldsymbol{x}_{i,g}  y_{i,g}\right).
$$

Given the identification of $\gamma$, we can also identify $\lambda^c_{d_i}$ and $\lambda_{k,d_{i,g}}$ for all $k$ and $d_{i,g}$ using the moment conditions for non-isolated agents. Indeed, this identification problem is similar to the one studied in the main text. The only difference is that it focuses on non-isolated agents.

Overall, we can then identify $\lambda_{k,d_{i,g}}$, which represent the total peer effect at each rank, and the total conformity effect $\lambda^c_{d_{i,g}}$. The total spillover effect is therefore given by $\lambda^s_{d_{i,g}} = \sum_{k = 1}^{d_{i,g}} \lambda_{k,d_{i,g}} - \lambda^c_{d_{i,g}}$. However, we cannot disentangle spillover and conformity effects at each rank separately; we can only identify their total effects across all ranks.

\subsection{Consistency and Asymptotic Normality}\label{append:inference}

Let $\mathbf{Z}_g$ be the matrix of instruments in the $g$-th network. As discussed in Section \ref{sec:ident}, $\mathbf{Z}_g$ includes the quantiles of $\boldsymbol{x}_{j,g}$ and ${\bar{\boldsymbol{x}}_{j,g}}$ among friends, which we use to instrument $\tilde{\boldsymbol{y}}_{i,g}$. As usual, $\mathbf{Z}_g$ also includes exogenous regressors $\boldsymbol{x}_{i,g}$ and ${\bar{\boldsymbol{x}}_{i,g}}$. To establish consistency and asymptotic normality, we impose the following regularity conditions.
\begin{assumption}\label{ass:asymptotics} \quad
\begin{enumerate}[label=\textbf{\alph*.}]
\item \label{ass:asymptotics:full_rank} The matrices $\displaystyle \plim_{m \to \infty} \frac{1}{m} \sum_{g=1}^m {\mathbf{W}_g}^\intercal \mathbf{Z}_g$ and $\displaystyle \plim_{m \to \infty} \frac{1}{m} \sum_{g=1}^m {\mathbf{Z}_g}^\intercal \mathbf{Z}_g$ are of full row rank.  
\item $\displaystyle \plim_{m \to \infty} \frac{1}{m} \sum_{g=1}^m {\mathbf{Z}_g}^\intercal \boldsymbol{\varepsilon}_g = 0$. 
\item \label{ass:asymptotics:asymptotic_normality} $\displaystyle \frac{1}{\sqrt{m}} \sum_{g=1}^m {\mathbf{Z}_g}^\intercal \boldsymbol{\varepsilon}_g \xrightarrow{d} \mathcal{N} \left( 0, \boldsymbol{\Omega}_{\varepsilon} \right)$. 
\end{enumerate}
\end{assumption}

We establish the following result.

\begin{proposition}\label{prop:TSLS}
Under Assumption \ref{ass:asymptotics}, the TSLS estimator $\hat{\theta}^{\mathrm{TSLS}}$ using $\{\mathbf{Z}_g\}_{g=1}^m$ as instruments converges in probability to $\theta$ and \[
\sqrt{m} \left( \hat{\theta}^{\mathrm{TSLS}} - \theta \right) \xrightarrow{d} \mathcal{N} \left( 0, \boldsymbol{\Omega} \right)
\] 
with some consistently estimable covariance matrix $\boldsymbol{\Omega}$, as $m \to \infty$. 
\end{proposition}
\begin{proof}
Define the sample moments
$$
\hat{\mathbf{Q}}_{WZ} = \frac{1}{m}\sum_{g=1}^m \mathbf{W}_g^\intercal \mathbf{Z}_g,
\quad
\hat{\mathbf{Q}}_{ZZ} = \frac{1}{m}\sum_{g=1}^m \mathbf{Z}_g^\intercal \mathbf{Z}_g,
\quad
\hat{\mathbf{Q}}_{ZY} = \frac{1}{m}\sum_{g=1}^m \mathbf{Z}_g^\intercal \boldsymbol{y}_g.
$$
The TSLS estimator is
$$
\hat{\theta}^{\mathrm{TSLS}}
=
\left( \hat{\mathbf{Q}}_{WZ} \hat{\mathbf{Q}}_{ZZ}^{-1} \hat{\mathbf{Q}}_{ZW} \right)^{-1}
\left( \hat{\mathbf{Q}}_{WZ} \hat{\mathbf{Q}}_{ZZ}^{-1} \hat{\mathbf{Q}}_{ZY} \right),
$$
where $\hat{\mathbf{Q}}_{ZW} = \hat{\mathbf{Q}}_{WZ}^\intercal$. Since $\boldsymbol{y}_g = \mathbf{W}_g \theta + \boldsymbol{\varepsilon}_g$, we have
$$
\hat{\theta}^{\mathrm{TSLS}} - \theta
=
\left( \hat{\mathbf{Q}}_{WZ} \hat{\mathbf{Q}}_{ZZ}^{-1} \hat{\mathbf{Q}}_{ZW} \right)^{-1}
 \hat{\mathbf{Q}}_{WZ} \hat{\mathbf{Q}}_{ZZ}^{-1} 
\left( \frac{1}{m}\sum_{g=1}^m \mathbf{Z}_g^\intercal \boldsymbol{\varepsilon}_g \right).
$$
Under Part \ref{ass:asymptotics:full_rank} of Assumption \ref{ass:asymptotics}, $\plim \hat{\mathbf{Q}}_{WZ} = \mathbf{Q}_{WZ}$,~ and ~ $\plim \hat{\mathbf{Q}}_{ZZ} =\mathbf{Q}_{ZZ}$ with $\mathbf{Q}_{ZZ}$ not singular. Consequently, $\plim \hat{\mathbf{Q}}_{WZ} \hat{\mathbf{Q}}_{ZZ}^{-1} \hat{\mathbf{Q}}_{ZW} = \mathbf{Q}_{WZ} \mathbf{Q}_{ZZ}^{-1} \mathbf{Q}_{ZW}$. By Part b of Assumption \ref{ass:asymptotics}, $\plim \frac{1}{m}\sum_{g=1}^m \mathbf{Z}_g^\intercal \boldsymbol{\varepsilon}_g = 0$. Hence $\hat{\theta}^{\mathrm{TSLS}}$ converges in probability to $\theta$, as $m\to \infty$.

Moreover,
$$
\sqrt{m}(\hat{\theta}^{\mathrm{TSLS}} - \theta)
=
\left( \hat{\mathbf{Q}}_{WZ} \hat{\mathbf{Q}}_{ZZ}^{-1} \hat{\mathbf{Q}}_{ZW} \right)^{-1}
\hat{\mathbf{Q}}_{WZ} \hat{\mathbf{Q}}_{ZZ}^{-1}
\left( \frac{1}{\sqrt{m}}\sum_{g=1}^m \mathbf{Z}_g^\intercal \boldsymbol{\varepsilon}_g \right).
$$
Using Parts \ref{ass:asymptotics:full_rank}-\ref{ass:asymptotics:asymptotic_normality} of Assumption \ref{ass:asymptotics},
$\sqrt{m}(\hat{\theta}^{\mathrm{TSLS}} - \theta)
\xrightarrow{d} \mathcal{N}(0,\boldsymbol{\Omega})
$, where 
$$
\boldsymbol{\Omega}
=
\left( \mathbf{Q}_{WZ}\mathbf{Q}_{ZZ}^{-1}\mathbf{Q}_{ZW} \right)^{-1}
\mathbf{Q}_{WZ}\mathbf{Q}_{ZZ}^{-1}
\boldsymbol{\Omega}_{\varepsilon}
\mathbf{Q}_{ZZ}^{-1}\mathbf{Q}_{ZW}
\left( \mathbf{Q}_{WZ}\mathbf{Q}_{ZZ}^{-1}\mathbf{Q}_{ZW} \right)^{-1}.
$$\end{proof}

\subsection{Reduced-form Representation of the Model}\label{app:reduced_form}

To illustrate why it is hard to find lower-level conditions for Part~\ref{ass:id:nonsingular} of Assumption~\ref{ass:id}, such as those found for the LiM model, as well as to motivate our choice of instruments, we will next derive the reduced form of our model. To begin, note that  Proposition \ref{prop:uniqueNE} guarantees there is a unique equilibrium $\mathbf{y}_g$ given a realization of $\boldsymbol{\varepsilon}_g$. Let $\pi_g$ denote an ordering on individuals in network $g$, $\{1,\ldots, n_g\}$, in terms of their outcomes. For example, $\pi_g(1)=3$ means that individual $1$'s outcome $y_{1,g}$ is the third smallest among $\mathbf{y}_g$. Let $\mathbb{B}_g(\pi_g)$ denote an $n_g \times n_g$ peer effect coefficient matrix such that its $i$-th row $j$-th column component corresponds to the peer effect coefficient for individual $j$'s outcome on individual $i$'s outcome in network $g$, given the ordering $\pi_g$. Then, Proposition \ref{prop:uniqueNE} implies that there is a unique $\pi_g$ such that 
$$
\mathbf{y}_g = 
\mathbb{B}_g(\pi_g) \mathbf{y}_g + \mathbf{X}_g \gamma + \mathbf{\bar{X}}_g\delta  + \boldsymbol{\epsilon}_g
$$
where $\boldsymbol{\epsilon}_g = \begin{pmatrix} \epsilon_{1,g} & \cdots & \epsilon_{n_g,g} \end{pmatrix}^\intercal$ denotes a realization of $\boldsymbol{\varepsilon}_g$, and therefore 
$$
\mathbf{y}_g = \left( I_{n_g} - \mathbb{B}_g(\pi_g) \right)^{-1} \left( \mathbf{X}_g \gamma + \bar{\mathbf{X}}_g \delta + \boldsymbol{\epsilon}_g \right). 
$$
By taking expectation over the set of every possible ordering, $\Pi_{n_g}$, we obtain a reduced-form representation of $\E [\mathbf{y}_g]$. The following corollary gives the result of this procedure.

\begin{corollary}\label{cor:reduced_form}
Let Assumptions \ref{ass:equilibrium} and \ref{ass:exogeneity} hold and let $\textup{Pr}_g \{ \pi \}$ denote the conditional probability of an ordering $\pi$ being chosen as an equilibrium, given $\mathbf{X}_g$ and $\mathbf{A}_g$. Then, the peer effects model \eqref{eq:peer_effect_model} admits a reduced-form representation with residual terms $\eta_{i,g}$: $$
\E \left[ y_{i,g}  \right] = \sum_{j=1}^{n_g} \theta_{i,j,g} ({\boldsymbol{x}_{j,g}}^\intercal \gamma + \boldsymbol{\bar{x}}_{j,g}^\intercal\delta ) + \eta_{i,g}
$$
where \begin{align*}
\theta_{i,j,g} &= \sum_{\pi \in \Pi_{n_g}} \theta_{i,j,g} (\pi) \textup{Pr}_g \left\lbrace \pi \right\rbrace \\ 
\eta_{i,g} &= \sum_{\pi \in \Pi_{n_g}} \sum_{j=1}^{n_g} \theta_{i,j,g} (\pi) \E \left[ \varepsilon_{j,g} | \pi \right] \cdot \textup{Pr}_g \left\lbrace \pi \right\rbrace
\end{align*}
and $\Big( \theta_{i,1,g}(\pi), \ldots, \theta_{i,n_g,g}(\pi) \Big)$ is the $i$-th row of the $n_g \times n_g$ matrix $\left( I_{n_g} - \mathbb{B}_g(\pi) \right)^{-1}$. Moreover, there also exists a reduced-form relationship between the ordered peer outcome $\tilde{y}_{i,k}$ and $\{\boldsymbol{x}_{j,g}\}_{j=1}^n$, with residual terms $\{\tilde{\eta}_{j,k,g}\}_{j=1}^{n_g}$: \begin{align*}
    \E \left[ \tilde{y}_{i,k,g}  \right] &= \sum_{j=1}^{n_g} \tilde{\theta}_{i,k,j,g}  ({\boldsymbol{x}_{j,g}}^\intercal \gamma + {\boldsymbol{\bar{x}}_{j,g}}^\intercal\delta )+ \tilde{\eta}_{i,k,g}
\end{align*}
where \begin{align*}
\tilde{\theta}_{i,k,j,g} &= \sum_{\pi \in \Pi_{n_g}} \tilde{\theta}_{i,k,j,g} (\pi) \textup{Pr}_g \left\lbrace \pi \right\rbrace \\
\tilde{\eta}_{i,k,g} &= \sum_{\pi \in \Pi_{n_g}} \sum_{j=1}^{n_g} \tilde{\theta}_{i,k,j,g} (\pi) \E \left[ \varepsilon_{j,g} | \pi \right] \cdot \textup{Pr}_g \left\lbrace \pi \right\rbrace. 
\end{align*} 
and $\Big( \tilde{\theta}_{i,k,1,g}(\pi), \ldots, \tilde{\theta}_{i,k,n_g,g}(\pi) \Big)$ is one row of the $n_g \times n_g$ matrix $\left( I_{n_g} - \mathbb{B}_g (\pi) \right)^{-1}$, which corresponds to the $k$-th lowest performing peer of individual $i$. 
\end{corollary}

\begin{proof}
Let $\pi$ denote an ordering on $\{1,\dots,n_g\}$ with some tie-breaking rule, where the ordering comes from $\mathbf{y}_g$: $$
\pi: \{1,\ldots, n_g\} \mapsto \{1,\ldots,n_g\}
$$ 
and where $\pi(i)$ is the rank of $y_{i,g}$ in $\mathbf{y}_g$, so that $\pi(i) < \pi(j)$ whenever $y_{i,g} < y_{j,g}$. Note that $\pi$ is a function of $\boldsymbol{\varepsilon}_g$; from Proposition \ref{prop:uniqueNE}, we have shown that there exists a unique equilibrium and thus a unique ordering for each realization $\boldsymbol{\epsilon}_g$. Let $\Pi_{n_g}$ denote the set of all possible orderings on $\{y_{1,g}, \ldots, y_{n_g,g}\}$. $\Pi_{n_g}$ is a function of $\mathbf{X}_g, \mathbf{A}_g$ and the distribution of $\boldsymbol{\varepsilon}_g$. 

For any realization of $\boldsymbol{\epsilon}_g$ and the corresponding ordering $\pi$, construct an $n_g \times n_g$ matrix $\mathbb{B}_g(\pi)$ such that \begin{align*}
\mathbb{B}_g(\pi) &= \Big( \beta_{i,j,g} (\pi) \Big)_{i,j} \\
\beta_{i,j,g}(\pi) &= \begin{cases}
    0 & \text{if } a_{i,j,g} = 0 \\
    \beta_{k, d_{i,g}} \text{ with } k = \sum_{j'=1}^{n_g} a_{i,j',g} \mathbf{1}\{\pi(j') \leq \pi(j)\} & \text{if } a_{i,j,g} = 1
\end{cases}
\end{align*}
$\mathbb{B}_g(\pi)$ takes the ordering $\pi$ as fixed and finds the corresponding rank-dependent coefficient $\beta_{k,d}$ for each of the peer outcomes. To show that $(I_{n_g} - \mathbb{B}_g(\pi))$ is invertible for any $\pi$: suppose that there is some nonzero $\boldsymbol{v} \in \mathbb{R}^{n_g}$ such that $\boldsymbol{v} = \mathbb{B}_g(\pi) \boldsymbol{v}$. Then, for some $i$ such that $|v_i| = \| \boldsymbol{v}\|_\infty$, 
$$
| v_i | = \left| \sum_{j=1}^{n_g} \beta_{i,j,g} (\pi) v_j \right| \leq \sum_{j=1}^{n_g} | \beta_{i,j,g} (\pi) | \cdot | v_j | \leq \bar{\beta} | v_i | < | v_i |,
$$
leading to a contradiction; $I_{n_g} - \mathbb{B}_g(\pi)$ is full rank. Then, 
\begin{align*}
\mathbf{y}_g &= \mathbb{B}_g(\pi) \mathbf{y}_g + \mathbf{X}_g \gamma + \bar{\mathbf{X}}_g \delta + \boldsymbol{\varepsilon}_g \\
\mathbf{y}_g &= \left( I_{n_g} - \mathbb{B}_g (\pi) \right)^{-1} \mathbf{X}_g \gamma + \left( I_{n_g} - \mathbb{B}_g (\pi) \right)^{-1} \bar{\mathbf{X}}_g \delta + \left( I_{n_g} - \mathbb{B}_g(\pi) \right)^{-1} \boldsymbol{\varepsilon}_g. 
\end{align*}
We have a reduced-form linear relationship between $\mathbf{y}_g$ and $(\mathbf{X}_g, \bar{\mathbf{X}}_g)$. 

Since the reduced-form relationship holds for every realization of $\boldsymbol{\epsilon}_g$ such that the ordering $\pi$ stays the same, we can consider a conditional expectation of the linear relationship given the event that $\boldsymbol{\varepsilon}_g$ induces the ordering $\pi$. Then, we have the reduced-form representations.
\end{proof}

\subsection{Representation of the LiM and LiS estimands}\label{app:missspecproofdetails}
By plugging in
\begin{align*}
\E \left[ y_{i,g} \right] &= \sum_{k=1}^{d_{i,g}} \beta_{k,d_{i,g}} \E \left[ \tilde{y}_{i,k,g} \right] + {\boldsymbol{x}_{i,g}}^\intercal \gamma + {\bar{\boldsymbol{x}}_{i,g}}^\intercal \delta
\end{align*}
into the LiM estimand, we obtain
 \begin{align*}
\beta^{\mathrm{LiM}} &:= \left( \frac{1}{m} \sum_{g=1}^m \left( \tilde{\mathbf{Z}}_g^{\mathrm{LiM}} \right)^\intercal \mathbf{G}_g \E[\mathbf{y}_g] \right)^{-1} \frac{1}{m} \sum_{g=1}^m \left( \tilde{\mathbf{Z}}_g^{\mathrm{LiM}} \right)^\intercal \E [\mathbf{y}_g] \\
&= \left( \frac{1}{m} \sum_{g=1}^m \left( \tilde{\mathbf{Z}}_g^{\mathrm{LiM}} \right)^\intercal \mathbf{G}_g \E[\mathbf{y}_g] \right)^{-1} \frac{1}{m} \sum_{g=1}^m \left( \tilde{\mathbf{Z}}_g^{\mathrm{LiM}} \right)^\intercal \begin{pmatrix} \sum_{k=1}^{d_{1,g}} \beta_{k,d_{1,g}} \E \left[ \tilde{y}_{1,k,g} \right] \\ \vdots \\ \sum_{k=1}^{d_{n_g,g}} \beta_{k,d_{n_g,g}} \E \left[ \tilde{y}_{n_g,k,g} \right] \end{pmatrix} + o_p(1). 
\end{align*}
$o_p(1)$ term appears since \begin{align*}
    \frac{1}{m} \sum_{g=1}^m \left( \tilde{\mathbf{Z}}_g^{\mathrm{LiM}} \right)^\intercal \tilde{\mathbf{X}}_g^{\mathrm{LiM}} &= \frac{1}{m} \sum_{g=1}^m \left( {{\mathbf{Z}}_g}^\intercal - {\mathbf{P}^{\mathrm{LiM}}}^\intercal \left( \tilde{\mathbf{X}}_g^{\mathrm{LiM}}\right)^\intercal \right) \tilde{\mathbf{X}}_g^{\mathrm{LiM}} \\
    &= \frac{1}{m} \sum_{g=1}^m {{\mathbf{Z}}_g}^\intercal \tilde{\mathbf{X}}_g^{\mathrm{LiM}} - {\mathbf{P}^{\mathrm{LiM}}}^\intercal \cdot \frac{1}{m} \sum_{g=1}^m \left( \tilde{\mathbf{X}}_g^{\mathrm{LiM}}\right)^\intercal \tilde{\mathbf{X}}_g^{\mathrm{LiM}} = o_p(1). 
\end{align*}
Suppose that $\boldsymbol{z}_{i,g}$ and therefore its residual $\tilde{\boldsymbol{z}}_{i,g}^{\mathrm{LiM}}$ are scalar variables. Then, 
\begin{align*}
\beta^{\text{LiM}} &= 
\frac{\frac{1}{m}\sum_{g=1}^m \sum_{i=1}^{n_g} \tilde{\boldsymbol{z}}_{i,g}^{\mathrm{LiM}} \sum_{k=1}^{d_{i,g}}\E \left[ \tilde{y}_{i,k,g} \right] \beta_{k,d_{i,g}}}{\frac{1}{m}\sum_{g=1}^m \sum_{i=1}^{n_g} \tilde{\boldsymbol{z}}_{i,g}^{\mathrm{LiM}} \frac{1}{d_{i,g}} \sum_{k=1}^{d_{i,g}} \E \left[ \tilde{y}_{i,k,g} \right]} + o_p(1) \\
&= \frac{\sum_{d=1}^{\bar{d}} \sum_{k=1}^{d}\frac{1}{m}\sum_{g=1}^m \sum_{i=1}^{n_g} \tilde{\boldsymbol{z}}_{i,g}^{\mathrm{LiM}} \E \left[ \tilde{y}_{i,k,g}\mathbf{1}\{d_{i,g}=d\} \right] \beta_{k,d}}{\sum_{d=1}^{\bar{d}} \frac{1}{m}\sum_{g=1}^m \sum_{i=1}^{n_g} \frac{1}{d} \sum_{k=1}^{d} \tilde{\boldsymbol{z}}_{i,g}^{\mathrm{LiM}} \E \left[ \tilde{y}_{i,k,g} \mathbf{1}\{d_{i,g}=d\}\right]} + o_p(1). 
\end{align*}
By letting
\begin{align*}
w_{k,d}^{\text{LiM}} &= \frac{ \frac{1}{m}\sum_{g=1}^m \sum_{i=1}^{n_g} \frac{1}{d} \tilde{\boldsymbol{z}}_{i,g}^{\mathrm{LiM}} \E \left[ \tilde{y}_{i,k,g} \mathbf{1}\{d_{i,g}=d\} \right] }{\sum_{d=1}^{\bar{d}} \sum_{k=1}^{d} \frac{1}{m}\sum_{g=1}^m \sum_{i=1}^{n_g} \frac{1}{d} \tilde{\boldsymbol{z}}_{i,g}^{\mathrm{LiM}} \E \left[ \tilde{y}_{i,k,g} \mathbf{1}\{d_{i,g}=d\} \right]},
\end{align*}
$\beta^\text{LiM}$ becomes $$
\beta^{\text{LiM}} = \sum_{d=1}^{\bar{d}} \sum_{k=1}^d w_{k,d}^\text{LiM} d \beta_{k,d} + o_p(1) 
$$
as desired. The weights satisfy the sum-to-one constraint: $\sum_{d=1}^{\bar{d}} \sum_{k=1}^d w_{k,d}^\text{LiM} = 1$. 

Likewise, for the LiS peer effect parameter, we get
\begin{align*}
\beta^{\text{LiS}} &= \frac{\sum_{d=1}^{\bar{d}} \sum_{k=1}^{d} \frac{1}{m} \sum_{g=1}^m \sum_{i=1}^{n_g} \tilde{\boldsymbol{z}}_{i,g}^{\text{LiS}} \E \left[\tilde{y}_{i,k,g} \mathbf{1}\{d_{i,g}=d\} \right] \beta_{k,d} }{\sum_{d=1}^{\bar{d}} \frac{1}{m} \sum_{g=1}^m \sum_{i=1}^{n_g} \tilde{\boldsymbol{z}}_{i,g}^{\text{LiS}} \E \left[ \sum_{k=1}^{d} \tilde{y}_{i,k,g} \mathbf{1}\{d_{i,g}=d\} \right]} + o_p(1) \\
&= \sum_{d=1}^{\bar{d}} \sum_{k=1}^d w_{k,d}^{\text{LiS}} \beta_{k,d} + o_p(1),
\end{align*}
where
$$
w_{k,d}^{\text{LiS}} = \frac{\frac{1}{m}\sum_{g=1}^m \sum_{i=1}^{n_g} \tilde{\boldsymbol{z}}_{i,g}^{\text{LiS}} \E \left[ \tilde{y}_{i,k,g} \mathbf{1}\{d_{i,g}=d\}\right]}{\sum_{d=1}^{\bar{d}} \sum_{k=1}^{d}\frac{1}{m}\sum_{g=1}^m \sum_{i=1}^{n_g} \tilde{\boldsymbol{z}}_{i,g}^{\text{LiS}} \E \left[ \tilde{y}_{i,k,g} \mathbf{1}\{d_{i,g}=d\}\right]}
$$
and $\tilde{\boldsymbol{z}}_{i,g}^{\text{LiS}}$ is the residual from the linear-in-sums projection $\mathbf{P}^{\mathrm{LiS}}$. Again, the weights sum to one: $\sum_{d=1}^{\bar{d}} \sum_{k=1}^d w_{k,d}^{\text{LiS}} = 1$.

When $\boldsymbol{z}_{i,g}$ is not a scalar, we can modify the argument above and get similar expressions by replacing $\tilde{\boldsymbol{z}}_{i,g}$ with $c^\intercal \tilde{\boldsymbol{z}}_{i,g}$ with some weighting $c$. The weighting depends on how we weight over-identifying moments. For example, the TSLS principle uses the weighting vector $$
c^\intercal = \frac{1}{m} \sum_{g=1}^m \E \left[ \mathbf{y}_g^\intercal \right]\mathbf G_g^\intercal \tilde{\mathbf{Z}}_g \left( \frac{1}{m} \sum_{g=1}^m { \tilde{\mathbf{Z}}_g}^\intercal \tilde{\mathbf{Z}}_g \right)^{-1} ;
$$
misspecification weights depend on the weighting matrix when the LiM or LiS is overidentified.

\subsection{Encompassing Test on Empirical Specifications}\label{app:test}

The quantile specification discussed in Subsection \ref{sec:implementation:restriction} has the merits of parsimony and tractability. However, an important prerequisite for using such a restricted model specification is determining whether the model restriction sufficiently captures the heterogeneous pattern in the construction of peer effects, since misspecifying the model restriction can lead to incorrect counterfactual analyses and misguided policy recommendations. To check this empirically, we make use of the encompassing test proposed by \citet{smith1992non}, which compares two possibly misspecified model restrictions. Thus, the encompassing test provides us with data-driven guidance on the choice of the model specification.

For notational brevity, let $h$ denote a true model, which potentially follows a tractable specification from Subsection \ref{sec:implementation:restriction} \[
    y_{i,g} = \sum_{k=1}^M \beta^k \tilde{y}_{i,g}^k + {\boldsymbol{x}_{i,g}}^\intercal \gamma + {\bar{\boldsymbol{x}}_{i,g}}^\intercal \delta + \varepsilon_{i,g}.
\]
By stacking at the network level, we can rewrite the model as \begin{equation}\label{eq:y:T}
    \mathbf{y}_g = \mathbf{W}_g \theta + \boldsymbol{\varepsilon}_g
\end{equation}
where $\mathbf{W}_g$ is the network-level matrix stacking of $\boldsymbol{w}_{i,g}$ and $\theta$ is the vector of the model parameters $\big( \beta^1, \dots, \beta^M, \gamma, \delta \big)^\intercal$. 

Suppose that we are given two working model restrictions $h_1$ and $h_2$ and are given two sets of instruments specific to model restrictions $\{\mathbf Z_g^1\}_{g=1}^m$ and $\{\mathbf{Z}_g^2 \}_{g=1}^m$, while allowing (true) $h$ to potentially differ from both $h_1$ and $h_2$. The encompassing test is constructed for a pair of model restrictions. In practice, we consider a sequence of model restrictions with growing complexity. For example, in the quantile peer effect model, we gradually increase the number of quantile levels to include in the model restriction whenever a given specification is rejected in favor of another specification that includes more quantile levels. 

The TSLS estimators from the respective model restrictions are $$
\hat{\theta}^s = \big(\mathbf H^s \big)^{-1} \cdot \frac{1}{m} \sum_{g=1}^m {\mathbf{W}_g^s}^{\intercal} \mathbf{Z}_g^s \left(\frac{1}{m} \sum_{g=1}^m{\mathbf Z_g^s}^{\intercal} \mathbf{Z}_g^s \right)^{-1} \frac{1}{m} \sum_{g=1}^m{\mathbf{Z}_g^s}^{\intercal} \mathbf y_g
$$
where $\mathbf H^s = \frac{1}{m} \sum_{g=1}^m {\mathbf{W}_g^s}^{\intercal} \mathbf{Z}_g^s \big( \frac{1}{m} \sum_{g=1}^m {\mathbf Z_g^s}^{\intercal} \mathbf{Z}_g^s \big)^{-1} \frac{1}{m} \sum_{g=1}^m {\mathbf{Z}_g^s}^{\intercal} \mathbf W_g^s$ for $s=1,2$. To discuss the behaviors of the two TSLS estimators under misspecification, we assume regularity on the instruments.
\begin{assumption} \label{ass:encompas} For any $s\in\{1, \, 2\}$, 
\begin{enumerate}[label=\textbf{\alph*.}]
\item $\plim \frac{1}{m} \sum_{g=1}^m {\mathbf Z_g^s}^{\intercal} \mathbf{Z}_g^s$ is a nonstochastic full-rank matrix.
\item $\plim \frac{1}{m} \sum_{g=1}^m {\mathbf Z_g^s}^{\intercal} \mathbf{W}_g^s$ is a nonstochastic full-column rank matrix.
\item $\plim \frac{1}{m} \sum_{g=1}^m {\mathbf Z_g^s}^{\intercal} \boldsymbol{\varepsilon}_g = \boldsymbol{0}$.
\end{enumerate}
\end{assumption}
\noindent Then, we can define pseudo-true coefficient parameters $$
\theta^s = \plim_{m \to \infty} \big(\mathbf H^s \big)^{-1} \cdot \frac{1}{m} \sum_{g=1}^m {\mathbf{W}_g^s}^{\intercal} \mathbf{Z}_g^s \left( \frac{1}{m} \sum_{g=1}^m {\mathbf Z_g^s}^{\intercal} \mathbf{Z}_g^s \right)^{-1} \frac{1}{m} \sum_{g=1}^m {\mathbf{Z}_g^s}^{\intercal} \mathbf{W}_g \theta \quad \forall s = 1,2.
$$
Importantly, whenever $\mathbf{W}_g \theta$ lies in the column space of $\mathbf{W}_g^s$ for every $g$, $\mathbf{W}_g^s \theta^s = \mathbf{W}_g \theta$; even though the working model $h_s$ differs from the true model $h$, the misspecified model still retrieves the systematic variation of the true model $h$. Using $\theta^1$ and $\theta^2$, we rewrite Equation \eqref{eq:y:T} as:
\begin{equation}\label{eq:y:Ts}
    \mathbf{y}_g = \mathbf{W}_g^s \theta^s + \boldsymbol{\varepsilon}_g^s \quad \forall s = 1,2
\end{equation}
where $\boldsymbol{\varepsilon}_g^s = \boldsymbol{\varepsilon}_g + \mathbf{W}_g \theta - \mathbf{W}_g^s \theta^s$. 

The encompassing test \citep{smith1992non} assesses whether the first model specified with $h_1$ and therefore $\theta^1$ replicates the features of the second model specified with $h_2$. Let $$
\psi = \plim_{m \to \infty} \big(\mathbf H^2 \big)^{-1} \cdot \frac{1}{m} \sum_{g=1}^m {\mathbf{W}_g^2}^{\intercal} \mathbf{Z}_g^2 \left( \frac{1}{m} \sum_{g=1}^m {\mathbf Z_g^2}^{\intercal}\mathbf{Z}_g^2 \right)^{-1} \frac{1}{m} \sum_{g=1}^m {\mathbf{Z}_g^2}^{\intercal} {\boldsymbol{\varepsilon}}_g^1.
$$
$\psi$ is the TSLS estimand of regressing $\{\boldsymbol{\varepsilon}_g^1\}_{g=1}^m$, the residual vectors under the first model restriction, on the second model restriction regressors $\{\mathbf{W}_g^2\}_{g=1}^m$ using the second model restriction instruments $\{\mathbf{Z}_g^2\}_{g=1}^m$. $\psi \neq 0$ suggests that features of the observed data relevant for the specification with $h_2$ are not completely captured by the model with $h_1$, leading to instrumented variation of $\mathbf{W}_g^2$ being correlated with $\boldsymbol{\varepsilon}_g^1$. In that case, the first model restriction must be rejected in favor of the second.

The test statistic of the encompassing test is $$
\hat{\psi} = \big( {\mathbf H^2} \big)^{-1} \cdot \frac{1}{m} \sum_{g=1}^m {\mathbf{W}_g^2}^{\intercal} \mathbf{Z}_g^2 \left( \frac{1}{m} \sum_{g=1}^m {\mathbf Z_g^2}^{\intercal}\mathbf{Z}_g^2 \right)^{-1} \frac{1}{m} \sum_{g=1}^m {\mathbf{Z}_g^2}^{\intercal} \hat{\boldsymbol{\varepsilon}}_g^1
$$
where $\hat{\boldsymbol\varepsilon}_g^1 = \mathbf{y}_g - \mathbf{W}_g^1 \hat{\theta}^1$. Proposition \ref{prop:encompassing} is the formal asymptotic normality result for the test statistic. 
\begin{proposition}\label{prop:encompassing}
 Under Assumptions~\ref{ass:exogeneity} and~\ref{ass:encompas}, and the regularity conditions of Assumption~\ref{ass:encompas:append} (see Appendix \ref{app:encompassing}):
\begin{enumerate}[label=\textbf{\alph*.}]
    \item $\hat{\psi}$ is a consistent estimator of $\psi$.
    \item $\sqrt{m}(\hat{\psi} - \psi)$ converges weakly to $N(\mathbf{0}, \boldsymbol\Omega_{\psi})$. 
\end{enumerate}
\end{proposition}
\begin{proof}
    See Appendix \ref{app:encompassing}.
\end{proof}
Since the $\hat{\boldsymbol{\varepsilon}}_g^1$ is obtained from an initial regression, the asymptotic variance of $\sqrt{m}(\hat{\psi} - \psi)$ depends on the sampling variability from this initial regression as well; a naive IV standard error obtained from treating $\hat{\boldsymbol{\varepsilon}}_g^1$ as $\boldsymbol{\varepsilon}_g^1$ would be invalid. Another important complication in the asymptotic theory is that both $h_1$ and $h_2$ may differ from $h$, neither being correctly specified. We therefore use an approach similar to those employed under misspecified moment functions to derive the asymptotic distribution \citep{hall2003large}.

A merit of the encompassing test is that it remains consistent even when both models are misspecified, under some conditions. Note that the test statistic is asymptotically centered around $\psi$. Thus, the test is consistent as long as $\psi \neq 0$ under an alternative. Since \begin{align*}
    \psi &=({\mathbf H^2})^{-1} \cdot \frac{1}{m} \sum_{g=1}^m {\mathbf{W}_g^2}^{\intercal} \mathbf{Z}_g^2 \left( \frac{1}{m} \sum_{g=1}^m {\mathbf Z_g^2}^{\intercal}\mathbf{Z}_g^2 \right)^{-1} \frac{1}{m} \sum_{g=1}^m {\mathbf{Z}_g^2}^{\intercal} \boldsymbol{\varepsilon}_g^1 \\
    &= ({\mathbf H^2})^{-1} \cdot \frac{1}{m} \sum_{g=1}^m {\mathbf{W}_g^2}^{\intercal} \mathbf{Z}_g^2 \left( \frac{1}{m} \sum_{g=1}^m {\mathbf Z_g^2}^{\intercal}\mathbf{Z}_g^2 \right)^{-1} \frac{1}{m} \sum_{g=1}^m {\mathbf{Z}_g^2}^{\intercal} \left( \mathbf{W}_g \theta - \mathbf{W}_g^1 \theta^1 \right), 
\end{align*}
$\psi$ equals a zero vector if and only if the structural part of the true model, $\mathbf{W}_g \theta$, and its approximation under the first model restriction, $\mathbf{W}_g^1 \theta^1$, are equivalent from the perspective of $$(\mathbf{H}^2)^{-1} \cdot \frac{1}{m} \sum_{g=1}^m {\mathbf{W}_g^2}^\intercal \mathbf{Z}_g^2 \left( \frac{1}{m} \sum_{g=1}^m{\mathbf{Z}_g^2}^\intercal \mathbf{Z}_g^2 \right)^{-1} {\mathbf{Z}_g^2}^\intercal.
$$
Thus, even though the test seemingly tests the null hypothesis of the first model restriction against the alternative hypothesis of the second model restriction, the test still has power against an alternative hypothesis which is not the second model restriction, as long as $(\mathbf{H}^2)^{-1} \cdot \frac{1}{m} \sum_{g=1}^m {\mathbf{W}_g^2}^\intercal \mathbf{Z}_g^2 \left( \frac{1}{m} \sum_{g=1}^m {\mathbf{Z}_g^2}^\intercal \mathbf{Z}_g^2 \right)^{-1} {\mathbf{Z}_g^2}^\intercal$ preserves the nonzero directions of $
\{\mathbf{W}_g \theta - \mathbf{W}_g^1 \theta^1 \}_{g=1}^m.
$

\subsection{Proof for Proposition \ref{prop:encompassing}}\label{app:encompassing}

We start by showing that $\hat{\psi}$ is consistent for $\psi$:
\begin{align*}
\plim_{m \to \infty} \hat{\psi} - \psi &= \plim_{m \to \infty} \big(\mathbf{H}^2 \big)^{-1} \frac{1}{m} \sum_{g=1}^m {\mathbf{W}_g^2}^{\intercal} \mathbf{Z}_g^2 \left( \frac{1}{m} \sum_{g=1}^m {\mathbf{Z}_g^2}^{\intercal} \mathbf{Z}_g^2 \right)^{-1} \frac{1}{m} \sum_{g=1}^m {\mathbf{Z}_g^2}^{\intercal} \left( \hat{\boldsymbol{\varepsilon}}_g^1 - \boldsymbol{\varepsilon}_g^1 \right) \\
&= \plim_{m \to \infty} \big(\mathbf{H}^2 \big)^{-1} \frac{1}{m} \sum_{g=1}^m {\mathbf{W}_g^2}^{\intercal} \mathbf{Z}_g^2 \left( \frac{1}{m} \sum_{g=1}^m {\mathbf{Z}_g^2}^{\intercal} \mathbf{Z}_g^2 \right)^{-1} \frac{1}{m} \sum_{g=1}^m {\mathbf{Z}_g^2}^{\intercal} \left( \mathbf y_g - \mathbf W_g^1 \hat{\theta}^1 - \boldsymbol{\varepsilon}_g^1 \right) \\
&= \plim_{m \to \infty} \big(\mathbf{H}^2 \big)^{-1} \frac{1}{m} \sum_{g=1}^m {\mathbf{W}_g^2}^{\intercal} \mathbf{Z}_g^2 \left( \frac{1}{m} \sum_{g=1}^m {\mathbf{Z}_g^2}^{\intercal} \mathbf{Z}_g^2 \right)^{-1} \frac{1}{m} \sum_{g=1}^m {\mathbf{Z}_g^2}^{\intercal}  \mathbf W_g^1 \left( \theta^1 - \hat{\theta}^1 \right) = 0. 
\end{align*}
The second and third equalities use $\mathbf{y}_g = {\mathbf{W}_g^1} \theta^1 + \boldsymbol{\varepsilon}_g^1$. The last equality holds from Assumption \ref{ass:encompas}. This completes the proof of the first statement of the proposition.

To show that $\sqrt{m}(\hat{\psi} -  \psi)$ is asymptotically normally distributed, we define $$
\tilde{\boldsymbol\varepsilon}_g : = \boldsymbol\varepsilon_g^1 - \mathbf W_g^2 \psi,
$$
where $\tilde{\boldsymbol\varepsilon}_g$ is the error term of the linear approximation of $\boldsymbol\varepsilon^1$ onto $\mathbf{W}_g^2$. Recall that $\psi= \textbf{0}$ when $h_1$ is true. The main challenge in deriving the asymptotic distribution is that $h_1$ may be misspecified, in which case both $\E[ \boldsymbol{\varepsilon}_g^1]$ and $\E [ \tilde{\boldsymbol{\varepsilon}}_g]$ may differ from zero. Let $\displaystyle \boldsymbol u_1 = \frac{1}{m} \sum_{g=1}^m {\mathbf{Z}_g^1}^{\intercal} \boldsymbol\varepsilon_g^1$, $\displaystyle \boldsymbol u_2 = \frac{1}{m} \sum_{g=1}^m {\mathbf{Z}_g^2}^{\intercal} \tilde{\boldsymbol\varepsilon}_g$. For $s=1,2$, let $\displaystyle \boldsymbol u_{s,3} = \frac{1}{m} \sum_{g=1}^m \operatorname{vec}\big({\mathbf{W}_g^s}^{\intercal} \mathbf{Z}_g^s \big)$ and $\displaystyle \boldsymbol u_{s,4} = \frac{1}{m} \sum_{g=1}^m \operatorname{vec} \big( {\mathbf{Z}_g^s}^{\intercal} \mathbf{Z}_g^s \big)$, where the operator $\operatorname{vec}$ vectorizes a matrix by stacking its columns sequentially. Lastly, let $\boldsymbol{u} = \left(\boldsymbol u_1^{\intercal}, \, \boldsymbol u_2^{\intercal}, \, \boldsymbol u_{1,3}^{\intercal}, \, \boldsymbol u_{1,4}^{\intercal}, \, \boldsymbol u_{2,3}^{\intercal}, \, \boldsymbol u_{2,4}^{\intercal}\right)^{\intercal}$. We introduce the following conditions.

\begin{assumption}\label{ass:encompas:append}
\hfill
\begin{enumerate}[label=\textbf{\alph*.}]
    \item For any $s\in\{1, \, 2\}$, $\displaystyle \plim_{m \to \infty} \left(\boldsymbol u_{s,3} - \mathbb E[\boldsymbol u_{s,3}]\right) = 0$ and $\displaystyle \plim_{m \to \infty} \left(\boldsymbol u_{s,4} - \mathbb E[\boldsymbol u_{s,4}]\right) = 0$. \label{ass:encompas:append:a}
    
    \item $\displaystyle \plim_{m \to \infty} \left(\boldsymbol u_1 - \mathbb E[\boldsymbol u_1]\right) = 0$ and $\plim \left(\boldsymbol u_2 - \mathbb E[\boldsymbol u_2]\right) = 0$. \label{ass:encompas:append:b}
    
    \item $\sqrt{m}\left(\boldsymbol{u} - \mathbb E[\boldsymbol{u}]\right)$ converges in distribution to $\mathcal{N}(\mathbf 0, \,\boldsymbol\Omega_{\boldsymbol u})$. \label{ass:encompas:append:c}
\end{enumerate}
\end{assumption}

Assumption~\ref{ass:encompas:append} imposes weak regularity conditions to establish the asymptotic distribution of $\sqrt{m}(\hat{\psi} - \psi)$. Conditions~\ref{ass:encompas:append:a} and~\ref{ass:encompas:append:b} follow from the weak law of large numbers applied to $\boldsymbol u_1$, $\boldsymbol u_2$, $\boldsymbol u_{s,3}$, and $\boldsymbol u_{s,4}$, which are sample means. Condition~\ref{ass:encompas:append:c} follows from standard central limit theorems (CLT) because $\boldsymbol{u}$ is a sample mean. The reason for also requiring the CLT to apply to $\boldsymbol u_{s,3}$ and $\boldsymbol u_{s,4}$ is that $\mathbb{E}[\boldsymbol u_1]$ and $\mathbb{E} [\boldsymbol u_2]$ may differ from zero when both $h_1$ and $h_2$ are misspecified \citep{hall2003large}. We have 
\allowdisplaybreaks
\begin{align*} 
    &\sqrt{m}(\hat{\psi} -  \psi) \\
    &= \sqrt{m} \big( \mathbf H^2 \big)^{-1} \frac{1}{m} \sum_{g=1}^m {\mathbf{W}_g^2}^{\intercal} \mathbf{Z}_g^2 \left( \frac{1}{m} \sum_{g=1}^m {\mathbf Z_g^2}^{\intercal} \mathbf{Z}_g^2 \right)^{-1} \frac{1}{m} \sum_{g=1}^m {\mathbf{Z}_g^2}^{\intercal} \hat{\boldsymbol{\varepsilon}}_g^1 - \sqrt{m}\psi \\
    &= \sqrt{m} \big( \mathbf H^2 \big)^{-1} \frac{1}{m} \sum_{g=1}^m {\mathbf{W}_g^2}^{\intercal} \mathbf{Z}_g^2 \left( \frac{1}{m} \sum_{g=1}^m {\mathbf Z_g^2}^{\intercal} \mathbf{Z}_g^2 \right)^{-1} \frac{1}{m} \sum_{g=1}^m {\mathbf{Z}_g^2}^{\intercal} \left( \mathbf{W}_g^1 \big(\theta^1 -\hat{\theta}^1 \big) + \boldsymbol{\varepsilon}_g^1  \right) - \sqrt{m}\psi \\
    &= \sqrt{m} \big( \mathbf H^2 \big)^{-1} \frac{1}{m} \sum_{g=1}^m {\mathbf{W}_g^2}^{\intercal} \mathbf{Z}_g^2 \left( \frac{1}{m} \sum_{g=1}^m {\mathbf Z_g^2}^{\intercal} \mathbf{Z}_g^2 \right)^{-1} \frac{1}{m} \sum_{g=1}^m {\mathbf{Z}_g^2}^{\intercal} \left( \mathbf{W}_g^1 \big(\theta^1 -\hat{\theta}^1 \big) + \tilde{\boldsymbol{\varepsilon}}_g \right) \\
    &= \sqrt{m} \big( \mathbf H^2 \big)^{-1} \frac{1}{m} \sum_{g=1}^m {\mathbf{W}_g^2}^{\intercal} \mathbf{Z}_g^2 \left( \frac{1}{m} \sum_{g=1}^m {\mathbf Z_g^2}^{\intercal} \mathbf{Z}_g^2 \right)^{-1} \frac{1}{m} \sum_{g=1}^m {\mathbf{Z}_g^2}^{\intercal} \tilde{\boldsymbol{\varepsilon}}_g \\
    & \qquad - \sqrt{m} \big( \mathbf H^2 \big)^{-1} \frac{1}{m} \sum_{g=1}^m {\mathbf{W}_g^2}^{\intercal} \mathbf{Z}_g^2 \left( \frac{1}{m} \sum_{g=1}^m {\mathbf Z_g^2}^{\intercal} \mathbf{Z}_g^2 \right)^{-1} \frac{1}{m} \sum_{g=1}^m {\mathbf{Z}_g^2}^{\intercal} \mathbf{W}_g^1 \\
    & \qquad \qquad \cdot \big(\mathbf H^1 \big)^{-1} \frac{1}{m} \sum_{g=1}^m {\mathbf{W}_g^1}^{\intercal} \mathbf{Z}_g^1 \left(\frac{1}{m} \sum_{g=1}^m{\mathbf Z_g^1}^{\intercal} \mathbf{Z}_g^1 \right)^{-1} \frac{1}{m} \sum_{g=1}^m{\mathbf{Z}_g^1}^{\intercal} \boldsymbol{\varepsilon}_g^1 
\end{align*}

From the last equation, even though $\mathbb{E}[{\mathbf{Z}_g^1}^{\intercal} \boldsymbol\varepsilon_g^1]$ and $\mathbb{E} [{\mathbf{Z}_g^2}^{\intercal} \tilde{\boldsymbol\varepsilon}_g ]$ are not zero, we can adapt results from the misspecified moment condition framework \citep[see][]{hall2003large}. Specifically, Assumption~\ref{ass:encompas:append} guarantees that $\sqrt{m}(\hat{\psi} -  \psi)$ weakly converges to a normal distribution with mean zero. The exact expression of the asymptotic variance depends on $\boldsymbol\Omega_{\boldsymbol u}$ and is given in \citet{hall2003large}. In practice, we estimate this asymptotic variance using a paired bootstrap method.

\section{Additional Results on Monte Carlo Simulation} \label{app:sim}
\subsection{Comparing OLS and TSLS estimates}

\begin{table}[htpb]
\vspace{3mm}
    \centering
    \small
    \renewcommand{\arraystretch}{1.8}
    \caption{Comparison between the OLS Estimator and the TSLS Estimator}
    \begin{tabular}{l cccc c cccc c}
        \hline\hline
         & \multicolumn{4}{c}{Bias} & & \multicolumn{4}{c}{MSE} & \\
         \cline{2-5} \cline{7-10}       
         & (1) & (2) & (3) & (4) & & (5) & (6) & (7) & (8) \\
         \hline \multicolumn{4}{l}{\textbf{Panel A: Rank Models}} \\ 
         \hline 
$\hat{\gamma}_{1}^{OLS}$ & -0.021 & -0.026 & -0.022 & -0.038 &  & 0.009 & 0.010 & 0.005 & 0.007 \\ 
$\hat{\beta}_{1,2}^{OLS}$ & 0.044 & 0.021 & 0.033 & 0.017 &  & 0.006 & 0.002 & 0.003 & 0.001 \\ 
$\hat{\beta}_{2,2}^{OLS}$ & -0.013 & -0.007 & 0.006 & 0.004 &  & 0.004 & 0.001 & 0.002 & 0.001 \\ 
$\hat{\gamma}_{1}^{IV}$ & 0.001 & 0.001 & 0.001 & -0.002 &  & 0.010 & 0.010 & 0.005 & 0.006 \\ 
$\hat{\beta}_{1,2}^{IV}$ & 0.012 & 0.003 & -0.000 & 0.001 &  & 0.012 & 0.002 & 0.007 & 0.002 \\ 
$\hat{\beta}_{2,2}^{IV}$ & -0.008 & -0.002 & 0.002 & 0.000 &  & 0.011 & 0.002 & 0.007 & 0.001 \\ 
\multicolumn{4}{l}{\textbf{Panel B: Quantile Models}} \\
\hline 
$\hat{\gamma}_{1}^{OLS}$ & -0.023 & -0.027 & -0.023 & -0.038 &  & 0.009 & 0.010 & 0.005 & 0.007 \\ 
$\hat{\beta}^{1,OLS}$ & 0.034 & 0.016 & 0.027 & 0.015 &  & 0.004 & 0.001 & 0.002 & 0.001 \\ 
$\hat{\beta}^{2,OLS}$ & -0.004 & -0.002 & 0.012 & 0.006 &  & 0.003 & 0.001 & 0.001 & 0.000 \\ 
$\hat{\gamma}_{1}^{IV}$ & 0.001 & 0.001 & 0.001 & -0.002 &  & 0.009 & 0.010 & 0.005 & 0.006 \\ 
$\hat{\beta}^{1,IV}$ & 0.003 & 0.001 & -0.001 & 0.000 &  & 0.006 & 0.001 & 0.002 & 0.001 \\ 
$\hat{\beta}^{2,IV}$ & -0.001 & -0.001 & 0.001 & 0.000 &  & 0.005 & 0.001 & 0.002 & 0.001 \\ 
 \hline
         $\bar{\gamma}$ & 1 & 2 & 1 & 2 & & 1 & 2 & 1 & 2 \\
         $m$ & 2 & 2 & 5 & 5 & & 2 & 2 & 5 & 5  \\
         \hline
    \end{tabular}
    \label{tab:sim_OLS} \begin{minipage}{0.97\textwidth}
    \footnotesize \vspace{3mm} \textit{Notes:} The values of the coefficients used in DGP are as follows: $\beta = \big(0.2, 0.5 \big)$, $\gamma = \bar{\gamma} \big(1 ,0.25\big)$ and $\delta = \big( -0.5, 1.2 \big)$. The instruments used in TSLS estimation for the rank models are peers' covariates ordered by their own values. The quantile models use the quantiles of the distribution of each covariate amongst an individual's peers. \end{minipage}
    \vspace{3mm}
\end{table}

Table \ref{tab:sim_OLS} compares the OLS estimator and the TSLS estimator using ordered peer covariates as instruments in a sparse network setting, setting $\bar{d} = 2$. We simulate from the quantile model with $\beta = (0.2,0.5)$. This then implies $\beta_{1,1} = 0.7, \beta_{1,2} = 0.2, \beta_{2,2} = 0.5$. To mimic a realistic sample, we let each network be small, consisting of 50 individuals. We then estimate the models with 2 or 5 networks, giving respective sample sizes of 100 or 250 total observations. All estimators include network fixed effects. 

As expected, OLS is biased, and the bias is unaffected by the sample size. However, when we increase the strength of the covariates by increasing their coefficients in absolute value, the bias of the OLS peer effect estimators decreases. This is caused by the increased covariate strength reducing the relative share of endogeneity in the outcomes, which then reduces the impact of the reflection bias. 

\subsection{Simulations Without Contextual Effects} \label{app:sim_without_contextual}
For completeness, we here present our simulation results from the same DGP as the main text, except we have removed the contextual effects: Table \ref{tab:sim_OLS_without_contextual} repeats Table \ref{tab:sim_OLS} while letting $\delta=\textbf{0}$ and Tables \ref{tab:sim_change_ng_noxbar}-\ref{tab:sim_compare_estimand_without_contextual} do the same with Tables \ref{tab:sim_change_ng}-\ref{tab:sim_compare_estimand}. 

\begin{table}[!t]
\vspace{3mm}
    \centering
    \small
    \renewcommand{\arraystretch}{1.8}
    \caption{Comparison between the OLS Estimator and the TSLS Estimator}
    \begin{tabular}{l cccc c cccc c}
        \hline\hline
         & \multicolumn{4}{c}{Bias} & & \multicolumn{4}{c}{MSE} & \\
         \cline{2-5} \cline{7-10}       
         & (1) & (2) & (3) & (4) & & (5) & (6) & (7) & (8) \\
         \hline \multicolumn{4}{l}{\textbf{Panel A: Rank Models}} \\ 
         \hline 
$\hat{\gamma}_{1}^{OLS}$ & -0.005 & -0.003 & -0.001 & -0.001 &  & 0.009 & 0.009 & 0.005 & 0.005 \\ 
$\hat{\beta}_{1,2}^{OLS}$ & 0.027 & 0.008 & 0.030 & 0.010 &  & 0.003 & 0.001 & 0.002 & 0.001 \\ 
$\hat{\beta}_{2,2}^{OLS}$ & -0.025 & -0.007 & -0.027 & -0.009 &  & 0.003 & 0.001 & 0.002 & 0.001 \\ 
$\hat{\gamma}_{1}^{IV}$ & 0.003 & 0.003 & -0.001 & -0.001 &  & 0.009 & 0.009 & 0.005 & 0.005 \\ 
$\hat{\beta}_{1,2}^{IV}$ & 0.005 & 0.001 & 0.002 & 0.001 &  & 0.010 & 0.002 & 0.009 & 0.002 \\ 
$\hat{\beta}_{2,2}^{IV}$ & -0.005 & -0.001 & -0.002 & -0.001 &  & 0.009 & 0.002 & 0.008 & 0.002 \\ 
\multicolumn{4}{l}{\textbf{Panel B: Quantile Models}} \\  
\hline
$\hat{\gamma}_{1}^{OLS}$ & -0.005 & -0.002 & -0.003 & -0.002 &  & 0.009 & 0.009 & 0.005 & 0.005 \\ 
$\hat{\beta}^{1,OLS}$ & 0.017 & 0.005 & 0.018 & 0.006 &  & 0.002 & 0.001 & 0.001 & 0.000 \\ 
$\hat{\beta}^{2,OLS}$ & -0.016 & -0.004 & -0.016 & -0.005 &  & 0.002 & 0.000 & 0.001 & 0.000 \\ 
$\hat{\gamma}_{1}^{IV}$ & 0.002 & 0.002 & -0.001 & -0.001 &  & 0.009 & 0.009 & 0.005 & 0.005 \\ 
$\hat{\beta}^{1,IV}$ & 0.000 & -0.000 & 0.001 & 0.000 &  & 0.003 & 0.001 & 0.002 & 0.001 \\ 
$\hat{\beta}^{2,IV}$ & -0.000 & 0.000 & -0.000 & -0.000 &  & 0.003 & 0.001 & 0.002 & 0.001 \\ 
 \hline
         $\bar{\gamma}$ & 1 & 2 & 1 & 2 & & 1 & 2 & 1 & 2 \\
         $m$ & 2 & 2 & 5 & 5 & & 2 & 2 & 5 & 5  \\
         \hline
    \end{tabular}
    \label{tab:sim_OLS_without_contextual} \begin{minipage}{0.97\textwidth}
    \footnotesize \vspace{3mm} \textit{Notes:} The values of the coefficients used in DGP are as follows: $\beta = \big(0.2, 0.5 \big)$ and $\gamma = \bar{\gamma} \big(1 ,0.25\big)$. The instruments used in TSLS estimation for the rank models are peers' covariates ordered by their own values. The quantile models use the quantiles of the distribution of each covariate amongst an individual's peers. \end{minipage}
    \vspace{3mm}
\end{table}

\begin{table}[!t]
\vspace{3mm}
    \centering
    \small
    \renewcommand{\arraystretch}{1.8}
    \caption{Finite Sample Performance of Rank and Quantile Models} \label{tab:sim_change_ng_noxbar}
    \begin{tabular}{l ccccc c ccccc c}
        \hline \hline
         & \multicolumn{5}{c}{Bias } & & \multicolumn{5}{c}{MSE} \\
         \cline{2-6} \cline{8-12}       
         & (1) & (2) & (3) & (4) &(5) &  & (6) & (7) & (8)& (9) & (10) \\
                 
\multicolumn{5}{l}{\textbf{Panel A: Saturated model}} \\ 
         \hline 
$\beta_{1,2}$ & -0.001 & -0.001 & -0.001 & -0.000 & 0.000 &  & 0.007 & 0.007 & 0.004 & 0.002 & 0.001 \\ 
$\beta_{2,2}$ & 0.001 & 0.001 & 0.001 & 0.000 & -0.000 &  & 0.006 & 0.006 & 0.004 & 0.002 & 0.001 \\ 
$\beta_{1,5}$ & -0.007 & 0.012 & -0.001 & 0.000 & 0.000 &  & 0.088 & 0.067 & 0.060 & 0.015 & 0.015 \\ 
$\beta_{2,5}$ & 0.012 & -0.022 & 0.003 & 0.003 & -0.003 &  & 0.323 & 0.288 & 0.260 & 0.092 & 0.135 \\ 
$\beta_{3,5}$ & -0.008 & 0.005 & -0.004 & -0.011 & 0.005 &  & 0.352 & 0.312 & 0.281 & 0.240 & 0.271 \\ 
$\beta_{4,5}$ & 0.014 & 0.010 & 0.007 & 0.012 & -0.003 &  & 0.253 & 0.271 & 0.161 & 0.169 & 0.147 \\ 
$\beta_{5,5}$ & -0.010 & -0.004 & -0.005 & -0.005 & 0.001 &  & 0.084 & 0.072 & 0.045 & 0.027 & 0.015 \\ 

\multicolumn{5}{l}{\textbf{Panel B: Restricted model } } \\
         \hline
      $\beta^1$ & 0.008 & 0.001 & -0.003 & -0.000 & -0.001 &  & 0.035 & 0.047 & 0.022 & 0.005 & 0.003 \\ 
$\beta^2$ & -0.024 & 0.001 & 0.008 & -0.002 & 0.001 &  & 0.227 & 0.335 & 0.160 & 0.047 & 0.027 \\ 
$\beta^3$ & 0.022 & -0.002 & -0.005 & 0.004 & -0.000 &  & 0.224 & 0.289 & 0.138 & 0.056 & 0.030 \\ 
$\beta^4$ & -0.006 & 0.001 & 0.000 & -0.002 & -0.000 &  & 0.034 & 0.033 & 0.016 & 0.008 & 0.004 \\ 
 \hline
         $m$ &  5 & 10 & 20 & 50& 100 & &  5 & 10 & 20 & 50& 100 \\
         $\bar{d}$ &  5 & 5 & 5 & 5& 5 & & 5 & 5 & 5 & 5& 5 \\
         \hline
    \end{tabular}
     \begin{minipage}{0.97\textwidth}   \footnotesize \vspace{3mm} \textit{Notes:}  The data generating process is the quantile model with $\beta = \left(-0.05,0.35,0.15,0.1\right)$. $m$ gives the number of networks used in the simulation and $\bar{d}$ is the maximum number of peers for each individual. All the simulations include network fixed effects. The instruments used are the distribution of the peers' covariate values, as measured either by the ranked peer covariate values or the quantiles of peer covariate values. The results are based on 10{,}000 draws.
 \end{minipage}
\end{table}

\begin{table}[t!]
\centering
\caption{Monte Carlo Simulations --- Encompassing Tests}
\footnotesize
\begin{threeparttable}
\begin{tabular}{P{1.6cm}P{1.6cm}lP{1.6cm}P{1.6cm}lP{1.6cm}P{1.6cm}}
\toprule
\multicolumn{2}{c}{2. vs. 3.} & &\multicolumn{2}{c}{3. vs. 4.} & &\multicolumn{2}{c}{4. vs. 5.}\\[0.5ex]
\cline{1-2} \cline{4-5} \cline{7-8}
\addlinespace[0.5ex]
5\% & 10\% && 5\% & 10\% && 5\% & 10\%\\
\midrule
\multicolumn{8}{c}{DGP A:   $\beta = (0, 0.05, 0.2, 0.3)$}      \\[0.5ex]
1.000    & 1.000   &    & 0.201   & 0.307  &   & 0.008  & 0.020 \\[1.5ex]
\multicolumn{8}{c}{DGP B: $\beta =   (0.3, 0.2, 0.05, 0)$}      \\[0.5ex]
1.000    & 1.000   &    & 0.987   & 0.994  &   & 0.013  & 0.027 \\[1.5ex]
\multicolumn{8}{c}{DGP C: $\beta = (0,   0.275, 0.275, 0)$}     \\[0.5ex]
1.000    & 1.000   &    & 0.999   & 1.000  &   & 0.008  & 0.021 \\[1.5ex]
\multicolumn{8}{c}{DGP D: $\beta =   (0.275, 0, 0, 0.275)$}     \\[0.5ex]
0.011    & 0.023   &    & 0.008   & 0.016  &   & 0.012  & 0.026 \\[1.5ex]
\multicolumn{8}{c}{DGP E: $\beta =   (-0.05, 0.35, 0.15, 0.1)$} \\[0.5ex]
1.000    & 1.000   &    & 1.000   & 1.000  &   & 0.010  & 0.023 \\[1.5ex]
\multicolumn{8}{c}{DGP F (LiM model):   $\beta = 0.55$}         \\[0.5ex]
1.000    & 1.000   &    & 0.991   & 0.997  &   & 0.703  & 0.815 \\\bottomrule
\end{tabular}
\begin{tablenotes}[para,flushleft]
The columns labeled ``$a$ vs. $b$,'' for integers $a$ and $b$, report the share of rejections of the null hypothesis that the model with $a$ quantile levels does not perform worse than the model with $b$ quantile levels, at the significance levels indicated in the second row.
\end{tablenotes}
\end{threeparttable}
\end{table}

\begin{table}[t!]
\centering
\small
\caption{Monte Carlo Simulations --- Estimation of the Quantile Specification}
\label{tab:sim_compare_estimand_without_contextual}
\begin{threeparttable}
\begin{tabular}{cccclclcc}
\toprule
\multicolumn{4}{c}{Quantiles}                                           &            &    \multicolumn{1}{c}{LiM}           &           & \multicolumn{2}{c}{CES}           \\
\cline{1-4} \cline{6-6} \cline{8-9}
\addlinespace[0.5ex]
$\beta^1$ & $\beta^2$ & $\beta^3$ & $\beta^4$ & & 
$\beta^{\textrm{LiM}}$   & & $\rho$  & $\beta^{\textrm{CES}}$ \\
\midrule
\multicolumn{9}{c}{DGP A:   $\beta = (0, 0.05, 0.2, 0.3)$}            \\[0.5ex]
-0.000    & 0.050     & 0.200     & 0.300     &  & 0.544   &  & 34.759   & 0.554   \\
(0.004)   & (0.018)   & (0.037)   & (0.025)   &  & (0.003) &  & (12.301) & (0.004) \\[2ex]
\multicolumn{9}{c}{DGP B: $\beta =   (0.3, 0.2, 0.05, 0)$}                         \\[0.5ex]
0.300     & 0.200     & 0.050     & 0.000     &  & 0.544   &  & -3.165   & 0.530   \\
(0.006)   & (0.027)   & (0.049)   & (0.030)   &  & (0.004) &  & (0.165)  & (0.003) \\[2ex]
\multicolumn{9}{c}{DGP C: $\beta = (0,   0.275, 0.275, 0)$}                        \\[0.5ex]
0.000     & 0.275     & 0.275     & 0.000     &  & 0.544   &  & 2.918    & 0.552   \\
(0.004)   & (0.020)   & (0.040)   & (0.027)   &  & (0.003) &  & (0.243)  & (0.003) \\[2ex]
\multicolumn{9}{c}{DGP D: $\beta =   (0.275, 0, 0, 0.275)$}                        \\[0.5ex]
0.275     & -0.000    & 0.000     & 0.275     &  & 0.555   &  & -1.102   & 0.545   \\
(0.005)   & (0.023)   & (0.044)   & (0.028)   &  & (0.003) &  & (0.132)  & (0.003) \\[2ex]
\multicolumn{9}{c}{DGP E: $\beta =   (-0.05, 0.35, 0.15, 0.1)$}                    \\[0.5ex]
-0.050    & 0.350     & 0.150     & 0.100     &  & 0.543   &  & 4.696    & 0.556   \\
(0.004)   & (0.020)   & (0.039)   & (0.026)   &  & (0.003) &  & (0.553)  & (0.003) \\[2ex]
\multicolumn{9}{c}{DGP F (LiM model):   $\beta = 0.55$}                            \\[0.5ex]
0.112     & 0.198     & 0.093     & 0.147     &  & 0.550   &  & 1.001    & 0.550   \\
(0.005)   & (0.024)   & (0.046)   & (0.029)   &  & (0.003) &  & (0.078)  & (0.003) \\\bottomrule
\end{tabular}
\begin{tablenotes}[para,flushleft] 
\footnotesize
The models are simulated and estimated 1{,}000 times. Values without parentheses represent average peer effect estimates, while those in parentheses correspond to standard deviations. The instrument matrix includes the quantiles of $\boldsymbol{x}$ and $\bar{\boldsymbol{x}}$ among friends, computed at ten levels uniformly spaced between 0 and 1. DGPs A--E are generated from the proposed quantile-based model, with four quantiles at $\{0, ~1/3, ~ 2/3,~ 1\}$, and $\beta = (\beta^1, ~\beta^2, ~\beta^3, ~\beta^4)$ is the vector of peer effects at each quantile. DGP F follows the standard LiM model with only spillover effects, where $\beta^\textup{LiM} = 0.55$. All estimations account for unobserved network heterogeneity using fixed effects.
\end{tablenotes}
\end{threeparttable}
\end{table}

\clearpage
\newpage 
\section{Additional Results on the Empirical Application}\label{append:addhealth}
In this section, we present supplemental results on the empirical application. Table \ref{tab:size} presents the sample size for each outcome. The difference in sample sizes is because we remove observations with missing values for each outcome separately. Therefore, different observations are removed for each outcome.

\begin{table}[!h]
\centering
\caption{Sample Size}
\label{tab:size}
\begin{tabular}{lc}
\toprule
Outcome                     & $n$  \\
\midrule
Academic achievements (GPA) & 70{,}022  \\
Academic effort             & 76{,}261  \\
Extracurricular activities  & 79{,}694  \\
Future perception           & 75{,}822  \\
Trouble at school           & 76{,}206  \\
Smoking                     & 74{,}904  \\
Drinking                    & 74{,}741  \\
Risky behaviors             & 75{,}545  \\
Self-esteem                 & 71{,}649  \\
Physical exercise           & 71{,}683  \\
Fighting                    & 71{,}606 \\\bottomrule
\end{tabular}
\end{table}

Table \ref{tab:append:addhealth} presents various robustness checks. Panel A presents the results of the three-quantile model specification as in the main text, with instruments generated using the cross-fitting approach. To construct these instruments, we predict the quantiles of the outcome among friends for ten levels evenly spaced from 0 to 1. For each level, these predictions are obtained using random forests where the target variable is a quantile of the outcome among friends and the predictors are $\mathbf X_g$, $\bar{\mathbf X}_g$, and the minimum, maximum, and quartiles of $\mathbf X_g$ and $\bar{\mathbf X}_g$ among friends. To avoid overfitting, we randomly split the networks into five folds. For each network from a given fold, the random forest model is trained using networks from the other folds. 

Panel B presents the results of the three-quantile model specification, but after excluding students without friends and using the standard instruments as in the main text. Panel C presents the four-quantile specification with the regular instruments from the main text. 

\begin{table}[htbp]
\centering
\footnotesize
\caption{Empirical Results Using the Cross-Fitting Approach, without Isolated Students (IS), and Using a Four-Quantile Specification}
\label{tab:append:addhealth}
\begin{threeparttable}
\begin{tabular}{cccd{1}cccd{1}cccc}
\toprule
\multicolumn{3}{c}{Panel A: Cross-Fitting} && \multicolumn{3}{c}{Panel B: Excluding IS} && \multicolumn{4}{c}{Panel C: Four quantiles}\\
\cmidrule{1-3} \cmidrule{5-7} \cmidrule{9-12} \addlinespace[0.5ex]
\multicolumn{1}{c}{$\beta^1$} & \multicolumn{1}{c}{$\beta^2$} & \multicolumn{1}{c}{$\beta^3$} && 
\multicolumn{1}{c}{$\beta^1$} & \multicolumn{1}{c}{$\beta^2$} & \multicolumn{1}{c}{$\beta^3$} && 
\multicolumn{1}{c}{$\beta^1$} & \multicolumn{1}{c}{$\beta^2$} & \multicolumn{1}{c}{$\beta^3$} & \multicolumn{1}{c}{$\beta^4$}\\
\midrule
\multicolumn{12}{c}{Academic   achievements (GPA)}                                                      \\[0.5ex]
-0.121  & 1.013   & -0.085  &  & 0.051   & 0.778   & 0.056   &  & 0.300   & -0.540  & 1.400   & -0.374  \\
(0.109) & (0.165) & (0.096) &  & (0.068) & (0.101) & (0.060) &  & (0.120) & (0.265) & (0.259) & (0.113) \\[1.5ex]
\multicolumn{12}{c}{Academic effort}                                                                    \\[0.5ex]
0.262   & 0.211   & 0.207   &  & 0.179   & 0.375   & 0.103   &  & 0.200   & 0.002   & 0.432   & -0.024  \\
(0.052) & (0.129) & (0.067) &  & (0.053) & (0.114) & (0.081) &  & (0.070) & (0.155) & (0.121) & (0.095) \\[1.5ex]
\multicolumn{12}{c}{Extracurricular   activities}                                                       \\[0.5ex]
0.011   & 0.724   & 0.064   &  & 0.051   & 0.711   & 0.062   &  & -0.176  & 0.815   & 0.050   & 0.075   \\
(0.088) & (0.100) & (0.048) &  & (0.062) & (0.072) & (0.033) &  & (0.160) & (0.275) & (0.186) & (0.040) \\[1.5ex]
\multicolumn{12}{c}{Future perception}                                                                  \\[0.5ex]
0.191   & 0.285   & 0.137   &  & 0.160   & 0.433   & 0.092   &  & 0.160   & 0.246   & 0.048   & 0.174   \\
(0.103) & (0.236) & (0.128) &  & (0.056) & (0.127) & (0.073) &  & (0.071) & (0.271) & (0.395) & (0.202) \\[1.5ex]
\multicolumn{12}{c}{Trouble at school}                                                                  \\[0.5ex]
0.093   & 0.630   & 0.107   &  & 0.063   & 0.575   & 0.092   &  & 0.423   & -0.483  & 0.816   & -0.076  \\
(0.079) & (0.094) & (0.068) &  & (0.080) & (0.099) & (0.064) &  & (0.130) & (0.251) & (0.159) & (0.081) \\[1.5ex]
\multicolumn{12}{c}{Smoking}                                                                            \\[0.5ex]
-0.116  & 0.761   & 0.154   &  & -0.103  & 0.739   & 0.169   &  & 0.059   & 0.106   & 0.559   & 0.071   \\
(0.077) & (0.089) & (0.032) &  & (0.070) & (0.089) & (0.028) &  & (0.143) & (0.216) & (0.131) & (0.040) \\[1.5ex]
\multicolumn{12}{c}{Drinking}                                                                           \\[0.5ex]
0.109   & 0.394   & 0.129   &  & -0.000  & 0.466   & 0.124   &  & 0.210   & -0.065  & 0.387   & 0.086   \\
(0.094) & (0.134) & (0.021) &  & (0.075) & (0.090) & (0.019) &  & (0.225) & (0.341) & (0.192) & (0.033) \\[1.5ex]
\multicolumn{12}{c}{Risky behaviors}                                                                    \\[0.5ex]
-0.098  & 0.667   & 0.106   &  & -0.002  & 0.569   & 0.176   &  & -0.237  & 0.733   & -0.041  & 0.220   \\
(0.072) & (0.095) & (0.032) &  & (0.060) & (0.105) & (0.030) &  & (0.180) & (0.333) & (0.226) & (0.055) \\[1.5ex]
\multicolumn{12}{c}{Self-esteem}                                                                        \\[0.5ex]
0.382   & 0.011   & 0.308   &  & 0.135   & 0.562   & 0.020   &  & 0.224   & -0.170  & 0.754   & -0.211  \\
(0.303) & (0.628) & (0.340) &  & (0.068) & (0.154) & (0.085) &  & (0.096) & (0.224) & (0.253) & (0.152) \\[1.5ex]
\multicolumn{12}{c}{Physical exercise}                                                                  \\[0.5ex]
-0.027  & 0.659   & -0.070  &  & 0.178   & 0.389   & 0.086   &  & 0.099   & 0.306   & 0.124   & 0.083   \\
(0.065) & (0.088) & (0.057) &  & (0.058) & (0.080) & (0.053) &  & (0.120) & (0.223) & (0.200) & (0.071) \\[1.5ex]
\multicolumn{12}{c}{Fighting}                                                                           \\[0.5ex]
0.226   & 0.168   & 0.282   &  & 0.225   & 0.241   & 0.269   &  & 0.390   & -0.224  & 0.295   & 0.208   \\
(0.080) & (0.124) & (0.046) &  & (0.076) & (0.116) & (0.039) &  & (0.133) & (0.237) & (0.162) & (0.057) \\\bottomrule
\end{tabular}
\end{threeparttable}
\end{table}

\newpage

\section{Constructing Sample Quantiles}\label{app:samp_quant}

When $d_i = 0$, we set $\tilde{y}_i^k = 0$ for $k=1, \dots, M$. This decision is innocuous if the model includes separate intercepts for isolated and non-isolated individuals. When $d_i=M > 0$, we simply let $\tilde{y}_{i}^k = \tilde{y}_{i,k}$ for $k=1,\dots,M$. 

When $d_i > 0$ and $d_i \neq M$, we compute sample quantiles using linear interpolation. For a given quantile index $k$ and sample size $d_i$, we first compute the (possibly non-integer) position $h(k, d_i) = \frac{k-1}{M-1}(d_i - 1) + 1$, which maps the quantile level $(k-1)/(M-1)$ to a location within the ordered peer outcomes $\{\tilde{y}_{i,k}\}_{k=1}^{d_i}$. When $h(k, d_i)$ is not an integer, the $k$-th sample quantile $\tilde{y}_i^k$ is obtained by linearly interpolating between the two peer outcomes $\lfloor h(k,d_i) \rfloor$ and $\lfloor h(k,d_i) \rfloor + 1$, weighted by their respective distances from $h(k, d_i)$: \begin{align*}
\tilde{y}_{i}^k = \begin{cases}
\tilde{y}_{i, \lfloor h(k,d_i) \rfloor} & \text{ if } h(k,d_i) = \lfloor h(k,d_i) \rfloor \\
\bigl(h(k,d_i) - \lfloor h(k,d_i)\rfloor\bigr)\,
\tilde{y}_{i,\lfloor h(k,d_i)\rfloor + 1} 
\\
\qquad + \bigl(\lfloor h(k,d_i)\rfloor + 1 - h(k,d_i)\bigr)\,
\tilde{y}_{i,\lfloor h(k,d_i)\rfloor} & \text{ if } h(k,d_i) \neq \lfloor h(k,d_i) \rfloor \\
\end{cases}
\end{align*}
This definition of the sample quantile is popular and is known as the Type 7 sample quantile \citep[see][]{hyndman1996sample}.

\end{document}